\documentclass[runningheads]{llncs}

\usepackage[utf8]{inputenc}
\usepackage[T1]{fontenc}
\usepackage{graphicx}

\usepackage{thmtools, thm-restate}

\usepackage{amsmath}
\usepackage{amssymb}
\usepackage{engord}
\usepackage{setspace}
\usepackage{hyperref}
\usepackage{cleveref}
\usepackage{multirow}

\usepackage{url}

\usepackage{tikz}
\usetikzlibrary{shapes,arrows,backgrounds,positioning,patterns,fit}
\usetikzlibrary{matrix,calc,automata,graphs,chains,trees}
\usetikzlibrary{tikzmark}
\tikzset{
->, 
>=stealth', 
node distance=2cm, 
initial text=$ $, 
}

\usepackage{pgfplots}
\pgfplotsset{compat=newest}

\usepackage[ruled,vlined]{algorithm2e}
\usepackage{booktabs}   
\usepackage{subcaption} 
\usepackage{sidecap} 
\usepackage{array}
\usepackage{mdframed} 
\usepackage{stmaryrd} 
\usepackage[inline]{enumitem} 
\usepackage{transparent}
\usepackage{wrapfig, framed}
\usepackage{multirow}
\usepackage{graphicx}
\definecolor{light-gray}{gray}{0.5}
\definecolor{lacecolor}{rgb}{1,0,0}
\usepackage{marginnote} 

\usepackage{colortbl} 

\usepackage{mathtools}
\usepackage{xspace}
\usepackage{relsize}

\begin{document}
\title{
Verifying Tree-Manipulating Programs via CHCs
}
%
%
 \author{Marco Faella\inst{1}\orcidID{0000-0001-7617-5489}
 \and
 Gennaro Parlato\inst{2}\orcidID{0000-0002-8697-2980} 
}
 \authorrunning{M. Faella and G. Parlato}
\institute{University of Naples Federico II, Italy\\ 
\email{m.faella@unina.it}\and
University of Molise, Italy\\
\email{gennaro.parlato@unimol.it}
}

%
\maketitle              

\newcommand\goesto{$\!\rightarrow$}

\newtheorem{problemdef}{Problem}

\newcommand{\FunUp}{\mathit{Up}}
\newcommand{\FunDown}{\mathit{Down}}
\newcommand{\FunInt}{\mathit{Internal}}

\newcommand{\MemSafe}
{\sc{MemSafety}}

\newcommand{\PtrFieldIsNil}{\mathit{field\_is\_nil}}
\newcommand{\FinalVal}{\mathit{final\_val}}

\newcommand{\FrameExit}
{\mathit{frame\_exit}}

\newcommand{\LabelExit}
{\mathit{label\_exit}}

\newcommand{\StepFrame}
{\mathit{step}}

\newcommand{\FirstFrame}
{\mathit{first\_frame}}

\newcommand{\ConsistentFirstFrame}
{\mathit{consistent\_first\_frames}}

\newcommand{\HeapCond}
{\mathit{heap\_cond}}

\newcommand{\FindPtr}
{\mathit{find\_ptr}}
\newcommand{\PointsHere}{\mathit{points\!\_\!here}}
\newcommand{\PtrHere}{\mathit{ptr\!\_\!here}}

\newcommand{\LastAssignmentToField}
{\mathit{last\_assignment\_to\_field}}

\newcommand{\Active}{\mathit{active}}

\newcommand{\PHatHere}
{\mathit{p\_hat\_here}}

\newcommand{\EndOfLace}
{\mathit{end\_of\_lace}}

\newcommand{\NoAssignmenToField}
{\mathit{no\_assignment\_to\_field}}

\newcommand{\Successor}{\mathit{lace\!\_\!succ}}
\newcommand{\SuccessorTree}{\mathit{lace\!\_\!succ}^{+}}
\newcommand{\SuccessorPost}{\mathit{lace\!\_\!succ}^{++}}

\newcommand{\State}{\mathbf{S}}

\newcommand{\state}{\mathit{state}}
\newcommand{\stateptr}{\mathit{state\!\_\!ptr}}
\newcommand{\statechld}{\mathit{state\!\_\!chld}}
\newcommand{\nothing}{\emptyset}

\newcommand{\EndOfLAce}
{\mathit{label\_exit}_\ExitStatus}

\newcommand{\lookforroot}
{\mathbf{LookingForRoot}}

\newcommand{\Apre}{\mathcal{A}_\mathrm{pre}}
\newcommand{\Apost}{\mathcal{A}_\mathrm{post}}
\newcommand{\Anotpost}{\mathcal{A}_{\neg\mathrm{post}}}

\newcommand{\psiPre}
{\psi_{\mathit{pre}}}

\newcommand{\chc}{\mathit{chc}}

\newcommand{\startframe}{\mathit{start}}

\newcommand{\ConsistentChild}{\mathit{consistent\!\_\!child}}
\newcommand{\ConsistentActive}{\mathit{consistent\_active\_input\_tree}}
\newcommand{\ConsistentUpd}{\mathit{consistent\_updNON\_ESISTE}}
\newcommand{\ConsistentChildTree}{\mathit{consistent\!\_\!child}^{+}}
\newcommand{\ConsistentChildPost}{\mathit{consistent\!\_\!child}^{++}}

\newcommand{\Step}{\mathit{step}}
\newcommand{\Continues}{\mathit{continues}}

\newcommand{\SAT}{\textsc{Sat}}
\newcommand{\Yes}{\sc Yes}
\newcommand{\No}{\sc No}

\newcommand{\prefix}{\mathit{prefix}}
\newcommand{\ExitStatus}{\mathit{Ex}}

\newcommand{\invCex}{\mathbf{Lab}}
\newcommand{\prob}{\mathcal{I}}
\newcommand{\ExitStatusP}[1]{\textsc{ExitStatus} \big(P,m,n,#1\big)}
\newcommand{\CexParam}[1]{\mathcal{C}_{\mathrm{ex}}\big(P,k,m,n,#1\big)}
\newcommand{\Cex}{\mathcal{C}_{\mathrm{ex}}(\mathcal{I})}
\newcommand{\Ckt}{\mathcal{C}_{\mathrm{kt}}\big(P,m,n\big)}
\newcommand{\leaf}{\mathit{leaf}}
\newcommand{\invCexPre}{\mathbf{Pre}}
\newcommand{\invPrePost}{\mathbf{PP}}
\newcommand{\invCexPost}{\mathit{Post}}
\newcommand{\CexPreParam}[2]{
\mathcal{C}_{\mathrm{pre}}\big(P,m,n,#1,#2\big)}
\newcommand{\CexPre}{\mathcal{C}_{\mathrm{pre}}(\mathcal{I},\Apre)}
\newcommand{\tainting}{\mathbf{T}}

\newcommand{\accept}{\mathit{accept}}
\newcommand{\extra}{\bar{\alpha}}
\newcommand{\invParam}{\mathcal{P}}
\newcommand{\CexTree}{\mathcal{C}_{\tainting}\big(P,m,n\big)}
\newcommand{\CexState}[1]{\mathcal{C}_{\State}\big(P,m,n,#1\big)}
\newcommand{\CexPrePost}[1]{\mathcal{C}_{\invPrePost}\big(P,m,n,#1\big)}
\newcommand{\wild}{\cdot}

\newcommand{\sidecomment}[1]{
\checkoddpage
\ifoddpage
{\marginnote{\makebox[3cm][l]{\hspace{0.8cm}\footnotesize\textcolor{gray}{#1}}}[-3ex]}
\else
{\reversemarginpar\marginnote{\makebox[3cm][l]{\hspace{0.8cm}\footnotesize\textcolor{gray}{#1}}}[-3ex]}
\fi}


\renewcommand{\epsilon}{\varepsilon}
\newcommand{\pre}{\mathit{pre}}
\newcommand{\out}{\mathit{heap}}
\newcommand{\inp}{\mathit{in}}
\newcommand{\post}{\mathit{post}}

\newcommand{\lacetree}{knitted-tree\xspace}
\newcommand{\Lacetree}{Knitted-tree\xspace}
\newcommand{\LaceTree}{Knitted-Tree\xspace}

\newcommand{\TWA}{TWA\xspace}

\renewcommand{\coloneq}{\coloneqq}
\newcommand{\child}{\mathit{child}}

\newcommand{\Type}{\mathit{type}}

\newcommand{\symb}[1]{\langle{#1}\rangle}
\newcommand{\fwd}{\boldsymbol{-}}

\newcommand{\dtree}{\mathcal{T}}
\newcommand{\ktree}{\mathcal{K}}

\newcommand{\prog}{P}

\newcommand{\nats}{\mathbb{N}}
\newcommand{\lace}{\mathit{kt}}

\newcommand{\pointername}{p}

\newcommand{\datanames}{\mathit{D}}
\newcommand{\dataname}{d}
\newcommand{\datatype}{\mathcal{D}}

\newcommand{\point}{\mathbf{point}} 
\newcommand{\data}{\mathbf{data}} 
\newcommand{\datafield}{\mathit{val}} 

\newcommand{\ptrid}{\mathit{p\_id}}

\newcommand{\activechild}{\mathit{active\!\_\!child}}

\newcommand{\pfield}{\mathit{pfield}}
\newcommand{\findptr}{\mathit{find\!\_ptr}}
\newcommand{\ptrage}{\mathit{ptr\_age}}
\newcommand{\ptrval}{\mathit{ptr\_val}}
\newcommand{\dir}{\mathit{dir}}

\newcommand{\myswarrow}{0}
\newcommand{\mysearrow}{1}
\newcommand{\ext}{\mathit{Ext}}
\newcommand{\dsig}{\mathcal{S}}
\newcommand{\Ksig}{\dsig_\ktree}
\newcommand{\inpsig}{\dsig^\Sigma}

\newcommand{\reals}{\mathbb{R}}
\newcommand{\ints}{\mathbb{Z}}
\newcommand{\bool}{\mathbb{B}}

\newcommand{\RBT}{{\sc Rbt}\xspace}

\newcommand{\BST}{{\sc Bst}\xspace}
\newcommand{\AVL}{{\sc AVL}\xspace}

\newcommand{\CHC}{\mbox{CHC}\xspace}
\newcommand{\MSOD}{\mbox{\sc Mso-D}\xspace}
\newcommand{\MSO}{\mbox{\sc Mso}\xspace}
\newcommand{\FOL}{\mbox{\sc Fol}\xspace}

\newcommand*{\defeq}{\stackrel{\mathsmaller{\mathsf{def}}}{=}}

\newcommand{\theorypreds}{\mathcal{D}^\mathrm{rel}}
\newcommand{\prop}{p}
\newcommand{\rel}{r}
\newcommand{\fun}{f}

\newcommand{\tuple}[1]{\langle #1 \rangle}
\newcommand{\LeftChild}{\mathit{left}}
\newcommand{\RightChild}{\mathit{right}}
\newcommand{\Child}{\mathit{child}}
\newcommand{\Root}{\mathit{root}}
\newcommand{\Leaf}{\mathit{leaf}}
\newcommand{\qdot}{\,.\,}

\newcommand{\SDTA}{{\sc Sdta}\xspace}
\newcommand{\SDTWA}{{\sc Sdtwa}\xspace}

\newcommand{\run}{\rho}

\newcommand{\DTL}{{\sc Dtl}\xspace}

\newcommand{\TWDTT}{{\sc Dtt}\xspace}

\newcommand{\aut}{\mathcal{A}}
\newcommand{\unreach}{\bot}

\newcommand{\taint}{\mathit{taint}}
\newcommand{\taintptr}{\mathit{taint\!\_\!ptr}}
\newcommand{\Taintnode}{\mathit{Tainted\!\_\!Node}}
\newcommand{\Taintpointer}{\mathit{Tainted\!\_\!Ptr}}

\newcommand{\taintnode}{\mathit{tainted\_node}}
\newcommand{\taintpointer}{\mathit{tainted\_ptr}}
\newcommand{\rootptr}{\widehat{p}}

\newcommand{\MSODEF}{\mbox{\sc Mso-D$_{\exists,\forall}$}\xspace}

\newcommand{\dataexp}{\mathit{exp}}
\newcommand{\stmt}{\mathit{stmt}} 
\newcommand{\pcstmt}{\mathit{pc\_stmt}} 
\newcommand{\ctrlstmt}{\mathit{ctrl\_stmt}} 
\newcommand{\recstmt}{\mathit{rec\_stmt}} 
\newcommand{\heapstmt}{\mathit{heap\_stmt}} 
\newcommand{\field}[2]{{\mathit{#1}}\!\rightarrow\!{\mathit{#2}}}

\newcommand{\New}{\mathbf{new}}
\newcommand{\Free}{\mathbf{free}}
\newcommand{\Skip}{\mathbf{skip}}
\newcommand{\WHILE}{\mathbf{while}}
\newcommand{\GOTO}{\mathbf{goto}}

\newcommand{\PC}{\mathit{PC}_P} 
\newcommand{\PV}{\mathit{PV}_{\!P}}
\newcommand{\DV}{\mathit{DV}_{\!P}} 

\newcommand{\Int}{\mathbf{int}} 
\newcommand{\Bool}{\mathbf{bool}} 

\newcommand{\pointer}{\mathbf{pointer}} 

\newcommand{\pointers}{\mathit{PF}} 
\newcommand{\nextpc}{\mathit{succ}} 

\newcommand{\NIL}{\mathbf{nil}} 
\newcommand{\NOP}{\textsc{nop}}
\newcommand{\Error}{\textsc{err}}
\newcommand{\OOM}{\textsc{oom}}
\newcommand{\RWD}{\textsc{rwd}}

\newcommand{\exit}{\mathbf{exit}} 
\newcommand{\type}{\mathit{type}} 

\newcommand{\true}{\ensuremath{\mathit{true}}}
\newcommand{\false}{\ensuremath{\mathit{false}}}

\newcommand{\exitst}{\mathit{exit}}

\newcommand{\activefield}{\mathit{active}}
\newcommand{\avail}{\mathit{avail}}
\newcommand{\instr}{\mathit{event}}
\newcommand{\Instr}{\mathit{Event}}
\newcommand{\prevfr}{\mathit{prev}}
\newcommand{\nextfr}{\mathit{next}}
\newcommand{\Dir}{\mathit{Dir}}
\newcommand{\here}{\mathbf{here}}
\newcommand{\pc}{\mathit{pc}} 
\newcommand{\edge}{\mathit{edge}} 
\newcommand{\chg}{\mathit{upd}}
\newcommand{\isnil}{\mathit{isnil}}

\newcommand{\below}{\mathrm{below}}
\newcommand{\prev}{\mathrm{prev}}

\newcommand{\EStatus}{\mathit{ExStatus}}
\newcommand{\estatus}{\mathit{ex}}
\newcommand{\stclean}{\mathit{C}}
\newcommand{\stoverflow}{\mathit{O}}
\newcommand{\stpartial}{\mathit{P}}
\newcommand{\sterror}{\mathit{E}}

\newcommand{\len}{\mathit{len}}

\newcommand{\KT}{\mathit{kt}}

\begin{abstract}
Programs that manipulate tree-shaped data structures often require complex, specialized proofs that are difficult to generalize and automate. This paper introduces a unified, foundational approach to verifying such programs. 
Central to our approach is the {\em \lacetree encoding}, modeling each program execution as a tree structure capturing input, output, and intermediate states. Leveraging the compositional nature of 
 {\lacetree}s, we encode these structures as constrained Horn clauses (\CHC{s}), reducing verification to \CHC\ satisfiability task. 
To illustrate our approach, we focus on {\em memory safety} and show how it naturally leads to simple, modular invariants.
%

\end{abstract}

\section{Introduction}\label{sec:intro}

Ensuring automatic or semi-automatic verification of programs that manipulate dynamic memory (heaps) presents numerous challenges. First, heaps can grow unboundedly in size and store unbounded data, requiring expressive logics to capture intricate invariants, preconditions, and postconditions. Next, dynamic shape changes complicate the maintenance of structural constraints, such as 
binary search tree 
ordering or balance. Aliasing further obscures these constraints and breaks invariants. Unbounded recursion and loops, common in tree algorithms, add complexity by making termination reasoning non-trivial. Finally, incomplete or missing specifications often force verifiers to infer properties on-the-fly, and integrating various verification techniques (e.g., SMT-solving, abstract interpretation, and interactive theorem proving) remains a nontrivial task.
 
This paper presents a foundational approach for automated analysis of heap-manipulating programs, particularly those involving tree data structures. By combining automata and logic-based methods, we reduce 
verification to checking satisfiability of constrained Horn clauses ({\CHC}s). This reduction allows us to capitalize on advancements in {\CHC} solvers~\cite{DBLP:journals/corr/abs-2404-14923,De_Angelis_2022}. 
We demonstrate our approach on the \emph{memory safety problem}, ensuring that no execution 
causes crashes (e.g., null-pointer dereferences, use-after-free, or illegal frees) or nontermination. 

The core of our methodology maps an entire program execution $\pi$ on an input data tree $T$ into a single tree data structure called a {\bf \lacetree}. 
This structure encapsulates the input, output, and all intermediate configurations of $\pi$. Its underlying tree, or {\em backbone}, is derived from~$T$ by adding a fixed number of 
inactive nodes to allow dynamic node allocation. 
Each node is labeled by a sequence of records, or \emph{frames}, connected in a global linear sequence called the \emph{lace},
where consecutive frames may belong to the same node or adjacent nodes,
resembling a knitting tree.
Each frame describes changes to the associated node (e.g., pointer updates) and records the current program state, preserving the backbone's original structure while also representing the final heap that may differ graph-wise. However, the {\lacetree}'s parameters -- the number of extra nodes added to the backbone and the number of frames per node -- may not capture every possible execution, potentially excluding some 
from our analysis.

\begin{wrapfigure}[18]{L}{4.9cm}
\vspace{-1.2cm}
\scriptsize
\centering
\begin{align*}
 & \pointer \, \mathit{head},\mathit{prev},\mathit{cur},\mathit{tmp}\\
 & \Int \, \mathit{key}\\
\mathit{0:}\ \ \  &\mathit{cur} \;\mathbf{\coloneq}\; \mathit{head}\,\mathbf{;}\\
\mathit{1:}\ \ \ & \WHILE\;(\mathit{cur} \,\mathbf{\neq}\,\NIL\, \&\&\,\field{cur}{\datafield} \,\mathbf{\not=}\,\mathit{key})\;\mathbf{do}\\
\mathit{2:}\ \ \ & \ \ \ \ \ \mathit{tmp} \;\mathbf{\coloneq}\; \field{\mathit{cur}}{\mathit{next}}\,\mathbf{;}\\
\mathit{3:}\ \ \ & \ \ \ \ \ \field{\mathit{cur}}{\mathit{next}} \;\mathbf{\coloneq}\; \mathit{prev}\,\mathbf{;}\\
\mathit{4:}\ \ \ & \ \ \ \ \ \mathit{prev} \;\mathbf{\coloneq}\; \mathit{cur}\,\mathbf{;}\\
\mathit{5:}\ \ \ & \ \ \ \ \ \mathit{cur} \;\mathbf{\coloneq}\; \mathit{tmp}\,\mathbf{;}\\
& \mathbf{od}\,\mathbf{;}\\
\mathit{6:}\ \ \ & \mathbf{if}\;(\mathit{cur} \,\mathbf{\not=}\,\NIL)\;\mathbf{then}\,\,\,\text{\textcolor{gray}{We found the key}}\\
\mathit{7:}\ \ \ &  \ \ \ \ \ 
\mathit{tmp} \;\mathbf{\coloneq}\; \field{cur}{\mathit{next}}\,\mathbf{;}\\
\mathit{8:}\ \ \ &  \ \ \ \ \ 
\field{head}{\mathit{next}} \;\mathbf{\coloneq}\; \mathit{tmp} \,\mathbf{;}\, \text{\textcolor{gray}{Rewind for $\mathit{head}$}}\\
\mathit{9:}\ \ \ &  \ \ \ \ \ 
\mathit{head} \;\mathbf{\coloneq}\; \mathit{cur} \,\mathbf{;} \quad \text{\textcolor{gray}{Rewind for $\mathit{cur}$}} \\
&\mathbf{else}\\
\mathit{10:}\ \ \ &  \ \ \ \ \ 
\mathit{head} \;\mathbf{\coloneq}\; \mathit{prev} \,\mathbf{;}\\
&\mathbf{fi}\,\mathbf{;}\\
\mathit{11:}\ \ \ &\exit\,\mathbf{;}
\end{align*}
\vspace{-0.6cm}
\caption{
Running example.
}
\label{fig:running example}
\end{wrapfigure}

\paragraph{Example.}
We illustrate our encoding method with a simple program shown in Fig.~\ref{fig:running example} that manipulates a singly linked list, specifically designed to highlight the key features of our encoding methodology.
The program takes a list of integers with $\mathit{head}$ pointing to the first node and a value stored in $\mathit{key}$. It reverses the list up to and including the first node containing the key, then appends the remaining nodes.    
For example, given the input list 1 \goesto 2 \goesto 3 \goesto 4 \goesto 5 and $\mathit{key} =$~3, 
the output list is: 3 \goesto 2 \goesto 1 \goesto 4 \goesto 5.
The \lacetree corresponding to this program execution is shown in~Fig.~\ref{fig:kt1}. 

The \lacetree's structure matches the input list.
The label of each node is displayed next to it, with the lace being
depicted by red arrows and numbers.
In our example, the lace starts at the frame with ordinal 1 of node $u_1$, 
takes two local steps to the frames with ordinals 2 and 3, then moves to the frame numbered 4 in node $u_2$, and so on.
Note that consecutive frames in the lace either belong to the same node,
or to adjacent nodes.
Due to space constraints, only a selection of the information contained within each frame is displayed.

\noindent Our encoding's main innovation is how it handles pointer fields and variables:
\begin{enumerate}
    \item an update to a pointer field is stored in its node;
    \item an update to a pointer variable is stored in the node it points to. 
\end{enumerate}
For example, the lace's first frame includes the event
$\symb{\mathit{head} \coloneq \here}$, indicating that 
$\mathit{head}$ initially points to the first node of the input structure.
This initial assignment is implicit.
The second frame corresponds to the execution of 
$\mathit{cur} \coloneq \mathit{head}$ at line 0.
%
\begin{figure}[t]
\begin{tikzpicture}[remember picture,node distance=0.6cm]
\newcommand{\id}[1]{\multicolumn{1}{c}{\tiny\textcolor{blue}{#1}}}
\newcommand{\headp}{\mathit{head}}
\newcommand{\tmp}{\mathit{tmp}}
\newcommand{\prevp}{\mathit{prev}}
\newcommand{\curp}{\mathit{cur}}
\newcommand{\nextf}{\mathit{next}}
\newcommand{\mycoloneq}{\!\coloneq\!}
\renewcommand{\here}{\mathbf{\scriptscriptstyle here}}
\newcommand{\printcounter}[1]{%
    \multicolumn{1}{c}{\tikzmarknode[lacelabel]{f#1}{#1}}
}
\tikzset{ball/.style={circle, thick, draw, minimum size=20pt}}
\tikzset{log/.style={font=\scriptsize,inner sep=1pt}}
\tikzset{lace/.style={->,draw=lacecolor}}
\tikzset{lacelabel/.style={text=lacecolor, inner sep=2pt}}
\tikzset{internal/.style={}}
\tikzset{across/.style={}}
\setlength\arrayrulewidth{1pt}\arrayrulecolor{blue}
\node[ball] (u1) at (0,0) {$u_1$};
\node[ball] (u2) at ([yshift=-2.1cm]u1) {$u_2$};
\node[ball] (u3) at ([yshift=-2.1cm]u2) {$u_3$};
\node[ball] (u4) at ([yshift=-1.8cm]u3) {$u_4$};
\node[ball] (u5) at ([yshift=-1.2cm]u4) {$u_5$};
\draw[thick] (u1) -- (u2);
\draw[thick] (u2) -- (u3);
\draw[thick] (u3) -- (u4);
\draw[thick] (u4) -- (u5);
\node[log, right=of u1] (U1) {
\begin{tabular}{|c|c|c|c|c|c|c|c|c|}
\multicolumn{1}{c}{} &\printcounter{1} &\printcounter{2} &\printcounter{3}       &\printcounter{5} &\printcounter{6} &\printcounter{18}      \\ \hline
$\datafield\!:1$  &$\pc\!:0$                      &$\pc\!:1$                     &$\pc\!:2$ &$\pc\!:4$                     &$\pc\!:5$            &$\pc\!:9$ \\ 
          &$\mathit{key}\!:3$                      &$\symb{\curp\mycoloneq\here}$ &        &$\symb{\nextf\mycoloneq\NIL}$ &$\symb{\prevp\mycoloneq\here}$ &$\symb{\nextf\mycoloneq\tmp}$ \\
          &$\symb{\headp\mycoloneq\here}$ &                            &        &$\chg_{\tmp}$               &                   &$\chg_{\tmp,\curp,\prevp}$ \\
\hline
\id{1}    &\id{2}                       &\id{3}                      &\id{4}  &\id{5}                      &\id{6}             &\id{7} 
\end{tabular}
};
\node[log, right=of u2] (U2) {
\begin{tabular}{|c|c|c|c|c|c|c|c|}
\multicolumn{1}{c}{}         &\printcounter{4}                          &\printcounter{7}                          &\printcounter{8}       &\printcounter{10}                            &\printcounter{11}                           &\printcounter{17}        &\printcounter{19} \\ \hline
$\datafield\!:2$ &$\pc\!:3$                    &$\pc\!:1$                    &$\pc\!:2$ &$\pc\!:4$                       &$\pc\!:5$                      &$\pc\!:8$   &$\pc\!:9$ \\
         &$\symb{\tmp\mycoloneq\here}$ &$\symb{\curp\mycoloneq\here}$ &       &$\symb{\nextf\mycoloneq\prevp}$ &$\symb{\prevp\mycoloneq\here}$ &$\RWD_2$ &$\RWD_{7}$ \\ 
         &         &$\chg_{\prevp}$        &  &$\chg_{\tmp}$ & &$\chg_{\tmp,\curp}$ & \\ \hline
\id{1}   &\id{2}                     &\id{3}                     &\id{4}  &\id{5}                        &\id{6}                      &\id{7}    &\id{8}
\end{tabular}
};
\node[log, right=of u3] (U3) {
\begin{tabular}{|c|c|c|c|c|c|c|c|c|c|}
\multicolumn{1}{c}{}         &\printcounter{9}                          &\printcounter{12}                          &\printcounter{13}      &\printcounter{14}      &\printcounter{16}            &\printcounter{20}       \\ \hline
$\datafield\!:3$ &$\pc\!:3$                    &$\pc\!:1$                     &$\pc\!:6$ &$\pc\!:7$ &$\pc\!:8$       &$\pc\!:11$ \\ 
         &$\symb{\tmp\mycoloneq\here}$ &$\symb{\curp\mycoloneq\here}$ &        &        &$\RWD_2$      &$\symb{\headp\mycoloneq\here}$ \\ 
         &                           &$\chg_{\prevp}$             &        &        &$\chg_{\tmp}$ &         \\ \hline
\id{1}   &\id{2}                     &\id{3}                      &\id{4}  &\id{5}  &\id{6}        &\id{7}
\end{tabular}
};
\node[log, right=of u4] (U4) {
\begin{tabular}{|c|c|c|c|c|c|c|}
\multicolumn{1}{c}{}         &\printcounter{15}                         \\ \hline
$\datafield\!:4$ &$\pc\!:8$                    \\ 
         &$\symb{\tmp\mycoloneq\here}$ \\ \hline
\id{1}   &\id{2}                          
\end{tabular}
};
\node[log, right=of u5, yshift=-0.1cm] (U5) {
\begin{tabular}{|c|c|c|c|c|c|c|}
\hline
$\datafield\!:5$         \\ \hline
\id{1}
\end{tabular}
};
\pgfmathsetmacro{\botoffset}{-1.1}
\draw[lace,internal] (f1) to (f2);
\draw[lace,internal] (f2) to (f3);
\draw[lace,across] (f3)+(-3pt,\botoffset) to[controls={+(0,-0.5) and +(-0.1,1)}] (f4.west);
\draw[lace,across] (f4.east) to[controls={+(0.3,0.7) and +(-0.3,-0.6)}] ($(f5)+(-3pt,\botoffset)$);
\draw[lace,internal] (f5) to (f6);
\draw[lace,across] (f6)+(-3pt,\botoffset) to[controls={+(0,-0.8) and +(0.5,0.6)}] (f7.north east);
\draw[lace,internal] (f7) to (f8);
\draw[lace,across] (f8)+(-3pt,\botoffset) to[controls={+(0,-0.5) and +(-0.1,1)}] (f9.west);
\draw[lace,across] (f9.east) to[controls={+(0.3,0.7) and +(-0.3,-0.6)}] ($(f10)+(-3pt,\botoffset)$);
\draw[lace,internal] (f10) to (f11);
\draw[lace,across] (f11)+(-3pt,\botoffset) to[controls={+(0,-0.5) and +(0.5,0.6)}] (f12.north east);
\draw[lace,internal] (f12) to (f13);
\draw[lace,internal] (f13) to (f14);
\draw[lace,across] (f14)+(-3pt,\botoffset) to[controls={+(0,-0.7) and +(-0.1,1)}] (f15.west);
\draw[lace,across] (f15.east) to[controls={+(0.5,0.1) and +(-0.3,-0.6)}] ($(f16)+(-3pt,\botoffset)$);
\draw[lace,across] (f16) to[controls={+(0.3,0.5) and +(-0.3,-0.6)}] ($(f17)+(-3pt,\botoffset)$);
\draw[lace,across] (f17) to[out=90,in=-90] ($(f18)+(-0.2,\botoffset)$);
\draw[lace,across] (f18)+(0.2,\botoffset) to[out=-90,in=90] (f19);
\draw[lace,across] (f19)+(-3pt,\botoffset) to[controls={+(0,-0.5) and +(0.5,0.6)}] (f20);
\arrayrulecolor{black}
\end{tikzpicture}
\caption{A \lacetree of the program in Fig.~\ref{fig:running example}
on the input list 1 \goesto 2 \goesto 3 \goesto 4 \goesto 5 and $\mathit{key} = 3$.
Blue numbers below the frames represent positions within the label,
while red numbers and arrows refer to the lace.} \label{fig:kt1}
\end{figure}
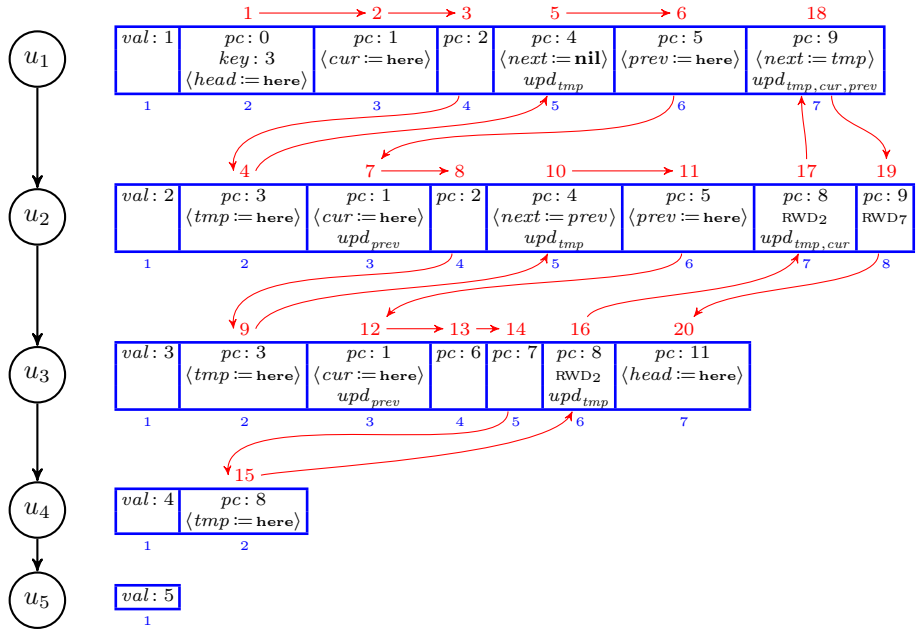
According to rule~2, when a pointer is dereferenced, it may be necessary to traverse the lace backward to find its latest assignment, a process called \emph{rewinding}.
In our example, the first rewinding occurs at line 8 when the current value of $\mathit{head}$ is needed.
In the \lacetree, frame 15 in node $u_4$ reaches line 8, but since the label of $u_4$ does not contain information about $\mathit{head}$, rewinding is triggered. Frames 16 and 17 are then added to the lace to go back to node $u_1$,
where $\mathit{head}$ currently points, and  frame 18 in $u_1$ reports the effect of the instruction
at line 8, consisting in the event 
$\symb{\mathit{next} \coloneq \mathit{tmp}}$.


\smallskip
{\Lacetree}s enjoy {\bf compositional properties} that are essential for our {\CHC}-based verification method. These properties allow subtree replacement: a subtree rooted at node~$v$ in one {\lacetree} can be swapped with a subtree rooted at node~$v'$ in another,
provided that the labels of $v$ and $v'$ satisfy a local consistency condition expressible in a quantifier-free first-order data theory. This replacement rule enables us to build a {\CHC} system whose minimal model precisely captures the set of valid node labels for the {\lacetree}s, allowing the detection of any reachable error configuration during execution.
If our analysis reports no errors, we need to check whether any execution
was truncated by the current choice of parameters.
With a small modification, the same {\CHC} system can also reveal whether 
that is the case;
if so, we can increase the parameters to encompass more executions. 
While verification is undecidable and may not terminate, our method significantly broadens the class of programs and properties amenable to automation. In particular, executions of several well-known programs are fully captured by {\lacetree}s with suitable parameters, as they traverse nodes a bounded number of times and admit simple {\CHC} solutions.

\smallskip

\noindent{\em Organization of the paper.} The rest of the paper is structured as follows.
\Cref{sec:preliminaries} introduces notation and definitions. 
\Cref{sec:lace} defines the \lacetree encoding, and \Cref{sec:compositionality}
details its compositionality. \Cref{sec:reasoning} describes our reduction to the {\CHC} satisfiability problem, presents a sound procedure to solve the memory safety problem,
and analyzes the structure and complexity of the required invariants. 
\Cref{sec:related_work} reviews related work, and
\Cref{sec:conclusions}
concludes with future directions. The appendices include substantial supplementary material.



\section{Preliminaries}\label{sec:preliminaries}



 
Let $\mathbb{N}$ denote the natural numbers that include $0$, $[i,j] \defeq \{k \in \mathbb{N} \mid i \leq k \leq j\}$, and $[j] \defeq [1,j]$.

\medskip
\noindent{\em Trees.} 
A $k$-ary tree $T$ is a finite, prefix-closed subset of $[k]^*$, where $k\in\mathbb{
N}$. Each element in $T$ is a {\em node}, with the {\em root} represented by the empty string $\epsilon$. The tree {\em edge relation} is implicit: for any $d \in [k]$, if both $\mathit{v}$ and $\mathit{v.d}$ are nodes in $T$, then $(\mathit{v}, \mathit{v.d})$ is an {\em edge}, making $\mathit{v.d}$ the $d$\textsuperscript{th} child of $\mathit{v}$, and $\mathit{v}$ the {\em parent} of $\mathit{v.d}$. 
 
\smallskip
\noindent{\em Data signatures.}
A \emph{data signature} $\dsig$ consists of pairs $\{ \mathit{id}_i : \mathit{type}_i \}_{i\in[n]}$, defining field names and their types (e.g., integers, Booleans $\bool$). An {\em evaluation} $\nu$ of $\dsig$ 
assigns each field name $\mathit{id}$ a type-specific value, denoted $\nu.\mathit{id}$. The {\em language} of $\dsig$, $L(\dsig)$, is the set of all its evaluations.

\smallskip
\noindent{\em Data trees.} A {\em data tree} with data signature $\dsig$, or $\dsig$-{\em tree}, is a pair $(T,\lambda)$ where $T$ is a tree and $\lambda$ is a labeling function $\lambda:T\rightarrow L(\dsig)$ that assigns an evaluation of $\dsig$ to each node $t \in T$. To simplify notation, the value of a field $\mathit{id}$ at node $t$ can be written as $t.\mathit{id}$ when $\lambda$ is clear from the context.
%


\smallskip
\noindent\emph{Constrained Horn clauses.} We use standard first-order logic (FOL) with equality~\cite{DBLP:conf/unu/MannaZ02} and formulas from a many-sorted, quantifier-free first-order theory $\mathcal{D}$ that includes program-relevant data types like arithmetic, reals, and arrays. 
We refer to $\mathcal{D}$ as the {\em data theory}.
\begin{definition}
Let $R$ be a set of uninterpreted fixed-arity relation symbols  representing 
unknowns. A {\bf constrained Horn clause} ({\CHC}) is a formula of the form $H\leftarrow C\wedge B_1\wedge\cdots\wedge B_n$  
where: {\bf(i)} $C$ is a {\em constraint}, 
    a 
    formula of the data theory $\mathcal{D}$  
    without symbols from $R$;
    {\bf(ii)} each $B_i$ is an application $r(v_1,\ldots,v_k)$ of a relation symbol $r\in R$ to first-order variables $v_1,\ldots,v_k$; 
    {\bf(iii)} $H$ (the  
    {\em head}) is either $\false$,
    or an application $r(v_1,\ldots,v_k)$ as in $B_i$. 
A {\CHC} is a {\em fact} if its body is only $C$  and a {\em query} if its head is $\false$. A finite set $\mathcal{C}$ of {\CHC}s forms a {\em system} by conjoining all 
{\CHC}s with free variables universally quantified. We assume the constraint semantics is predefined as a structure.  \qed
\end{definition}
A {\CHC} system $\mathcal{S}$ with relation symbols $R$ is \emph{satisfiable} if there exists an interpretation for each $r\in R$ that makes all clauses in $\mathcal{S}$ valid. Any such interpretation is called a \emph{solution} of $\mathcal{S}$. The {\CHC} {\em satisfiability problem} is the computational task of 
determining whether a given system $\mathcal{S}$ of {\CHC}s is satisfiable. 

Each {\CHC} system $\mathcal{S}$ has a unique minimal model under subset ordering,\footnote{ See~\cite{DBLP:journals/jacm/EmdenK76} for logic programs and \cite[Prop.~4.1]{DBLP:journals/jlp/JaffarM94} for constr. 
logic programs (or {\CHC}s).}
computable as the fixed-point of an operator derived from its clauses~\cite{DBLP:journals/jacm/EmdenK76,DBLP:journals/jlp/JaffarM94}. We use this fixed-point semantics to ensure the correctness of our reductions.

\subsubsection*{Heap-manipulating Programs. 
} 
\label{sec:programming language}



The {\em heap} is essential for dynamic memory allocation, allowing
memory blocks ({\em nodes}) to be allocated and deallocated during execution. 
We assume that nodes have a single data field and one or more pointer fields.
A {\em specific heap state} is 
defined as follows.
\begin{definition} A {\bf heap} is a tuple $\mathcal{H} = (N, \dsig, \data, \pointers)$ where 
\begin{itemize}

\item $N$ is a finite set of nodes, including a unique element $\NIL$ for free memory. 

\item $\dsig$ is a data signature defining the type of data that can be stored in a node. 

\item $\data: N\setminus\{\NIL\}\rightarrow L(\dsig)$ 
is a map modeling the data field of each node. 

\item $\pointers$ is a finite sequence of distinct pointer fields, each defined as a function of type $(N\setminus\{\NIL\})\rightarrow N$, representing the pointers of each node.
\qed
\end{itemize}
\end{definition}

We define $k$ as the number of pointer fields. For example,  
in Fig.~\ref{fig:running example}, $\pointers = \{\mathit{next} \}$ and $k=1$, while for 
binary trees use $\pointers = \{\mathit{left}, \mathit{right} \}$ and $k=2$.
\smallskip
\begin{wrapfigure}[15]{R}{8.5cm}
\centering
\vspace{-1.0cm}
{
  \scriptsize
\begin{mdframed}\raggedleft
\vspace{-0.4cm}
\begin{align*}
\text{\bf Program} \defeq \;\:& \mathit{decl} \:\: \mathit{block} \\ 
\mathit{decl}   \defeq \;\:&
(\pointer \, \mathit{id}(\mathbf{,}\mathit{id})^*)^* \: (\type \; \mathit{id}(\mathbf{,}\mathit{id})^*)^* \\
\mathit{block} \defeq \;\:& (\pc: (\ctrlstmt \mid \heapstmt) \,;)^+ \\
\ctrlstmt \defeq \;\:& d \,\mathbf{\coloneq}\, \dataexp\mid d_\mathrm{bool} \,\mathbf{\coloneq}\, \HeapCond\mid \Skip\mid \exit\\
& 
\mid \mathbf{if}\;\mathit{cond}\;\mathbf{then}\;
\mathit{block} \; \mathbf{else}\; \mathit{block}\;\mathbf{fi}\\
& 
\mid \WHILE\;\mathit{cond}\;\mathbf{do}\;\mathit{block}\;\mathbf{od} \mid \GOTO\;\pc \\
\heapstmt \defeq \;\:& \New\; p \mid \Free\; p 
\mid p \,\mathbf{\coloneq}\,\NIL
  \mid p \;\mathbf{\coloneq}\; q 
  \mid p \;\mathbf{\coloneq}\; \field{q}{\pfield} \\
& 
\mid \field{p}{\pfield} \;\mathbf{\coloneq}\; \NIL
  \mid \field{p}{\pfield} \;\mathbf{\coloneq}\; q \\
  & 
\mid \field{p}{dfield} \;\mathbf{\coloneq}\; \dataexp \mid d\;\mathbf{\coloneq}\;\field{p}{dfield}\\
\dataexp \defeq \;\:& d \mid \fun( \dataexp, \ldots, \dataexp )\\ 
\mathit{cond} \defeq \;\:& \rel( \dataexp, \ldots, \dataexp ) \mid (\mathbf{\neg})?\, \HeapCond
\\ 
\HeapCond \defeq \;\:&  p \,\mathbf{=}\, q 
 \mid p \,\mathbf{=}\,\NIL
\mid \field{p}{\pfield} \,\mathbf{=}\, q
\mid \field{p}{\pfield} \,\mathbf{=}\,\NIL
\end{align*}
\end{mdframed}
}
\label{fig: programming language syntax}
\end{wrapfigure}


\noindent{\em Syntax.}
The syntax of our programming language 
is shown on the right. 
Programs begin with declarations of pointer and data variables, followed by 
labeled instructions. 
%
Instructions include assignments, control flow, and heap operations. Data assignments are of the form $d \,\mathbf{\coloneq}\,\dataexp$, where $d$ is a data variable set to the value of the data expression $\dataexp$. Data expressions are built from data variables and combined using function symbols of the data theory $\mathcal{D}$.
Control flow instructions include $\Skip$, $\exit$, {\bf if-then-else} statements, and   $\WHILE$ loops. Boolean conditions ($\mathit{cond}$) are exclusively either \emph{data conditions} or \emph{heap conditions}.  
Heap conditions can be assigned to a Boolean variable with $d_\mathrm{bool} \,\mathbf{\coloneq}\, \HeapCond$, 
integrating them into Boolean theory.
Heap operations include 
$\New\; p$ (creates a new node, initializes its fields to undefined or $\NIL$, and assigns it to $p$) and $\Free\; p$ (deallocates the node pointed by $p$ and sets all pointers pointing to the node to $\NIL$). We also allow assignment and retrieval of pointer fields (i.e., $\pfield \in \pointers$) 
and data fields (i.e., $\mathit{dfield} \in \dsig)$.
%
Programs are {\bf valid} if they are well-formed, type-correct, uniquely labeled, and terminate with $\exit$.
Fig.~\ref{fig:running example} shows an example of a program.\footnote{
The $\WHILE$ condition goes beyond our syntax but is easily translatable into it.}
%
For a program $P$, $\PC$, $\PV$, and $\DV$ represent program counters, pointer variables, and data variables, respectively. The function $\nextpc$
defines the successor(s) of a program counter.
Statements can have from 0 to 2 successors:
most statements have a single successor, denoted by
$\nextpc(\pc)$;
the $\exit$ statement has no successor; {\bf if-then-else} and $\WHILE$ have two successors based on a Boolean condition: $\nextpc(\pc,\true)$ and $\nextpc(\pc,\false)$.
%
%
The language does not support function calls directly: non-recursive calls are inlined, and limited recursion, typical in tree-based algorithms, can be simulated (see~\Cref{app:recursion} for details).

\smallskip
\noindent{\em Semantics.} 
A program $P$ operates on a specialized heap called a {\em $P$-heap}, that includes all its pointers and data. 
A {\em configuration} of $P$ is a tuple $(\mathcal{H},\nu_p,\nu_d,\pc)$ consisting of a $P$-heap, an evaluation of the pointer variables, an evaluation of the data variables, and next instruction label. 
Focusing on tree-based programs, a configuration $c$ is {\em initial} if it meets the following conditions:
\begin{itemize}
    \item $\mathcal{H}$ is {\em isomorphic} to a data tree $\mathcal{T}$ via a bijection $\rho$ that maps each node in $\mathcal{H}$ to a node in $\mathcal{T}$, such that for all nodes $x,y$ in $\mathcal{H}$ and 
    $i\in |\pointers|$, 
    $y=\mathit{pf_i}(x)$ iff $\rho(y)=\rho(x).i$, where $\mathit{pf_i}$ is the $i$-th pointer field in $\pointers$. We refer to $\mathcal{T}$ as the data tree of $c$, and may use $\mathcal{T}$ in place of $\mathcal{H}$.
    
    \item $\nu_p$ maps the first pointer variable declared in $P$, conventionally denoted by $\rootptr$, to the $\mathcal{H}$ node corresponding to the root of $\mathcal{T}$ and maps the other pointer variables to $\NIL$.
   
    \item $\nu_d$ assigns each variable a non-deterministic value. 
    
    \item  
    $\pc$ is the label of the first statement in $P$.
\end{itemize}

A {\em transition} $c\rightarrow_P c'$ in $P$ occurs by executing the instruction at $\pc$ using standard semantics unless noted otherwise. If $\pc$ is an $\exit$ statement, $c$ becomes a {\em final configuration} with no further transitions. Attempting to dereference or deallocate a $\NIL$ pointer makes $c$ an {\em error configuration}.
An {\bf execution} $\pi$ of $P$ is a (possibly infinite) sequence of configurations $c_0c_1\ldots$ where:
    {\bf(i)} $c_0$ is initial, and
    {\bf(ii)} $c_{i-1}\rightarrow_P c_i$ for each $i\in\nats$. 
A finite $\pi$ that ends in a final or error configuration is a {\em terminating} or {\em buggy} {\em execution}, respectively. 
We aim to solve the following (undecidable) problem: 

\begin{problem}
A program $P$ is \emph{memory safe} if all its executions
terminate without reaching an error configuration. The {\bf memory safety problem} asks whether a given program is memory safe. 
\end{problem}

\section{\LaceTree{s}: Representing Executions as Data Trees} \label{sec:lace}
Our approach to solving the memory safety problem uses the \lacetree encoding, which models a program execution as a single data tree capturing inputs, outputs, and intermediate configurations. 

We first fix notation and assumptions. The encoding uses two parameters $m,n\in\nats$, explained later. 
For simplicity, we assume that heap nodes have a single data field $\datafield$ of an arbitrary type $\datatype$.
We consider a fixed program $P$ and omit it from most notations and statements.
Let $\pi$ be an execution of $P$ starting from an initial configuration $c_0$, where $\dtree = (T, \lambda)$ is a $k$-ary data tree of $c_0$ with signature $\inpsig = \{\datafield: \datatype\}$, and  $\rootptr$ 
points to the root of $\mathcal{T}$. 
The encoding maps $\pi$ to a set of data trees $\KT(\pi, m, n)$,
called the $(m,n)$-\emph{{\lacetree}s} of $\pi$.
We now describe a generic \lacetree $\ktree = (K, \mu)$ from the set $\KT(\pi, m, n)$.


\subsubsection{The backbone.} 

The {\em backbone} $K$ of $\ktree$ is the smallest tree such that:
\begin{enumerate*}[label=(\roman*)]
    \item the input tree $T$ is a subset of $K$ ($T \subset K$),
    \item all nodes from $T$ are internal nodes in $K$,
    \item each internal node of $K$ has degree $k+m$, and
    \item  $K$ has at least one internal node.
\end{enumerate*}
Note that, in the special case where $T$ is empty,
the backbone $K$ is a full tree
of height $2$, consisting of a root and its
$k+m$ children. 
 
Each node of $K$ represents a distinct heap node. Initially, all nodes in $T$ are active, and the rest are inactive; freeing an active node makes it inactive. 


\subsubsection{The node signature.} 

The backbone of a \lacetree depends only on the input data tree and parameters $m,n$, independent of the execution it represents.
Each backbone node is labeled with a sequence of \emph{frames} (a \emph{log}) tracking changes along $\pi$.
Frames form a doubly linked list called \emph{lace}, to maintain chronological order and enable bidirectional navigation. When consecutive operations involve different nodes, frames are inserted along the backbone path connecting them. 

\noindent
The data signature $\Ksig$ of an $(m,n)$-{\lacetree} is the following:
\setlength{\jot}{2.1pt}
{\small
\begin{align*}
\Ksig = \Big\{\: 
&\avail^i:\bool, &&\text{\textcolor{gray}{Is this frame available?}} \\[-1ex]
&\activefield^i:\bool, &&\text{\textcolor{gray}{Is this node allocated?}} \\
&\datafield^i: \datatype, &&\text{\textcolor{gray}{Current value of this node's data field}}\\
&\pc^i : \PC, &&\text{\textcolor{gray}{Program counter}} \\
&\{ d^i : \datatype_d \}_{d\in\DV}, 
&&\text{\textcolor{gray}{Current value of the data variables}} \\
&\{ \chg^i_p: \bool \}_{p \in \PV},
&&\text{\textcolor{gray}{Has $p$ been updated since the frame $i-1$?}} \\
&\{ \isnil^i_p: \bool \}_{p \in \PV},
&&\text{\textcolor{gray}{Is $p$ $\NIL$?}} \\
&\instr^i : \Instr, &&\text{\textcolor{gray}{A pointer update, rewind, or error}} \\
&\activechild^i: \bool^{k+m} &&\text{\textcolor{gray}{Is each child allocated?}} \\
&\nextfr^i: \Dir \times [2,n+1], &&\text{\textcolor{gray}{Link to the next frame}} \\[-1ex]
&\prevfr^i: \Dir \times [2,n]  &&\text{\textcolor{gray}{Link to the previous frame}}
&\!\!\!\!\!\Big\}_{i \in [n+1]},
\end{align*}
}
%
\hspace{-0.20cm}
where $\Dir=\{ \fwd,\, \uparrow \} \cup[k+m]$ encodes the position of an adjacent frame relative to the 
reference frame. 
Each node's label, or \emph{log}, has $n+1$ indexed frames in time order. Once a frame $f$ in a log $\sigma$ is named, its fields are referenced without the index (e.g.,
$f.\pc$ instead of $\sigma.\pc^i$). The last frame of a log handles \emph{label overflow} (if more than $n$ frames are needed), via the  $(n+1)$\textsuperscript{th} frame. The $\prevfr$ field holds a value from $\Dir \times [n]$, since a frame with index $n+1$ has no successor.
The $\instr$ field holds
an {\em event} taken from:
\begin{align*}
\Instr =\: &\big\{ \symb{\pfield \coloneq p} ,\,
                 \symb{\pfield \coloneq \NIL} ,\, 
                 \symb{p \coloneq \here} \mid 
                  \pfield \in \pointers, p \in \PV \big\} \\              
            &\cup\;\big\{\RWD_i \mid i\in [n] \big\} \cup \big\{\RWD_{i,p} \mid i\in [n], p \in \PV \big\}
            \;\cup\; \big\{ \NOP, \Error, \OOM \big\} \,.
\end{align*}
The first group represents updates to pointer fields and variables. $\RWD_i$ and $\RWD_{i,p}$ represent \emph{lace rewinding} events. Other symbols denote the empty event ($\NOP$), null-pointer dereference ($\Error$), and out-of-memory error ($\OOM$) caused by excessive use of the statement $\New$. 

\subsection{The Labeling Function}\label{sec: labelling of lacetrees} \label{sec:backbone}

The labeling function $\mu$ of $\ktree$ is defined inductively on the length of $\pi$.

\subsubsection{\bf Base case.}\label{sec:baselabel}
$\pi$ consists of an initial configuration, say $(\dtree,\nu_p,\nu_d,\pc)$. We encode the input tree $\dtree=(T,\lambda)$ into the backbone by setting the first frame of each node $t \in K$ as follows:
\begin{align*}
\mu(t).\avail^1 &= \false, &
\mu(t).\activefield^1 &= \begin{cases}
      \true  &\text{if } t \in T \\
      \false &\text{otherwise}
\end{cases} \\
\mu(t).\datafield^1 &= \begin{cases} 
      \lambda(t).\datafield &\text{if } t \in T \\
      \text{unspecified}    &\text{otherwise}
   \end{cases}
&
\mu(t).\activechild^1_j &= \begin{cases}
      \true  &\text{if } t.j \in T \\
      \false &\text{otherwise.}
      \end{cases}
\end{align*}
All other fields of the first frame are unspecified. 
The root’s \emph{second} frame stores 
$\pi$’s initial configuration:
$\mu(\epsilon).\avail^2 = \false$, 
$\mu(\epsilon).\prevfr^2 = (\fwd, 2)$ (a self-loop), $\mu(\epsilon).\pc^2$ is the first statement's label in $P$, and $\mu(\epsilon).\isnil_q^2 = \true$
for all pointer variables $q$ different from $\rootptr$. If $T$ is not empty, $\mu(\epsilon).\instr^2 = \symb{\rootptr \coloneq \here}$
and $\mu(\epsilon).\isnil_{\rootptr}^2 = \false$;
otherwise, $\mu(\epsilon).\instr^2 = \NOP$ and $\mu(\epsilon).\isnil_{\rootptr}^2 = \true$. All other frames 
are marked as available. 
The root's second frame also copies  $\activefield$, $\datafield$, and $\activechild$ from the first frame.

\subsubsection{Inductive case.}
We start with an overview of the encoding method, its properties, and the required notation. Let $\pi=\overline{\pi}c$, where $c$ is a configuration and $\overline{\pi}$ is a non-empty execution. 
Assume that $\overline{\ktree} = (K, \overline{\mu})$ is a {\lacetree} in $\KT(\overline{\pi}, m, n)$.
We define the labeling $\mu$ for $\pi$ by extending the lace of $\overline{\ktree}$ based on the last instruction executed in $\pi$. To aid understanding, we list some invariants for all {\lacetree}s, providing an informal explanation for brevity.

\smallskip
\noindent{\em The lace.} Besides individual node logs, we maintain a chronological order of all frames across all nodes. 
All the unavailable frames in the \lacetree with index greater than $1$ form a doubly linked list called the {\em lace} using the $\nextfr$ and $\prevfr$ fields of the frames. The first frame in the lace is the root's second frame.
Each frame is identified by a pair $(u, i)$, where $u$ is a node and $i \in [n+1]$ is the frame's index. A frame $(v,j)$ is the {\em lace successor} of frame $(u,i)$, denoted $(u,i) \to_{\nextfr} (v,j)$ (and $(u,i)$ is the {\em lace predecessor} of $(v,j)$, written $(v,j) \to_{\prevfr}(u,i)$) if $i,j>1$ and one of the following holds:
\begin{itemize}
\item $u=v$, $j=i+1$, $\mu(u).\nextfr^i = (\fwd,j)$, and $\mu(v).\prevfr^{j} = (\fwd,i)$;
\item $v$ is the $l$\textsuperscript{th} child of $u$, $\mu(u).\nextfr^i = (l,j)$, and $\mu(v).\prevfr^j = (\uparrow,i)$;
\item $u$ is the $l$\textsuperscript{th} child of $v$, $\mu(u).\nextfr^i = (\uparrow,j)$, and $\mu(v).\prevfr^j = (l,i)$.
\end{itemize}
Available frames' unspecified fields contribute to $\KT(\pi, m, n)$'s non-determinism.

\smallskip
\noindent{\em Properties of the frame fields.} In our inductive definition, we obtain the \lacetree for $\pi$ by extending that of $\overline{\pi}$ by appending frames for $\pi$'s final step. This helps us assign meanings to fields like $\chg_p$, $\isnil_p$, $d$, $\datafield$,  $\pc$, and $\activechild$ for the unavailable frames.
\label{page:upd}
The $\chg_p$ flag tracks changes to the pointer $p$ outside the current node: it is set to  $\true$ in frame $(u, i)$ (for $i>1$) if $p$ was assigned a non-$\NIL$ value in the part of the lace between frames $(u,i-1)$ and $(u,i)$ (excluding these frames). 
Thus, if frames $(u,i-1)$ and $(u,i)$ are adjacent in the lace (i.e., $(u,i-1) \to_\nextfr (u,i)$, aka an \emph{internal step}), all $\chg_p$ flags in $(u,i)$ are $\false$. 
The $\isnil_p$ flag is $\true$ in a frame if $p=\NIL$ at that point in the execution. 
The $\activechild$ flags help track the allocation of the child nodes of the backbone.
Other fields preserve their usual meanings.

\smallskip
\noindent{\em Pushing a frame.} 
Appending a frame to a log involves: \emph{(1)} finding the smallest index $i$ with $\avail^i$ is $\true$, and \emph{(2)} adding the new frame at position $i$. Hence, a log behaves like a stack, with the bottom frame at index $1$ and the {\em top frame} being the highest-index frame where $\avail$ is $\false$.

\smallskip
\noindent{\em Default values for a frame.} \label{sec:defaults}
Any frame pushed onto a log assumes default values unless specified otherwise. When pushing a frame $f$ on a node $u$, default values come from the preceding frame in the lace, $f^\prev$, or the frame below $f$ in $u$'s log, $f^\below$.  
Note that $f^\prev$ and $f^\below$ can be the same.
The default values for the fields of $f$ are as follows:  for all $p \in \PV$ and $d \in \DV$, 
\begin{itemize}
\item $\avail = \false$, $\instr = \NOP$, and $\chg_p = \false$;
\item $\activefield$ and $\datafield$ are copied from $f^\below$;
\item $\isnil_p$, $d$, and $\pc$ are copied from $f^\prev$; 
\item $\activechild_j$ is copied from $f^\prev.\activefield$ if $f^\prev$ belongs to
the $j$\textsuperscript{th} child of $u$; otherwise, it is copied from $f^\below.\activechild_j$;
\item $\prevfr$ points to $f^\prev$;
\item $\nextfr$ is unspecified and can take any value in different {\lacetree}s for the same execution.
\end{itemize}
Moreover, $f^\prev.\nextfr$ is updated to point to $f$, eliminating the non-determinism
of the $\nextfr$ field of the previous frame.

Despite the non-determinism in the $\nextfr$ field, identifying the last frame $f$ in a lace can be done by
checking $f$ and its lace successor $f'$.  
Specifically, $f$ is the last frame of the lace if $f'$ is available or $f'$ precedes $f$ in the lace,
which happens when field of $f'.\prevfr$ does not point to $f$.

Henceforth, assume that $\overline{f}$ is the final lace last frame in $\overline{\ktree}$, located as the top frame of node $\overline{t}$.
%
We first present the encoding of the statements that push a single new frame $f$ on the current node $\overline{t}$. The fields of $f$ are set to default values, except for those specified below.

\smallskip
\noindent{\bf Encoding of $p \coloneq \NIL$:} $f.\pc = \nextpc(\overline{f}.\pc)$
and $f.\isnil_p = \true$.

\smallskip

\noindent{\bf Encoding of $d \coloneq \dataexp$:} $f.\pc = \nextpc(\overline{f}.\pc)$ and $f.d$ is set to the value of $\dataexp$, with variables in $\DV$ evaluated using their values from $\overline{f}$.

\smallskip

\noindent{\bf Encoding of $\Skip$:} $f.\pc = \nextpc(\overline{f}.\pc)$.


\smallskip
The other statements may operate on different nodes in addition to $\overline{t}$. 
The main reasons to move to another node are to dereference a pointer or identify the node a pointer field refers to. To get this information, we {\em rewind the lace} by moving backward to find the most recent assignment to the relevant pointer. For example, to identify the node a pointer variable $p$ points to, we rewind the lace until we find a frame with the event $\symb{p \coloneqq \here}$. 

\smallskip
\noindent{\em Lace rewinding function.}
To enable rewinding, we define the auxiliary function $\findptr(p,\mathit{id})$, which takes a pointer $p \in \PV$ and a frame ID and returns a sequence of frame IDs by traversing the lace backward from $\mathit{id}$,
until the most recent assignment to $p$.
The sequence uses \emph{shortcuts}, including only the IDs where the lace moves between nodes and where such moves are \emph{relevant} to track $p$, as indicated by the $\chg_p$ flags. For example, in Fig.~\ref{fig:kt1},  rewinding from frame 15 to resolve $\mathit{head}$ gives the following:
\footnote{Using the global lace positions as frame IDs (red numbers in Fig.~\ref{fig:kt1}).}
$$\findptr(\mathit{head}, (u_4,2) ) = \findptr(\mathit{head}, 15) = 
\big( (u_3, 2), (u_2, 2), (u_1, 2) \big) 
= (9, 4, 1).$$
Frames 9 and 4 have a predecessor that belongs to another node; moreover, they represent the earliest occurrence of their node in the lace.
Contrast this with frame 12, which also follows a frame in another node, but is \emph{not} in the sequence because 12 is not the first visit to $u_3$, and
$\chg_\mathit{head}$ is false in 12.
Frame 1 is included as 
the value of $\mathit{head}$ in frame 15 was established in frame 1. Thus, rewinding from $15$ adds frames $16,17$ and $18$ for these IDs.

\smallskip
\noindent{\em Null pointer dereference.}
If the instruction located at $\overline{f}.\pc$ dereferences a pointer $p$, and $p$ is $\NIL$, it indicates an error.
Thus, if the flag $\isnil_p$ is $\true$ in $\overline{f}$, we push a new frame with the event $\Error$ onto the current node $\overline{t}$  to indicate a runtime error.
We now present the encoding for the other types of statement. 

\smallskip
\noindent{\bf Encoding of $\field{p}{\pfield} \coloneq \NIL$.}
We rewind the lace to find the latest assignment to $p$. 
Let $(\overline{t},\overline{i})$ be the identifier of $\overline{f}$,
and $id_1 = (u_1,i_1),\ldots, id_l = (u_l,i_l)$ be the sequence $\FindPtr\big( p, (\overline{t},\overline{i}) \big)$.
We push a new frame for each element of the sequence 
$id_1,\ldots,id_l$, to keep a faithful record of the movements necessary to simulate the current statement: 
for each $j \in [l-1]$, we push a frame onto $u_j$ with the event $\RWD_{i_j}$, without advancing the program counter. 
Finally, 
we push a frame $f$ onto $u_l$ with $f.\pc = \nextpc(\overline{f}.\pc)$ and $f.\instr = \symb{\pfield \coloneq \NIL}$. 

\smallskip
\noindent{\bf Encoding of $\field{p}{\pfield} \coloneq q$.}  
Same as the encoding for $\field{p}{\pfield} \coloneq \NIL$, but the final frame's event is set to $\symb{\pfield \coloneq q}$. 

\smallskip
\noindent{\bf Encoding of $p \coloneq q$.}
If $\isnil_q$ evaluates to $\true$ in $\overline{f}$,
we push a new frame onto $\overline{t}$ with $\isnil_p$ set to $\true$.
Otherwise, a rewinding operation takes the lace to the node pointed by $q$, where we push a frame with $\symb{p \coloneq \here}$.
In either case, the last frame also updates the program counter.

\smallskip
\noindent{\bf Encoding of $\New\ p$.} 
Let $j$ be the smallest index in $[k+1, k+m]$
such that $\activechild_j$ is $\false$.
If no such index exists, we push a frame on $\overline{t}$
with the event $\OOM$, representing an out-of-memory error.
Otherwise, the lace moves to the $j$\textsuperscript{th} child of $\overline{t}$
and pushes a frame $f$ there with
$f.\pc = \nextpc(\overline{f}.\mathit{pc})$, $f.\instr = \symb{p \coloneq \here}$,
$f.\activefield = \true$, and $f.\isnil_p = \false$.

\smallskip
\noindent{\bf Encoding of $\Free\ p$.}
A rewinding operation takes the lace to the node $u$ pointed by $p$,
where we push a frame $f$ with $\activefield = \false$ and
an updated $\pc$. 
In the new frame, $\isnil_q$ is set to $\true$ for every pointer
$q$ that currently points at $u$, including $p$.
To find such pointers, 
let $(u,i)$ be the id of $f$. Then, $q$ points at $u$ if the log of $u$  contains another frame $(u,j)$ such that: $j<i$, $(u,j)$ contains $\symb{q \coloneqq \here}$, and the $\chg_q$ flag is $\false$ in all frames
from $(u,j+1)$ to $(u,i)$.
%

\smallskip
\noindent{\bf Encoding of $p \coloneq \field{q}{\pfield}$.}
First, the lace moves to the node $u$ pointed by $q$ 
through a rewinding process.
Then, we search in $u$'s log for the most recent assignment to $\pfield$. 
We look for the largest index $i$ s.t.\ the frame $(u,i)$ is in use and contains an event of the form $\symb{\pfield \coloneqq \alpha}$, for some $\alpha$.
If no such index exists, $\pfield$ is interpreted as having its default value pointing to a child in the original input tree. 
We then distinguish the following cases:
\begin{description}
\item[{[$\alpha = \NIL$]}] We push a frame with $\isnil_p = \true$ on $u$.
\item[{[$\alpha = r$, for some $r \in \PV$]}] If $\isnil_r$ is $\true$ in the current frame,
the lace pushes a frame with $\isnil_p = \true$.
Otherwise, the lace moves again to the node pointed by $r$,
and there it pushes a frame with $\symb{p \coloneqq \here}$.
\item[{[The log of $u$ does not contain an explicit assignment to $\pfield$]}] If $\pfield$ is the $j$\textsuperscript{th} field in $\pointers$, $u.j \in T$, 
and $u.j$ is active (as encoded in the flag $\activechild_j$),
the lace moves to $u.j$ and pushes
a frame with event $\symb{p \coloneqq \here}$ there.
Otherwise, a frame with $\isnil_p = \true$ is pushed on $u$.
\end{description}
The last pushed frame always updates the program counter. 

\smallskip
\noindent{\bf Encoding Boolean conditions and control-flow statements.} 
Data conditions are evaluated locally using variable values in the current frame $\overline{f}$. For heap conditions, we may need to traverse the lace. We focus on conditions of the form $p=q$, since others (e.g., $\field{p}{\pfield} = q$) can be reduced to this form using auxiliary variables. 
%
To evaluate $p = q$, we first check the $\isnil$ flags: if both pointers have their $\isnil$ flags set to true, the condition is true; if the flags differ, it is false.
If this is inconclusive, we rewind the lace to find an assignment to $p$ or $q$. For example, upon finding $\symb{p \coloneqq \here}$ in frame $(u, i)$, we search in $u$'s log for the largest index $j<i$ where frame $(u, j)$ has $\symb{q \coloneqq \here}$. If none exists, the condition is false. If found, the condition holds if $q$ was unchanged between $(u, j)$ and $(u, i)$, verified via $\neg \bigvee_{l\in[j+1,i]} \chg_q^l$. A new frame is then pushed, updating the program counter accordingly. 
The process is symmetric if $\symb{q \coloneqq \here}$ is found first.

\begin{table}
\small
\setlength{\tabcolsep}{8pt}
\renewcommand{\arraystretch}{1.1}
\begin{center}
\begin{tabular}{lm{3.7cm}m{4.4cm}}
{\bf Statement} &{\bf Movement} &{\bf Information stored} \\ \hline\hline
$\field{p}{\pfield} \coloneq \NIL$ &Find $p$ or fail &$\symb{\pfield \coloneqq \NIL}$ \\ \hline
$\field{p}{\pfield} \coloneq q$    &Find $p$ or fail &$\symb{\pfield \coloneq q}\!$ or $\!\symb{\pfield \coloneqq \NIL}$ \\ \hline
$p \coloneq q$                         &Find $q$         &$\symb{p \coloneq \here}$ or set $\isnil_p$ \\ \hline
$p \coloneq \field{q}{\pfield}$    &\raggedright Find $q$ or fail, then find last assignment to $\pfield$ 
                                                 &$\symb{p \coloneq \here}$ or set $\isnil_p$ \\ \hline
$\field{p}{\datafield} \coloneq \dataexp$ &Find $p$ or fail   &Update $\datafield$ \\ \hline
$d \coloneq \field{p}{\datafield}$        &Find $p$ or fail   &Update $d$ \\ \hline
$\New\ p$                            &Move to the first inactive child or fail   &$\symb{p \coloneq \here}$ and set $\activefield$ \\ \hline
$\Free\ p$                           &Find $p$ or fail   &Reset $\activefield$ \\
\hline\hline
\end{tabular}
\end{center}
\caption{Summary of the encodings.}\label{tab:walking}
\end{table}



\paragraph{On the choice of parameters $m$ and $n$.} Parameter $m$ bounds the number of allocations: $m=0$ for programs without allocations, while $m=1$ is adequate for programs that insert a single new node. The parameter $n$ limits the passes and instructions executed per node; while some programs need unbounded labels, typical tree-like algorithms work with moderate $n$ (usually $\leq 10$).

\subsection{Relations with Program Executions}

Our first result establishes that \lacetree{s} provide an accurate and faithful representation of program executions.
\begin{restatable}{theorem}{thmcomputable} \label{thm:computable} 
    Given a program $P$ and parameters $m,n$, $\KT(\cdot, m, n)$ is computable.
    Moreover, there is a computable function $\mathit{exec}$ 
    such that, 
    for any data tree $\ktree$ that is a \lacetree of $P$, $\mathit{exec}(\ktree)$ is 
    an execution $\pi$ of $P$ s.t.\ $\ktree \in \KT(\pi,m,n)$.
\end{restatable}

Notice that 
the relation between executions $\pi$ and \lacetree{s} $\ktree$ defined by $\ktree \in \KT(\pi,m,n)$
is neither injective nor functional.
It is not functional due to the non-determinism in the encoding.
It is not injective because a \lacetree ending in a label overflow represents only a prefix of an execution and
thus relates to all executions sharing that prefix.



\smallskip
\noindent{\em Exit status of a \lacetree.} 
To distinguish how {\lacetree}s terminate, we introduce the notion of {\em exit status} for individual frames and for the entire {\lacetree}. Each frame $f$ of a {\lacetree} is assigned one of five statuses in $\EStatus = \{ \bf{C}, \bf{E}, \bf{O},$ $\bf{M}, \bf{N} \}$, with the following meanings:
\begin{itemize}
\item {\bf C}lean exit: The program counter (\textit{pc}) of $f$ points to an $\exit$ instruction.
\item Runtime {\bf E}rror: $f.\instr = \Error$, indicating a null-pointer dereference.
\item Label {\bf O}verflow: The index of $f$ in its label is $n+1$, indicating log overflow.
\item Out of {\bf M}emory: $f.\instr = \OOM$, indicating a failed attempt to allocate a node with the $\New$ statement due to the absence of inactive nodes.
\item {\bf N}one: Indicates that frame $f$ is not a terminal frame.
\end{itemize}
A frame $f$ is {\em terminal} if its status is not {\bf N}one. 
Indeed, the exit statuses different from {\bf N} terminate the lace,
hence only the last frame in the lace may have an exit status
different from {\bf N}.
The {\em exit status} of a \lacetree $\ktree$,
denoted $\exitst(\ktree)$, is the status of the last frame in its lace.
The theorem below links {\lacetree}'s exit status to its corresponding executions.

\begin{restatable}{theorem}{thmktexits} \label{thm:kt-exits}
For all executions $\pi$ and 
$\ktree \in \KT(\pi,m,n)$, the following hold:
\begin{enumerate}
\item \label{prop:tree-error} If $\exitst(\ktree) = \bf{E}$, then $\pi$ ends in an error 
configuration.
\item \label{prop:execution-error} If $\pi$ ends in an error configuration, then $\exitst(\ktree) \in \{ \bf{E}, \bf{O}, \bf{M} \}$.
\item \label{prop:infinite} If $\pi$ is infinite, then $\exitst(\ktree) = {\bf O}$.
\end{enumerate}
\end{restatable}



\section{Properties of {\LaceTree}s} \label{sec:compositionality}

\noindent{\em Prefix of a \lacetree.}
An $(m,n)$-\lacetree is a \lacetree associated with an execution $\pi$ in $\KT(\pi,m,n)$.
A \emph{prefix} of a \lacetree $\ktree$ is a data tree obtained from $\ktree$ by truncating its lace at a frame $f$, setting $f'\!.\avail$ to \true\ for all subsequent frames, 
and leaving $f.\nextfr$ unconstrained.

\medskip

\noindent{\bf Locality.}
For a label $\sigma$ and $i\in\mathbb{N}$,
we define $\sigma^{<i}$ (resp., $\sigma^{\leq i}$) as 
the label obtained by setting 
$\avail$ to true in all frames of $\sigma$ with indices $\geq i$ (resp., $>i$).

The following lemma states that each new frame in a \lacetree depends only on \emph{local} information from neighboring nodes. 
We use $f_1 \equiv f_2$ to indicate that frames $f_1$ and $f_2$
differ only in their $\nextfr$ field.
\begin{restatable}[Locality]{lemma}{lemlocality} \label{lem:locality}
There exist functions $\FunUp, \FunDown, \FunInt$ such that for logs $\sigma,\tau_1,\ldots,\tau_{k+m}$  of a node $u$ and of its children in a \lacetree prefix:
\begin{itemize}
    \item For all steps $(u.j,b) \rightarrow_\mathit{next} (u, a)$ in the lace,
    it holds $\sigma^a \equiv \FunUp(\tau_j^{\leq b}, j, \sigma^{< a})$. 
    
    \item For all steps $(u,a) \rightarrow_\mathit{next} (u.j,b)$ in the lace,
    it holds $\tau_j^b \equiv \FunDown(\sigma^{\leq a}, j, \tau_j^{<b})$.

    \item For all steps $(u,a) \rightarrow_\mathit{next} (u, a+1)$ in the lace,
    it holds $\sigma^{a+1} \equiv \FunInt(\sigma^{\leq a})$.
\end{itemize}
\end{restatable}
\noindent{\bf Compositionality.}
Using the functions $\FunUp$ and $\FunDown$ from \Cref{lem:locality}, we define the predicate $\ConsistentChild(\tau, j, \sigma)$.
This predicate is meant to check whether two logs $\sigma$ and $\tau$ may
belong to the same {\lacetree} as the logs of a node and its $j$\textsuperscript{th} child.
Specifically, it verifies that all pairs of consecutive frames, linked by $\nextfr$ and $\prevfr$, with one frame belonging to $\tau$ and the other frame belonging to $\sigma$, adhere to the functions $\FunUp$ and $\FunDown$. A detailed definition of $\ConsistentChild$ is in~\Cref{app:detailed-encoding-consistency-child}.


\smallskip


From $\ConsistentChild$ and \lacetree prefix definitions, the next lemma follows, ensuring that $\ConsistentChild$ holds on all parent-child log pairs.



%
\begin{restatable}{lemma}{lemconsistentchild} 
\label{lem:consistent-child1}
For all labels $\sigma, \tau \in L(\Ksig)$ and indices $j \in [k + m]$, if there exists a {\lacetree} prefix where $\sigma$ and $\tau$ are the logs of a node and its $j$\textsuperscript{th} child respectively, then $\ConsistentChild(\tau, j, \sigma)$ holds.
\end{restatable}

The following lemma establishes a key property of $\ConsistentChild$ for our verification approach, enabling {\lacetree} composition from different subtrees. 

\begin{restatable}[Compositionality]{lemma}{lemcompositionality} \label{lem:compositionality}
Let $\sigma_1,\sigma_2 \in L(\Ksig)$ be the logs of nodes $t_1,t_2$ in $(m,n)$-\lacetree prefixes $\ktree_1,\ktree_2$.
If $\ConsistentChild(\sigma_2, j, \sigma_1)$ holds true
for some $j \in [k+m]$,
then there is an $(m,n)$-\lacetree prefix $\ktree$ where $\sigma_1$ is the log of a node and $\sigma_2$
is the log of its $j$\textsuperscript{th} child.
Moreover, $\ktree$ is obtained by replacing the $j$\textsuperscript{th} subtree of $t_1$ in $\ktree_1$ with the subtree rooted at $t_2$ in $\ktree_2$.
\end{restatable}
\begin{proof}
    Let $\ktree$ be the data tree obtained by replacing the subtree rooted at the $j$\textsuperscript{th} child of $t_1$ in $\ktree_1$ with the subtree rooted at $t_2$ in $\ktree_2$.
    We prove that $\ktree$ is a \lacetree prefix by induction on the number of pairs $(a,b)$ where frames $\sigma_1^a$ and $\sigma_2^b$ are adjacent in the lace, each such pair representing an interaction between the parent's and the child's labels.

    When the above number is zero, there are no interactions between $t_1$ and its $j$\textsuperscript{th} child in $\ktree_1$.    
    Let $\pi_1$ be an execution s.t.\ $\ktree_1$ is a prefix of a \lacetree representing $\pi_1$. It is direct to show that $\ktree$ is a \lacetree prefix for an execution $\pi$ following the same steps as $\pi_1$, 
    starting with the input tree of $\ktree$. Since $\pi_1$ never visits the $j$\textsuperscript{th} child of $t_1$, and this subtree is the only difference between $\ktree_1$ and $\ktree$, $\pi$ is a valid execution of $P$.

    For the inductive case, consider the last interaction $(a,b)$ between $\sigma_1^a$ and $\sigma_2^b$.
    First, assume that such interaction is a step \emph{up} from frame $\sigma_2^b$ to $\sigma_1^a$. Let $\ktree_1'$ be derived from $\ktree_1$ by truncating its lace to end just before $\sigma_1^a$.
    Clearly, $\ktree_1'$ is still a \lacetree prefix, and the modified label $\sigma_1'$ of $t_1$ is obtained from $\sigma_1$ by removing the frames with index at least $a$, by setting their $\avail$ flags to true. We can now apply the inductive hypothesis to the labels $\sigma_1'$ and $\sigma_2$, because we have removed one interaction between them. Hence, we can assume that there is a single \lacetree prefix $\ktree'$ containing both labels as the logs of a parent and its $j$\textsuperscript{th} child. We then obtain the desired \lacetree prefix $\ktree$ from $\ktree'$ by reintroducing the sequence of frames removed from $\ktree_1$. We need to prove that adding those frames respects all rules of \lacetree{s}. The correctness of the first added frame, $\sigma_1^a$, is ensured by $\ConsistentChild(\sigma_2, j, \sigma_1)$, because it applies the function $\FunUp$ to all upward interactions between $\sigma_2$ and $\sigma_1$. In turn, Lemma~\ref{lem:locality} ensures that adhering to that function is 
    sufficient to establish the correctness of the next frame. 
    Subsequent frames can be reintroduced due to the unchanged surroundings. Lemma~\ref{lem:locality} ensures that no other information is relevant.  
    
    The other case to prove is when the last interaction is a step \emph{down} from the parent's frame $\sigma_1^a$ to the child's frame $\sigma_2^b$. Define $\sigma_2'$ as the label obtained from $\sigma_2$ with the frames of indices $b$ and above removed. Similar to the previous case, we apply the inductive hypothesis to the shortened label $\sigma_2'$ and its shortened \lacetree prefix $\ktree_2'$, resulting in a \lacetree prefix $\ktree'$. We then reintroduce the frames removed from $\ktree_2$ into $\ktree'$. The correctness of the first reintroduced frame is ensured by $\ConsistentChild$ checking the function $\FunDown$ on every downward interaction. The subsequent reintroduced frames are still valid because there are no steps returning to the parent of $t_2$, and their surroundings remain unchanged from $\ktree_2$.\qed
\end{proof}

\begin{example}
Consider the \lacetree in Fig.~\ref{fig:kt1}, 
with node labels $\sigma_1,\ldots,\sigma_5$,
 and another {\lacetree} of the same program on the input
 list  7 \goesto 8 \goesto 9 \goesto 3 \goesto 10 \goesto 11 \goesto 12 and $\mathit{key} = 3$.
 Let $v_4$ be the node of the second {\lacetree} with value $3$,
 and let $\tau_4$ be its label.
 By inspecting the second \lacetree, one can observe
 that the occupied frames of $\tau_4$ (i.e., those with $\avail=\false$) contain the same information as the occupied frames of $\sigma_3$. Therefore, $\ConsistentChild(\tau_4, 1, \sigma_2)$ holds, and by Lemma~\ref{lem:compositionality}, the two {\lacetree}s can be composed at nodes $u_2$-$v_4$ to construct a
 {\lacetree} for the input list 1 \goesto 2 \goesto 3 \goesto 10 \goesto 11 \goesto 12.
\end{example}



\section{Reasoning about {\LaceTree}s with {\CHC}s}\label{sec:reasoning}


We introduce a {\CHC} system $\Ckt$ for a program $P$ with parameters $m,n\in\nats$. 
It 
employs a single uninterpreted relation symbol, $\invCex(\sigma)$, where $\sigma$ matches the data signature $\Ksig$ of \lacetree{s}, ensuring the following:

\begin{restatable}
{theorem}{thmlabels} 
\label{thm:lab}
In the minimal model of $\Ckt$, 
$\invCex(\sigma)$ holds for a label $\sigma$ iff $\sigma$ labels a node in some $(m,n)$-\lacetree prefix $\ktree$ of $P$.
\end{restatable}
We define $\Ckt$ using the {\lacetree} rules constructing constructing partial {\lacetree}s (impractical for enumerating all {\lacetree}s). Instead, we rely on the compositionality lemma (\Cref{lem:compositionality}), which states that two consistent labels imply the existence of 
a {\lacetree} where those labels are logs of an internal node and one of its children, and the locality lemma (\Cref{lem:locality}) to extend the lace involving these nodes. This lemma entails that constructing a {\lacetree} involves adding frames to node logs so that any two consecutive frames belong to the same node or to neighboring backbone nodes. 
We use this property in the {\CHC} system, employing independent {\CHC}s to simulate adding a single frame. In the {\CHC} system, we simulate adding a single frame via $\FunUp$ (upward), $\FunDown$ (downward), and $\FunInt$ (within the same log). We define predicates $\Psi_{\FunDown}(\sigma, j, \tau_j, f), \Psi_{\FunUp}(\tau_j, j, \sigma, f)$, and $\Psi_{\FunInt}(\sigma, f) $ to constrain logs accordingly, with $f$ as the resulting frame.

\begin{figure}[t]
\begin{mdframed}[skipabove=0pt,skipbelow=0pt,innertopmargin=-7.0pt]
\begin{align*}
\texttt{\bf (I)}\:\,\,\,\quad\invCex(\sigma) &\leftarrow \len(\sigma, 1) \wedge \FirstFrame(\sigma^1) \quad\,\,\,\,\text{\textcolor{gray}{\em\small Initializing non-root nodes}}\\[1ex]
\texttt{\bf (II)}\,\,\,\quad\invCex(\sigma) &\leftarrow \len(\sigma, 2) 
\wedge \startframe(\sigma) 
\qquad\qquad\quad\,\text{\textcolor{gray}{\em\small Initializing the root node}}  \\[1ex]
\texttt{\bf (III)}\,\,\quad\invCex(\sigma) &\leftarrow 
\len(\sigma, i) \wedge\invCex(\sigma^{<i}) 
\wedge  \Psi_\FunInt(\sigma^{<i}, \sigma^i)\,\quad
\text{\textcolor{gray}{\em\small Internal step}} \\[1ex]
&\hspace{-1.1cm}\text{\textcolor{gray}{\em\small A step from the $j$\textsuperscript{th} child to its parent}}\\
\texttt{\bf (IV)}\,\,\quad \invCex(\sigma) &\leftarrow \
\len(\sigma, i)\,\wedge\,\invCex(\sigma^{<i}) \,\wedge \,\invCex(\tau) \;\;   \\ 
&\quad\,\,\,\wedge \, \ConsistentChild(\tau, j, \sigma^{<i})\,\wedge\, \Psi_\FunUp(\tau, j, \sigma^{<i}, \sigma^i) \\[1ex]
&\hspace{-1.1cm}\text{\textcolor{gray}{\em\small A step from a node to its $j$\textsuperscript{th} child}}\\
\texttt{\bf (V)}\quad\,\,\,\,\invCex(\tau) &\leftarrow
\len(\tau, i) \,\wedge\, \invCex(\sigma) \,\wedge\,\invCex(\tau^{<i})\\
&\quad\,\,\, 
\wedge \, \ConsistentChild(\tau^{<i}, j, \sigma)\,\wedge \,\Psi_\FunDown(\sigma, j, \tau^{<i}, \tau^i) 
\\[1.2ex] \hline\\[-1.2ex]
\texttt{\bf (VI)}\qquad\,\,\,\;\;\;\perp &\leftarrow \invCex(\sigma) \wedge \LabelExit(\sigma, \ExitStatus) \quad\,\:\: \text{\textcolor{gray}{\em\small Lace ends with status in Ex}}
\end{align*}
\end{mdframed}
\vspace{-0.8em} 
\caption{{\CHC}s (I)-(V) form the {\CHC} system $\Ckt$,
while the {\CHC} system $\Cex$ includes all the {\CHC}s in the figure. Here, $i \in [2,n]$ and $j \in [k+m]$.}
\label{fig:CHC-rules}
\vspace{-1em} 
\end{figure}

\Cref{fig:CHC-rules} details the {\CHC}s of $\Ckt$. 
While describing each {\CHC}, we  establish the ``only if'' direction of \Cref{thm:lab}, by induction on the number of {\CHC} applications needed to insert $\sigma$ into the minimal interpretation of $\invCex$. For completeness, the proof of the ``if'' direction appears in \Cref{app:proof6}.  

Before detailing the {\CHC}s, we introduce some notation. Let $\sigma\in L(\Ksig)$ and $i\in\mathbb{N}$. The formula $\len(\sigma, i)$ is true if all frames with indices in $[i]$ are unavailable and all other frames are available, i.e., $\len(\sigma, i) \defeq \neg \sigma^i.\avail \wedge \bigwedge_{j = i+1}^{n+1} \sigma^{j}.\avail \,.$
With an abuse of notation, we write $\sigma^{<i}$ in a {\CHC} as a shorthand for a fresh variable $\theta$, together with the conjunct 
$\bigwedge_{\ell \in [i-1]} (\theta^\ell = \sigma^\ell) \wedge \bigwedge_{\ell \in [i,n+1]} \theta^\ell.\avail$.

\smallskip

{\bf {\CHC}s~I} and {\bf II} ensure that $\invCex$ includes labels of all {\lacetree}s of $0$-length executions of $P$, forming the base case for the ``only if'' direction of \Cref{thm:lab}, since both are facts.
{\CHC}~I defines non-root labels with all but the first frame available, and 
consistent $\activefield$/$\activechild$, while {\CHC}~II defines root labels by disabling the first two frames and using  $\startframe(\sigma)$  for the second frame,
as per the base-case labeling (\Cref{sec:baselabel}).

\smallskip

The remaining {\CHC}s extend each node's log 
frame by frame,
following the inductive \lacetree label definition. For \Cref{thm:lab} (``only if'' direction), we assume inductively  that the body labels satisfy the claim, i.e., they label a node in a $(m,n)$-\lacetree prefix, and show that the head label does too. 

\smallskip

{\bf{\CHC}~III} handles  {\em internal steps}, where $\sigma^i$ follows $\sigma^{i-1}$ in the lace.
Here, $\invCex(\sigma^{<i})$ ensures 
that $\sigma^{<i}$ labels a node in an $(m,n)$-\lacetree prefix, while  $\Psi_\FunInt(\sigma^{<i}, \sigma^i)$ 
constrains 
$\sigma^i$ to encode the next internal step, as per \Cref{lem:locality}.

{\bf {\CHC}~IV} handles cases where the lace extends with a new frame pushed to the parent of the previous frame, typically during a rewind phase. 
The predicate $\ConsistentChild$ ensures that $\sigma^{<i}$ and $\tau$ belong to the same \lacetree prefix, as the log of a parent and its $j$\textsuperscript{th} child (\Cref{lem:compositionality}). Then, $\Psi_\FunUp$ extends the lace by adding a frame to $\sigma^{<i}$, following the topmost frame of $\tau$.



 
{\bf\CHC~V} handles the reverse of \CHC~IV, where the current lace extends from a parent to its $j$\textsuperscript{th} child. Using $\Psi_\FunDown$, it ensures that $\tau_j$ correctly extends $\tau_j^{<i}$ with a new frame for the latest step.

\subsubsection{The Exit Status Problem.}
We present a method to check whether a program can lead to a memory safety error via an execution that can be represented by an $(m,n)$-{\lacetree}. This reduces to solving a {\CHC} system: if unsatisfiable, such an execution exists.
%
%
We formalize this as the following decision problem.

\begin{problem}
An instance  of the {\bf exit status problem}
is a tuple $(P, m, n, \ExitStatus)$,
where $P$ is a program,
$m,n\in\nats$, and $\ExitStatus$ is a set of exit statuses excluding {\bf N}.
The exit status problem asks whether there exists an $(m,n)$-{\lacetree} of $P$ whose exit status is in $\ExitStatus$. 
\end{problem}


We solve the exit status problem using the {\CHC} system $\Cex$ (\Cref{fig:CHC-rules}), which includes 
    (i) all {\CHC}s from $\Ckt$, 
    crucial for \Cref{thm:lab},
    and (ii) a single query, {{\CHC}~VI}, to check for a \lacetree corresponding to a program execution with an exit status in $\ExitStatus$.
%
Combining \Cref{thm:lab} with the definition of {\CHC}~VI yields the main result of this section.

\begin{restatable}
{theorem}{thmExitStatusProblem} 
\label{thm:exit status problem} 
Let $\mathcal{I}$ be an instance of the exit-status problem. Then, $\mathcal{I}$ admits a positive answer if and only if the \CHC system $\Cex$ is unsatisfiable.
\end{restatable} 

\subsection{Verifying Memory Safety}
\label{sec:memory safety}
We begin by establishing two key theorems linking the exit status problem to memory safety, forming the basis for our method's correctness.

\begin{theorem}\label{thm:exitstatus-positive--memorysafety-negative}
    If the answer to the exit status problem $(P,m,n,\{ \mathbf{E} \})$ is positive,
    then the answer to the memory safety 
    for $P$ is negative.
\end{theorem}

\begin{restatable}
{theorem}{thmexitstatustomemorysafety} 
\label{thm:exitstatus-to-memorysafety}
    If the answer to the exit status problem $(P,m,n,\{ \mathbf{O}, \mathbf{M}, \mathbf{E} \})$ is negative,
    then the answer to the memory safety problem 
    for $P$ is positive.
\end{restatable}
\begin{wrapfigure}[11]{R}{7.3cm}
\vspace{-1.2cm}
\centering
\hspace*{-0.14\linewidth}
\begin{tikzpicture}[auto, node distance=1.1cm, shorten >= 1pt, shorten <= 2pt,every node/.style={font=\scriptsize}]
\node (input) at (0,0){};

\node [draw,
	fill=gray!25,
 	minimum width=4.0cm, 
	minimum height=0.5cm,
	below=0.6cm of input
]  (init) {$m\leftarrow m_0, n\leftarrow n_0$}; 

\node [draw,
	fill=gray!25,
 	minimum width=4.0cm, 
	minimum height=0.5cm,
	below=0.5cm of init
]  (error) {$\ExitStatusP{\{\mathbf{E}\}}$};

 \draw[-stealth] (init.south) -- (error.north) 
 	node[midway,right]{};

 \draw[-stealth] (error.east) -- ++ (0.65,0.0)
           node[at end,above right, xshift=-5.0pt, yshift=-2.3pt]{
           { \bf Memory}}
           node[at end,below right, xshift=-5.0pt, yshift=2.3pt]{
           { \bf unsafe}}
           node[midway, above]{\bf \Yes};

 \node [draw,
 	fill=gray!25, 
 	minimum width=4.0cm, 
 	minimum height=0.5cm,
 	below=0.5cm of error
 ] (correct) {$\ExitStatusP{\{\mathbf{M}\}}$};

 \draw[-stealth] (input.south) -- (init.north)
 	node[midway,right]{$P,m_0,n_0$};

 \draw[-stealth] (error.south) -- (correct.north) 
 	node[midway,right]{\No};

  \draw[-stealth] (correct.west) 
         -- ++ (-0.85,0) node[midway, above]{\Yes}
         node[midway, below]{$m\!\mathrel{+}\mathrel{+}$}
         -- ++ (0, 0.5) 
         |-(error.west)node[pos=0.5, above,text width=1.8cm, text centered]{};

 \node [draw,
 	fill=gray!25, 
 	minimum width=4.0cm, 
 	minimum height=0.5cm,
 	below=0.5cm of correct
 ] (nottree) {$\ExitStatusP{\{\mathbf{O}\}}$};

 \draw[-stealth] (nottree.west) 
         -- ++ (-0.85,0) node[midway, above ]{\Yes}
         node[midway, below ]{ $n\!\mathrel{+}\mathrel{+}$}
         -- ++ (0, 0.5) 
         |-(error.west)node[pos=0.5, above,text width=1.8cm, text centered]{};

 \draw[-stealth] (correct.south) -- (nottree.north) 
 	node[midway,right]{\No};


\draw[-stealth] (nottree.east) -- ++ (0.65,0.0)
           node[at end,above right, xshift=-4pt, yshift=-2pt]{
           { \bf Memory}}
           node[at end,below right, xshift=-4pt, yshift=2pt]{
           { \bf safe}}
          node[midway,above]{\No};
\end{tikzpicture}
\label{fig:MemSafe}
\end{wrapfigure}
\noindent{\bf Algorithm~\MemSafe:} We outline our algorithm on the right. 
We are given a program $P$ and initial values $m_0$ and $n_0$ for the two
parameters $m$ and $n$ of {\lacetree}s. 
Verification starts by solving the problem $\ExitStatusP{\{\mathbf{E}\}}$ \\
to detect null-pointer dereference errors. 
If the answer is positive, by  \Cref{thm:exitstatus-positive--memorysafety-negative} $P$ violates memory safety.
%
Otherwise, memory safety is not guaranteed, as the current values of $m$ and $n$ may not cover all executions. 
%
To address this, we solve $\ExitStatusP{X}$ with:
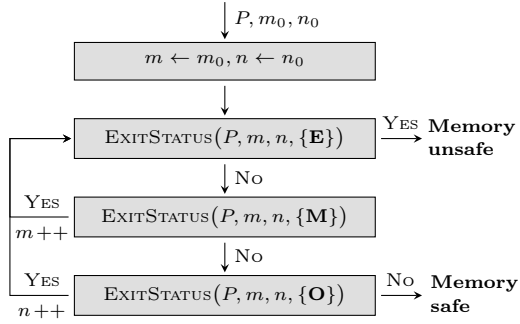
\begin{enumerate}
    \item 
    $X=\{\mathbf{M}\}$ 
    to detect out-of-memory failure from $\New$ allocations, and 
    \item 
    $X=\{\mathbf{O}\}$ to detect label overflow errors.
\end{enumerate} 
If both instances 
    are negative,  \Cref{thm:exitstatus-to-memorysafety} ensures that $P$ is memory-safe.
Otherwise, we increment $m$ or $n$ to broaden coverage and restart.
This may continue indefinitely if the semi-algorithm never terminates (due to undecidability) or no parameter values suffice to establish memory safety.    

\begin{theorem}
Algorithm~{\MemSafe} is a sound solution to the memory safety problem,
i.e., if it terminates,
it yields the correct answer. 
\end{theorem}

\subsection{Invariant Structure 
}\label{sec:invariants}
In this section, we examine the essential properties that a solution to the \CHC system presented in Fig.~\ref{fig:CHC-rules} must satisfy, with particular emphasis on the structure and complexity of the required invariants. To ground the discussion, we refer to the running example introduced in Section~\ref{sec:intro} and shown in Fig.~\ref{fig:running example},  representative of a broad class of procedures manipulating tree data structures. This example highlights both the challenges and the recurring structural patterns encountered in the synthesis of suitable invariants for \CHC systems.

The minimal model of 
$\invCex$ for logs of knitted trees (from lists containing the $\mathit{key}$ value) results in labels that can be classified as follows:
\begin{description}
    \item[Initial Node:] the first node of the lists is shaped as $u_1$  (see Fig.~\ref{fig:kt1}).
    
    \item[Intermediate Nodes:] one or more nodes like $u_2$ (depending on the length of the input list).
    
    \item[Target Node:] the node containing the special value, $u_3$.
    
    \item[Post-Target Node:] the node immediately following the special value, $u_4$.
    
    \item[Subsequent Nodes:] remaining nodes, similar to $u_5$.
\end{description}

Labels associated with nodes of the same category differ only in their data fields ($\mathit{val}$ and $\mathit{key}$), while all other fields (e.g., program counters) are equal across nodes  (for instance, nodes like $u_2$ consistently exhibit program counter $3$ in the second frame). Only the data fields
$\mathit{val}$ and $\mathit{key}$ depend on the input list, and their constraints are simple:
\begin{itemize}
    \item In each node, both $\mathit{val}$ and $\mathit{key}$ retain the same value across all frames.

    \item In an intermediate node, $\mathit{val} \neq \mathit{key}$;

    \item In the target node ($u_3$), $\mathit{val} = \mathit{key}$.
\end{itemize}
Input lists that do not contain the $\mathit{key}$ induce additional node types, subject to analogous constraints.
In Section~\ref{sec:conclusions}, we briefly outline how the regularities discussed above could be exploited 
in an implementation of our framework.



\section{Related Work}\label{sec:related_work}
Our work is related to many works in the literature in different ways. Here we focus on those that seem to be the closest to the results presented in this paper.

Our approach uses {\CHC} engines for backend analysis, similar to other verification methods~\cite{DBLP:journals/jar/ChampionCKS20,DBLP:conf/pldi/FedyukovichAB17,DBLP:journals/corr/GarocheKT16,DBLP:conf/cav/GurfinkelKKN15,DBLP:conf/fm/HojjatKGIKR12,DBLP:conf/cav/KahsaiRSS16,DBLP:conf/pldi/KobayashiSU11,DBLP:journals/toplas/MatsushitaTK21,DBLP:conf/cav/Gurfinkel22,DBLP:conf/aaai/FaellaP23,DBLP:conf/fmcad/EsenR22,DBLP:conf/esop/ItzhakySV24,DBLP:conf/ecai/FaellaP24}. {\CHC}s serve as an intermediate 
language, allowing us to focus on proof rules while solvers implement algorithms within a standard framework.
A primary challenge is encoding heap-allocated mutable data structures. While array theory is often used (e.g.,~\cite{DBLP:conf/fmcad/KomuravelliBGM15,DBLP:journals/fuin/AngelisFPP17a}), it can result in complex {\CHC}s. Our approach uses simple theories for basic data types, avoiding array theory unless necessary.
Traditional heap program analysis often relies on abstractions like shape analysis~\cite{DBLP:conf/popl/SagivRW99} to scale. Refinement types and invariants can be used to transform complex data structures, avoiding array theory (e.g.~\cite{DBLP:conf/pldi/RondonKJ08,DBLP:conf/sas/BjornerMR13,DBLP:conf/sas/MonniauxG16,DBLP:conf/lpar/KahsaiKRS17}).
This can lead to over-approximation in {\CHC}s and false positive by replacing heap operations with local object assertions, potentially missing global invariants but enabling efficient verification when local invariants suffice. A recent proposal suggests using an SMT-LIB theory of heaps for {\CHC}s to standardize heap data-structure representation~\cite{DBLP:conf/smt/EsenR22}.  

Our technique relates to tree automata, automata with auxiliary storage, and bounded tree-width graphs representing their executions. It also relates to Courcelle's theorem (proof), which reduces analysis to tree automata emptiness~\cite{DBLP:series/txtcs/FlumG06}. Inspired by Alur and Madhusudan's nested words to represent pushdown automata  executions~\cite{DBLP:journals/jacm/AlurM09}, and  their extensions for multistack and distributed automata by Madhusudan and Parlato~\cite{DBLP:conf/popl/MadhusudanP11}, 
we represent tree-manipulating program executions as {\lacetree}s, where nodes are frames and edges are $\nextfr$ and $\prevfr$ frame fields. Similar to La~Torre et al.~\cite{DBLP:conf/mfcs/TorreNP14} 
for multistack pushdown automata, our approach provides tree decompositions with a bounded tree width.
Instead of using tree automata emptiness for reachability analysis, we leverage {\CHC} solvers to enable a tree automata-like method with
enhanced data reasoning.
 Additionally, like Heizmann et al.~\cite{DBLP:conf/cav/HeizmannHP13}, we use automata for program analysis but replace counterexample-guided abstraction refinement with precise {\lacetree} representations and {\CHC} solvers for approximation and refinement. 

Our work extends decidable methods for bounded-pass heap-manipulating programs by supporting a broader range of properties, potentially at the cost of undecidability.  Mathur et al.~\cite{DBLP:journals/pacmpl/MathurMKMV20} achieve decidable memory safety for forest-like initial heaps and single-pass traversals, building on uninterpreted coherent programs ~\cite{DBLP:journals/pacmpl/MathurMV19,DBLP:conf/tacas/MathurM020}. They handle memory freeing but leave support for more complex postconditions for future work.
Alur and Černý~\cite{DBLP:conf/popl/AlurC11} reduce assertion checking of single-pass list-processing programs to data string transducers, achieving decidability with a single advancing variable. 
This approach is less flexible than Mathur et al.’s, as it doesn’t address memory safety or heap shape changes and is limited to data ordering and equality without handling explicit memory freeing.

Heap verification has been extensively studied using decidable logics such as first-order logic with reachability~\cite{DBLP:journals/corr/abs-0904-4902}, {\sc Lisbq} in the {\sc Havoc} tool~\cite{DBLP:conf/popl/LahiriQ08}, and fragments of separation logic~\cite{DBLP:conf/fsttcs/BerdineCO04,DBLP:conf/cav/PiskacWZ13,DBLP:conf/cade/EchenimIP21}. Some approaches interpret bounded tree width data structures on trees~\cite{DBLP:conf/cade/IosifRS13,DBLP:conf/popl/MadhusudanPQ11}. While these logics are often restrictive, others methods handle undecidable cases using heuristics, lemma synthesis, and programmer annotations~\cite{DBLP:conf/vstte/BanerjeeBN08,DBLP:journals/jacm/BanerjeeNR13,DBLP:conf/ecoop/BanerjeeNR08,DBLP:conf/aplas/BerdineCO05,DBLP:conf/fmco/BerdineCO05,DBLP:journals/jacm/CalcagnoDOY11,DBLP:journals/scp/ChinDNQ12,DBLP:conf/pldi/ChuJT15,DBLP:conf/tacas/DistefanoOY06,DBLP:conf/nfm/JacobsSPVPP11,DBLP:journals/pacmpl/MuraliPBLM22,DBLP:conf/cav/NguyenC08,DBLP:conf/pldi/PekQM14,DBLP:conf/cav/PiskacWZ14,DBLP:conf/fm/TaLKC16}.
In contrast, our \lacetree encoding promotes a separation of concerns, offloading the algorithmic burden to the underlying \CHC solver.

\section{Conclusions and Directions for Future Research}\label{sec:conclusions}
We presented a foundational, compositional approach to verifying programs that manipulate tree data structures. By modeling executions as \lacetree{s} and encoding them as {\CHC}s,  verification reduces to \CHC\ satisfiability. This enables modular reasoning and supports simple invariants. Overall, our method offers a uniform and scalable framework for automating the verification of a broad class of tree-manipulating programs.

\medskip
\noindent{\bf Future Work}.
Our approach opens multiple research directions.

\smallskip
\noindent{\underline{\em Efficient Implementation}.}
Labels are currently handled by a single predicate, $\invCex$. 
Performance can be improved via case splitting---encoding enumerated fields into predicate names---to simplify invariants (see \Cref{sec:invariants}). Moreover, precomputing all possible configurations of the enumerated fields arising in every possible frame of the program, together with (an overapproximation of) the within-node and across-neighbor adjacency relations between frames would produce 
a larger set of significantly simplified {\CHC}s that enforce consistency of the data with the enumerated structure of the knitted tree.

\smallskip
\noindent{\underline{\em Beyond Memory Safety}.} Full correctness requires verifying structural and functional properties. Using {\em symbolic data-tree automata} ({\SDTA}s), which integrate well with {\CHC}-based verification~\cite{DBLP:conf/cav/FaellaP22}, we can formally specify pre/postconditions--e.g., that inputs form a red-black tree and outputs a sorted list. Preconditions are easy to encode 
as they involve only the first frame of each node in the {\lacetree}, while postconditions require more effort due to the complexity of encoding the output heap within the {\lacetree} structure. For set-like trees, it is also important to verify that operations such as insertion, deletion, and search preserve key invariants, which can be checked via {\lacetree} logs. Termination can be verified by ensuring that labels do not overflow.


\smallskip
\noindent{\underline{\em Deductive Verification}.}
Our methodology is particularly suited for deductive verification of procedures with linear-time complexity—i.e., those that traverse each node a bounded number of times. We aim to develop a verification framework where program correctness is established by breaking down verification conditions into preconditions and postconditions for code segments, with each segment provably executable in linear time. This would bridge the gap between our approach and classical deductive verification techniques.

\smallskip
\noindent{\underline{\em More Structures}.}
Our approach naturally extends to programs manipulating multiple data structures, especially those with bounded treewidth. 
While there is no general method for all combinations, many can be handled with suitable encodings. For example, a program that traverses a red-black tree in-order and inserts values into a singly linked list can be modeled using knitted trees with bounded labels; our method can then verify that the output list contains the input values in non-increasing order. Some scenarios require more inventive encodings. For instance, checking equality of two lists via separate dummy roots leads to unbounded log growth. Instead, modeling both as a single list of paired elements keeps the log size bounded and tractable.

As noted in the previous section, the graphs induced by \lacetree{s} have bounded treewidth, suggesting applicability to a broad range of structures, including arrays, doubly-linked lists, trees with parent pointers, and, more generally, any structure with bounded treewidth and a canonical tree decomposition.

\smallskip
\noindent{\underline{\em Program Synthesis}.}
We also plan to explore syntax-guided synthesis (SyGuS) of tree-manipulating code. By expressing correctness properties as {\SDTA}s and reducing synthesis to {\CHC} solving, we aim to generate correct-by-construction procedures, extending extending recent work on synthesis from specifications~\cite{DBLP:journals/pacmpl/PolikarpovaS19}.

\subsubsection{Acknowledgements}
We sincerely thank the anonymous reviewers for their insightful feedback and constructive suggestions, which significantly improved the quality of this paper.
This work was partially supported by INdAM-GNCS Project 2025 - CUP: E53C24001950001, 
and the MUR projects SOP (Securing sOftware Platforms, CUP: H73C22000890001) as
part of the SERICS project (n. PE00000014 - CUP: B43C22000750006), and MUR project FAIR (Future AI Research, CUP: E63C220 02150007).

\bibliographystyle{splncs04}
\bibliography{ref2}

\begin{thebibliography}{10}
\providecommand{\url}[1]{\texttt{#1}}
\providecommand{\urlprefix}{URL }
\providecommand{\doi}[1]{https://doi.org/#1}

\bibitem{DBLP:conf/popl/AlurC11}
Alur, R., Cern{\'{y}}, P.: Streaming transducers for algorithmic verification
  of single-pass list-processing programs. In: Ball, T., Sagiv, M. (eds.)
  Proceedings of the 38th {ACM} {SIGPLAN-SIGACT} Symposium on Principles of
  Programming Languages, {POPL} 2011, Austin, TX, USA, January 26-28, 2011. pp.
  599--610. {ACM} (2011). \doi{10.1145/1926385.1926454}

\bibitem{DBLP:journals/jacm/AlurM09}
Alur, R., Madhusudan, P.: Adding nesting structure to words. J. {ACM}
  \textbf{56}(3),  16:1--16:43 (2009). \doi{10.1145/1516512.1516518}

\bibitem{De_Angelis_2022}
Angelis, E.D., K, H.G.V.: {CHC}-{COMP} 2022: Competition report. Electronic
  Proceedings in Theoretical Computer Science  \textbf{373},  44--62 (nov
  2022). \doi{10.4204/eptcs.373.5}

\bibitem{DBLP:conf/vstte/BanerjeeBN08}
Banerjee, A., Barnett, M., Naumann, D.A.: Boogie meets regions: {A}
  verification experience report. In: Shankar, N., Woodcock, J. (eds.) Verified
  Software: Theories, Tools, Experiments, Second International Conference,
  {VSTTE} 2008, Toronto, Canada, October 6-9, 2008. Proceedings. Lecture Notes
  in Computer Science, vol.~5295, pp. 177--191. Springer (2008).
  \doi{10.1007/978-3-540-87873-5\_16}

\bibitem{DBLP:conf/ecoop/BanerjeeNR08}
Banerjee, A., Naumann, D.A., Rosenberg, S.: Regional logic for local reasoning
  about global invariants. In: Vitek, J. (ed.) {ECOOP} 2008 - Object-Oriented
  Programming, 22nd European Conference, Paphos, Cyprus, July 7-11, 2008,
  Proceedings. Lecture Notes in Computer Science, vol.~5142, pp. 387--411.
  Springer (2008). \doi{10.1007/978-3-540-70592-5\_17}

\bibitem{DBLP:journals/jacm/BanerjeeNR13}
Banerjee, A., Naumann, D.A., Rosenberg, S.: Local reasoning for global
  invariants, part {I:} region logic. J. {ACM}  \textbf{60}(3),  18:1--18:56
  (2013). \doi{10.1145/2485982}

\bibitem{DBLP:conf/fsttcs/BerdineCO04}
Berdine, J., Calcagno, C., O'Hearn, P.W.: A decidable fragment of separation
  logic. In: Lodaya, K., Mahajan, M. (eds.) {FSTTCS} 2004: Foundations of
  Software Technology and Theoretical Computer Science, 24th International
  Conference, Chennai, India, December 16-18, 2004, Proceedings. Lecture Notes
  in Computer Science, vol.~3328, pp. 97--109. Springer (2004).
  \doi{10.1007/978-3-540-30538-5\_9}

\bibitem{DBLP:conf/fmco/BerdineCO05}
Berdine, J., Calcagno, C., O'Hearn, P.W.: Smallfoot: Modular automatic
  assertion checking with separation logic. In: de~Boer, F.S., Bonsangue, M.M.,
  Graf, S., de~Roever, W.P. (eds.) Formal Methods for Components and Objects,
  4th International Symposium, {FMCO} 2005, Amsterdam, The Netherlands,
  November 1-4, 2005, Revised Lectures. Lecture Notes in Computer Science,
  vol.~4111, pp. 115--137. Springer (2005). \doi{10.1007/11804192\_6}

\bibitem{DBLP:conf/aplas/BerdineCO05}
Berdine, J., Calcagno, C., O'Hearn, P.W.: Symbolic execution with separation
  logic. In: Yi, K. (ed.) Programming Languages and Systems, Third Asian
  Symposium, {APLAS} 2005, Tsukuba, Japan, November 2-5, 2005, Proceedings.
  Lecture Notes in Computer Science, vol.~3780, pp. 52--68. Springer (2005).
  \doi{10.1007/11575467\_5}

\bibitem{DBLP:conf/sas/BjornerMR13}
Bj{\o}rner, N.S., McMillan, K.L., Rybalchenko, A.: On solving universally
  quantified horn clauses. In: Logozzo, F., F{\"{a}}hndrich, M. (eds.) Static
  Analysis - 20th International Symposium, {SAS} 2013, Seattle, WA, USA, June
  20-22, 2013. Proceedings. Lecture Notes in Computer Science, vol.~7935, pp.
  105--125. Springer (2013). \doi{10.1007/978-3-642-38856-9\_8}

\bibitem{DBLP:journals/jacm/CalcagnoDOY11}
Calcagno, C., Distefano, D., O'Hearn, P.W., Yang, H.: Compositional shape
  analysis by means of bi-abduction. J. {ACM}  \textbf{58}(6),  26:1--26:66
  (2011). \doi{10.1145/2049697.2049700}

\bibitem{DBLP:journals/jar/ChampionCKS20}
Champion, A., Chiba, T., Kobayashi, N., Sato, R.: Ice-based refinement type
  discovery for higher-order functional programs. J. Autom. Reason.
  \textbf{64}(7),  1393--1418 (2020)

\bibitem{DBLP:journals/scp/ChinDNQ12}
Chin, W., David, C., Nguyen, H.H., Qin, S.: Automated verification of shape,
  size and bag properties via user-defined predicates in separation logic. Sci.
  Comput. Program.  \textbf{77}(9),  1006--1036 (2012).
  \doi{10.1016/J.SCICO.2010.07.004}

\bibitem{DBLP:conf/pldi/ChuJT15}
Chu, D., Jaffar, J., Trinh, M.: Automatic induction proofs of data-structures
  in imperative programs. In: Grove, D., Blackburn, S.M. (eds.) Proceedings of
  the 36th {ACM} {SIGPLAN} Conference on Programming Language Design and
  Implementation, Portland, OR, USA, June 15-17, 2015. pp. 457--466. {ACM}
  (2015). \doi{10.1145/2737924.2737984}

\bibitem{DBLP:journals/fuin/AngelisFPP17a}
{De Angelis}, E., Fioravanti, F., Pettorossi, A., Proietti, M.: Program
  verification using constraint handling rules and array constraint
  generalizations. Fundam. Informaticae  \textbf{150}(1),  73--117 (2017).
  \doi{10.3233/FI-2017-1461}

\bibitem{DBLP:journals/corr/abs-2404-14923}
{De Angelis}, E., K., H.G.V.: {CHC-COMP} 2023: Competition report. CoRR
  \textbf{abs/2404.14923} (2024). \doi{10.48550/ARXIV.2404.14923}

\bibitem{DBLP:conf/tacas/DistefanoOY06}
Distefano, D., O'Hearn, P.W., Yang, H.: A local shape analysis based on
  separation logic. In: Hermanns, H., Palsberg, J. (eds.) Tools and Algorithms
  for the Construction and Analysis of Systems, 12th International Conference,
  {TACAS} 2006 Held as Part of the Joint European Conferences on Theory and
  Practice of Software, {ETAPS} 2006, Vienna, Austria, March 25 - April 2,
  2006, Proceedings. Lecture Notes in Computer Science, vol.~3920, pp.
  287--302. Springer (2006). \doi{10.1007/11691372\_19}

\bibitem{DBLP:conf/cade/EchenimIP21}
Echenim, M., Iosif, R., Peltier, N.: Unifying decidable entailments in
  separation logic with inductive definitions. In: Platzer, A., Sutcliffe, G.
  (eds.) Automated Deduction - {CADE} 28 - 28th International Conference on
  Automated Deduction, Virtual Event, July 12-15, 2021, Proceedings. Lecture
  Notes in Computer Science, vol. 12699, pp. 183--199. Springer (2021).
  \doi{10.1007/978-3-030-79876-5\_11}

\bibitem{DBLP:journals/jacm/EmdenK76}
van Emden, M.H., Kowalski, R.A.: The semantics of predicate logic as a
  programming language. J. {ACM}  \textbf{23}(4),  733--742 (1976)

\bibitem{DBLP:conf/smt/EsenR22}
Esen, Z., R{\"{u}}mmer, P.: An {SMT-LIB} theory of heaps. In: D{\'{e}}harbe,
  D., Hyv{\"{a}}rinen, A.E.J. (eds.) Proceedings of the 20th Internal Workshop
  on Satisfiability Modulo Theories co-located with the 11th International
  Joint Conference on Automated Reasoning {(IJCAR} 2022) part of the 8th
  Federated Logic Conference (FLoC 2022), Haifa, Israel, August 11-12, 2022.
  {CEUR} Workshop Proceedings, vol.~3185, pp. 38--53. CEUR-WS.org (2022)

\bibitem{DBLP:conf/fmcad/EsenR22}
Esen, Z., R{\"{u}}mmer, P.: Tricera: Verifying {C} programs using the theory of
  heaps. In: Griggio, A., Rungta, N. (eds.) 22nd Formal Methods in
  Computer-Aided Design, {FMCAD} 2022, Trento, Italy, October 17-21, 2022. pp.
  380--391. {IEEE} (2022). \doi{10.34727/2022/ISBN.978-3-85448-053-2\_45}

\bibitem{DBLP:conf/cav/FaellaP22}
Faella, M., Parlato, G.: Reasoning about data trees using chcs. In: Shoham, S.,
  Vizel, Y. (eds.) Computer Aided Verification - 34th International Conference,
  {CAV} 2022, Haifa, Israel, August 7-10, 2022, Proceedings, Part {II}. Lecture
  Notes in Computer Science, vol. 13372, pp. 249--271. Springer (2022).
  \doi{10.1007/978-3-031-13188-2\_13}

\bibitem{DBLP:conf/aaai/FaellaP23}
Faella, M., Parlato, G.: Reachability games modulo theories with a bounded
  safety player. In: Williams, B., Chen, Y., Neville, J. (eds.) Thirty-Seventh
  {AAAI} Conference on Artificial Intelligence, {AAAI} 2023, Thirty-Fifth
  Conference on Innovative Applications of Artificial Intelligence, {IAAI}
  2023, Thirteenth Symposium on Educational Advances in Artificial
  Intelligence, {EAAI} 2023, Washington, DC, USA, February 7-14, 2023. pp.
  6330--6337. {AAAI} Press (2023). \doi{10.1609/AAAI.V37I5.25779}

\bibitem{DBLP:conf/ecai/FaellaP24}
Faella, M., Parlato, G.: A unified automata-theoretic approach to
  ltl\({}_{\mbox{f}}\) modulo theories. In: Endriss, U., Melo, F.S., Bach, K.,
  Diz, A.J.B., Alonso{-}Moral, J.M., Barro, S., Heintz, F. (eds.) {ECAI} 2024 -
  27th European Conference on Artificial Intelligence, 19-24 October 2024,
  Santiago de Compostela, Spain - Including 13th Conference on Prestigious
  Applications of Intelligent Systems {(PAIS} 2024). Frontiers in Artificial
  Intelligence and Applications, vol.~392, pp. 1254--1261. {IOS} Press (2024).
  \doi{10.3233/FAIA240622}

\bibitem{DBLP:conf/pldi/FedyukovichAB17}
Fedyukovich, G., Ahmad, M.B.S., Bod{\'{\i}}k, R.: Gradual synthesis for static
  parallelization of single-pass array-processing programs. In: Proceedings of
  the 38th {ACM} {SIGPLAN} Conference on Programming Language Design and
  Implementation, {PLDI} 2017, Barcelona, Spain, June 18-23, 2017. pp.
  572--585. {ACM} (2017)

\bibitem{DBLP:series/txtcs/FlumG06}
Flum, J., Grohe, M.: Parameterized Complexity Theory. Texts in Theoretical
  Computer Science. An {EATCS} Series, Springer (2006).
  \doi{10.1007/3-540-29953-X}

\bibitem{DBLP:journals/corr/GarocheKT16}
Garoche, P., Kahsai, T., Thirioux, X.: Hierarchical state machines as modular
  horn clauses. In: Proceedings 3rd Workshop on Horn Clauses for Verification
  and Synthesis, HCVS@ETAPS 2016, Eindhoven, The Netherlands, 3rd April 2016.
  {EPTCS}, vol.~219, pp. 15--28 (2016)

\bibitem{DBLP:conf/cav/Gurfinkel22}
Gurfinkel, A.: Program verification with constrained horn clauses (invited
  paper). In: Shoham, S., Vizel, Y. (eds.) Computer Aided Verification - 34th
  International Conference, {CAV} 2022, Haifa, Israel, August 7-10, 2022,
  Proceedings, Part {I}. Lecture Notes in Computer Science, vol. 13371, pp.
  19--29. Springer (2022). \doi{10.1007/978-3-031-13185-1\_2}

\bibitem{DBLP:conf/cav/GurfinkelKKN15}
Gurfinkel, A., Kahsai, T., Komuravelli, A., Navas, J.A.: The seahorn
  verification framework. In: Computer Aided Verification - 27th International
  Conference, {CAV} 2015, San Francisco, CA, USA, July 18-24, 2015,
  Proceedings, Part {I}. LNCS, vol.~9206, pp. 343--361. Springer (2015)

\bibitem{DBLP:conf/cav/HeizmannHP13}
Heizmann, M., Hoenicke, J., Podelski, A.: Software model checking for people
  who love automata. In: Sharygina, N., Veith, H. (eds.) Computer Aided
  Verification - 25th International Conference, {CAV} 2013, Saint Petersburg,
  Russia, July 13-19, 2013. Proceedings. Lecture Notes in Computer Science,
  vol.~8044, pp. 36--52. Springer (2013). \doi{10.1007/978-3-642-39799-8\_2}

\bibitem{DBLP:conf/fm/HojjatKGIKR12}
Hojjat, H., Konecn{\'{y}}, F., Garnier, F., Iosif, R., Kuncak, V.,
  R{\"{u}}mmer, P.: A verification toolkit for numerical transition systems -
  tool paper. In: {FM} 2012: Formal Methods - 18th International Symposium,
  Paris, France, August 27-31, 2012. Proceedings. LNCS, vol.~7436, pp.
  247--251. Springer (2012)

\bibitem{DBLP:conf/cade/IosifRS13}
Iosif, R., Rogalewicz, A., Sim{\'{a}}cek, J.: The tree width of separation
  logic with recursive definitions. In: Bonacina, M.P. (ed.) Automated
  Deduction - {CADE-24} - 24th International Conference on Automated Deduction,
  Lake Placid, NY, USA, June 9-14, 2013. Proceedings. Lecture Notes in Computer
  Science, vol.~7898, pp. 21--38. Springer (2013).
  \doi{10.1007/978-3-642-38574-2\_2}

\bibitem{DBLP:conf/esop/ItzhakySV24}
Itzhaky, S., Shoham, S., Vizel, Y.: Hyperproperty verification as {CHC}
  satisfiability. In: Weirich, S. (ed.) Programming Languages and Systems -
  33rd European Symposium on Programming, {ESOP} 2024, Held as Part of the
  European Joint Conferences on Theory and Practice of Software, {ETAPS} 2024,
  Luxembourg City, Luxembourg, April 6-11, 2024, Proceedings, Part {II}.
  Lecture Notes in Computer Science, vol. 14577, pp. 212--241. Springer (2024).
  \doi{10.1007/978-3-031-57267-8\_9}

\bibitem{DBLP:conf/nfm/JacobsSPVPP11}
Jacobs, B., Smans, J., Philippaerts, P., Vogels, F., Penninckx, W., Piessens,
  F.: Verifast: {A} powerful, sound, predictable, fast verifier for {C} and
  java. In: Bobaru, M.G., Havelund, K., Holzmann, G.J., Joshi, R. (eds.) {NASA}
  Formal Methods - Third International Symposium, {NFM} 2011, Pasadena, CA,
  USA, April 18-20, 2011. Proceedings. Lecture Notes in Computer Science,
  vol.~6617, pp. 41--55. Springer (2011). \doi{10.1007/978-3-642-20398-5\_4}

\bibitem{DBLP:journals/jlp/JaffarM94}
Jaffar, J., Maher, M.J.: Constraint logic programming: {A} survey. J. Log.
  Program.  \textbf{19/20},  503--581 (1994)

\bibitem{DBLP:conf/lpar/KahsaiKRS17}
Kahsai, T., Kersten, R., R{\"{u}}mmer, P., Sch{\"{a}}f, M.: Quantified heap
  invariants for object-oriented programs. In: Eiter, T., Sands, D. (eds.)
  LPAR-21, 21st International Conference on Logic for Programming, Artificial
  Intelligence and Reasoning, Maun, Botswana, May 7-12, 2017. EPiC Series in
  Computing, vol.~46, pp. 368--384. EasyChair (2017). \doi{10.29007/ZRCT}

\bibitem{DBLP:conf/cav/KahsaiRSS16}
Kahsai, T., R{\"{u}}mmer, P., Sanchez, H., Sch{\"{a}}f, M.: Jayhorn: {A}
  framework for verifying java programs. In: Computer Aided Verification - 28th
  International Conference, {CAV} 2016, Toronto, ON, Canada, July 17-23, 2016,
  Proceedings, Part {I}. LNCS, vol.~9779, pp. 352--358. Springer (2016)

\bibitem{DBLP:conf/pldi/KobayashiSU11}
Kobayashi, N., Sato, R., Unno, H.: Predicate abstraction and {CEGAR} for
  higher-order model checking. In: Proceedings of the 32nd {ACM} {SIGPLAN}
  Conference on Programming Language Design and Implementation, {PLDI} 2011,
  San Jose, CA, USA, June 4-8, 2011. pp. 222--233. {ACM} (2011)

\bibitem{DBLP:conf/fmcad/KomuravelliBGM15}
Komuravelli, A., Bj{\o}rner, N.S., Gurfinkel, A., McMillan, K.L.: Compositional
  verification of procedural programs using horn clauses over integers and
  arrays. In: Kaivola, R., Wahl, T. (eds.) Formal Methods in Computer-Aided
  Design, {FMCAD} 2015, Austin, Texas, USA, September 27-30, 2015. pp. 89--96.
  {IEEE} (2015). \doi{10.1109/FMCAD.2015.7542257}

\bibitem{DBLP:conf/mfcs/TorreNP14}
{La Torre}, S., Napoli, M., Parlato, G.: A unifying approach for multistack
  pushdown automata. In: Csuhaj{-}Varj{\'{u}}, E., Dietzfelbinger, M.,
  {\'{E}}sik, Z. (eds.) Mathematical Foundations of Computer Science 2014 -
  39th International Symposium, {MFCS} 2014, Budapest, Hungary, August 25-29,
  2014. Proceedings, Part {I}. Lecture Notes in Computer Science, vol.~8634,
  pp. 377--389. Springer (2014). \doi{10.1007/978-3-662-44522-8\_32}

\bibitem{DBLP:conf/popl/LahiriQ08}
Lahiri, S.K., Qadeer, S.: Back to the future: revisiting precise program
  verification using {SMT} solvers. In: Necula, G.C., Wadler, P. (eds.)
  Proceedings of the 35th {ACM} {SIGPLAN-SIGACT} Symposium on Principles of
  Programming Languages, {POPL} 2008, San Francisco, California, USA, January
  7-12, 2008. pp. 171--182. {ACM} (2008). \doi{10.1145/1328438.1328461}

\bibitem{DBLP:journals/corr/abs-0904-4902}
Lev{-}Ami, T., Immerman, N., Reps, T.W., Sagiv, M., Srivastava, S., Yorsh, G.:
  Simulating reachability using first-order logic with applications to
  verification of linked data structures. Log. Methods Comput. Sci.
  \textbf{5}(2) (2009)

\bibitem{DBLP:conf/popl/MadhusudanP11}
Madhusudan, P., Parlato, G.: The tree width of auxiliary storage. In: Ball, T.,
  Sagiv, M. (eds.) Proceedings of the 38th {ACM} {SIGPLAN-SIGACT} Symposium on
  Principles of Programming Languages, {POPL} 2011, Austin, TX, USA, January
  26-28, 2011. pp. 283--294. {ACM} (2011). \doi{10.1145/1926385.1926419}

\bibitem{DBLP:conf/popl/MadhusudanPQ11}
Madhusudan, P., Parlato, G., Qiu, X.: Decidable logics combining heap
  structures and data. In: Ball, T., Sagiv, M. (eds.) Proceedings of the 38th
  {ACM} {SIGPLAN-SIGACT} Symposium on Principles of Programming Languages,
  {POPL} 2011, Austin, TX, USA, January 26-28, 2011. pp. 611--622. {ACM}
  (2011). \doi{10.1145/1926385.1926455}

\bibitem{DBLP:conf/unu/MannaZ02}
Manna, Z., Zarba, C.G.: Combining decision procedures. In: Formal Methods at
  the Crossroads. From Panacea to Foundational Support, 10th Anniversary
  Colloquium of UNU/IIST, the International Institute for Software Technology
  of The United Nations University, Lisbon, Portugal, March 18-20, 2002,
  Revised Papers. LNCS, vol.~2757, pp. 381--422. Springer (2002)

\bibitem{DBLP:journals/pacmpl/MathurMV19}
Mathur, U., Madhusudan, P., Viswanathan, M.: Decidable verification of
  uninterpreted programs. Proc. {ACM} Program. Lang.  \textbf{3}({POPL}),
  46:1--46:29 (2019). \doi{10.1145/3290359}

\bibitem{DBLP:conf/tacas/MathurM020}
Mathur, U., Madhusudan, P., Viswanathan, M.: What's decidable about program
  verification modulo axioms? In: Biere, A., Parker, D. (eds.) Tools and
  Algorithms for the Construction and Analysis of Systems - 26th International
  Conference, {TACAS} 2020, Held as Part of the European Joint Conferences on
  Theory and Practice of Software, {ETAPS} 2020, Dublin, Ireland, April 25-30,
  2020, Proceedings, Part {II}. Lecture Notes in Computer Science, vol. 12079,
  pp. 158--177. Springer (2020). \doi{10.1007/978-3-030-45237-7\_10}

\bibitem{DBLP:journals/pacmpl/MathurMKMV20}
Mathur, U., Murali, A., Krogmeier, P., Madhusudan, P., Viswanathan, M.:
  Deciding memory safety for single-pass heap-manipulating programs. Proc.
  {ACM} Program. Lang.  \textbf{4}({POPL}),  35:1--35:29 (2020).
  \doi{10.1145/3371103}

\bibitem{DBLP:journals/toplas/MatsushitaTK21}
Matsushita, Y., Tsukada, T., Kobayashi, N.: Rusthorn: Chc-based verification
  for rust programs. {ACM} Trans. Program. Lang. Syst.  \textbf{43}(4),
  15:1--15:54 (2021). \doi{10.1145/3462205}

\bibitem{DBLP:conf/sas/MonniauxG16}
Monniaux, D., Gonnord, L.: Cell morphing: From array programs to array-free
  horn clauses. In: Rival, X. (ed.) Static Analysis - 23rd International
  Symposium, {SAS} 2016, Edinburgh, UK, September 8-10, 2016, Proceedings.
  Lecture Notes in Computer Science, vol.~9837, pp. 361--382. Springer (2016).
  \doi{10.1007/978-3-662-53413-7\_18}

\bibitem{DBLP:journals/pacmpl/MuraliPBLM22}
Murali, A., Pe{\~{n}}a, L., Blanchard, E., L{\"{o}}ding, C., Madhusudan, P.:
  Model-guided synthesis of inductive lemmas for {FOL} with least fixpoints.
  Proc. {ACM} Program. Lang.  \textbf{6}({OOPSLA2}),  1873--1902 (2022).
  \doi{10.1145/3563354}

\bibitem{DBLP:conf/cav/NguyenC08}
Nguyen, H.H., Chin, W.: Enhancing program verification with lemmas. In: Gupta,
  A., Malik, S. (eds.) Computer Aided Verification, 20th International
  Conference, {CAV} 2008, Princeton, NJ, USA, July 7-14, 2008, Proceedings.
  Lecture Notes in Computer Science, vol.~5123, pp. 355--369. Springer (2008).
  \doi{10.1007/978-3-540-70545-1\_34}

\bibitem{DBLP:conf/pldi/PekQM14}
Pek, E., Qiu, X., Madhusudan, P.: Natural proofs for data structure
  manipulation in {C} using separation logic. In: O'Boyle, M.F.P., Pingali, K.
  (eds.) {ACM} {SIGPLAN} Conference on Programming Language Design and
  Implementation, {PLDI} '14, Edinburgh, United Kingdom - June 09 - 11, 2014.
  pp. 440--451. {ACM} (2014). \doi{10.1145/2594291.2594325}

\bibitem{DBLP:conf/cav/PiskacWZ13}
Piskac, R., Wies, T., Zufferey, D.: Automating separation logic using {SMT}.
  In: Sharygina, N., Veith, H. (eds.) Computer Aided Verification - 25th
  International Conference, {CAV} 2013, Saint Petersburg, Russia, July 13-19,
  2013. Proceedings. Lecture Notes in Computer Science, vol.~8044, pp.
  773--789. Springer (2013). \doi{10.1007/978-3-642-39799-8\_54}

\bibitem{DBLP:conf/cav/PiskacWZ14}
Piskac, R., Wies, T., Zufferey, D.: Automating separation logic with trees and
  data. In: Biere, A., Bloem, R. (eds.) Computer Aided Verification - 26th
  International Conference, {CAV} 2014, Held as Part of the Vienna Summer of
  Logic, {VSL} 2014, Vienna, Austria, July 18-22, 2014. Proceedings. Lecture
  Notes in Computer Science, vol.~8559, pp. 711--728. Springer (2014).
  \doi{10.1007/978-3-319-08867-9\_47}

\bibitem{DBLP:journals/pacmpl/PolikarpovaS19}
Polikarpova, N., Sergey, I.: Structuring the synthesis of heap-manipulating
  programs. Proc. {ACM} Program. Lang.  \textbf{3}({POPL}),  72:1--72:30
  (2019). \doi{10.1145/3290385}

\bibitem{DBLP:conf/pldi/RondonKJ08}
Rondon, P.M., Kawaguchi, M., Jhala, R.: Liquid types. In: Gupta, R.,
  Amarasinghe, S.P. (eds.) Proceedings of the {ACM} {SIGPLAN} 2008 Conference
  on Programming Language Design and Implementation, Tucson, AZ, USA, June
  7-13, 2008. pp. 159--169. {ACM} (2008). \doi{10.1145/1375581.1375602}

\bibitem{DBLP:conf/popl/SagivRW99}
Sagiv, S., Reps, T.W., Wilhelm, R.: Parametric shape analysis via 3-valued
  logic. In: Appel, A.W., Aiken, A. (eds.) {POPL} '99, Proceedings of the 26th
  {ACM} {SIGPLAN-SIGACT} Symposium on Principles of Programming Languages, San
  Antonio, TX, USA, January 20-22, 1999. pp. 105--118. {ACM} (1999).
  \doi{10.1145/292540.292552}

\bibitem{DBLP:conf/fm/TaLKC16}
Ta, Q., Le, T.C., Khoo, S., Chin, W.: Automated mutual explicit induction proof
  in separation logic. In: Fitzgerald, J.S., Heitmeyer, C.L., Gnesi, S.,
  Philippou, A. (eds.) {FM} 2016: Formal Methods - 21st International
  Symposium, Limassol, Cyprus, November 9-11, 2016, Proceedings. Lecture Notes
  in Computer Science, vol.~9995, pp. 659--676 (2016).
  \doi{10.1007/978-3-319-48989-6\_40}

\end{thebibliography}

\newpage
\appendix

\section{Encoding Restricted Recursion Mechanisms in the Programming Language}\label{app:recursion}
In Section~\ref{sec:programming language}, we introduced a simple programming language for heap-manipulating programs that does not directly support procedure calls. For non-recursive procedures, we can handle calls by inlining their code at each call site. In this section, we describe an approach to simulating a specific form of recursive procedures, commonly used in procedures that manipulate tree data structures, as found in the literature. 

We focus on procedures where the first parameter (later denoted by $\mathit{first}$)
is a pointer. 
The corresponding argument that is passed at runtime must satisfy the following conditions:
 \begin{enumerate} 
    \item It references a node that has not been previously used for invocation. 
    \item It corresponds to a pointer field of the current node (i.e., the node on which the procedure was last invoked). 
    \item Its value is unchanged from its original value at the beginning of the program. 
 \end{enumerate} 
The remaining parameters can be any data values or pointer variables. Additionally, if the procedure is invoked with a $\NIL$ pointer as the first parameter, it must not call another procedure before returning. The code handling this case must be guarded by an $\mathbf{if}$ statement that checks whether the first argument of the function is equal to or different from $\NIL$. These constraints are typically satisfied by
functions that manipulate Binary Search Trees ({\BST}s), red-black trees, AVL trees, and similar structures.

We assume that procedures do not return data or pointer values; instead, any return values are managed via global variables. We also assume that the $\mathbf{main}$ (i.e., outmost) procedure is never called.  

Under these constraints, recursive programs can be transformed into a program without procedures by extending the input tree's structure to include additional fields for local data and pointer variables used within the procedures. We also include additional fields to orchestrate the simulation of the control flow during calls and returns, thereby eliminating the need for a system stack. 
To formalize this approach, we define an extended heap as follows:

\begin{definition} An {\bf extended heap} is a tuple $\mathcal{H}_{\mathit{ext}} = (N, \dsig, \data, \pointers, \dsig_{\mathit{ext}},$ $ \data_{\mathit{ext}}, \pointers_{\mathit{ext}})$ where 
$(N, \dsig, \data,  \pointers)$ and $(N, \dsig_{\mathit{ext}}, \data_{\mathit{ext}},  \pointers_{\mathit{ext}})$ are heaps.
\end{definition}

\noindent
The extended data signature $\dsig_{\mathit{ext}}$  includes: 
\begin{description}
    \item[Local Variable Fields:] Fields for the distinct local variables of the procedures (including the procedure arguments). The types of these fields match the types of the corresponding local variables. 
    \item[Program Counter:] A field named $\mathit{return\_lab}$ to keep track of the next statement to execute within a procedure after returning  from the last ongoing procedure call.
    \item[Parent Field:] A pointer field that points to the parent node in the input tree, facilitating the return control. This field is initialized accordingly and is never modified in the translated code.
\end{description}


\subsection{Code to Code Translation}
Here, we outline a method to rewrite the code and entirely eliminate recursive procedures, resulting in a program that behaves equivalently to the original.

\paragraph{\bf Initialize the local variables of the root:}
We add a fresh assignment to set the local variables related the main procedure to point at $\rootptr$. This establishes the property that the first argument of the procedure refers to the node on which it is invoked. For other nodes, this assignment is performed before the simulation starts. 

\paragraph{\bf Local variable replacement:} Within a procedure's body, replace all occurrences of local variables as follows. If $\mathit{loc}$ is a local variable, replace it with $\mathit{first}\rightarrow\mathit{loc}$. This modification is applied throughout the procedure code except for the code block executed when the procedure is called with $\NIL$ as the first parameter. This ensures that local variables correctly reference the appropriate heap node.

\paragraph{\bf Consolidate into a single procedure:}
Combine all existing functions into a single procedure by removing individual function declarations. We assume that all statement labels, variable names, and function parameter names are distinct to avoid naming conflicts.

\paragraph{\bf Replace procedure calls:} This process is divided into four phases. 
Consider a procedure call statement, and assume that the first parameter of the calling procedure is $\mathit{first}$, whereas the arguments of this call are $\mathit{arg}_1,\ldots,\mathit{arg}_\ell$.
Recall that, according to Condition~2 at the beginning of this section,
$\mathit{arg}_1 = \field{first}{\pfield}$, for some pointer field $\pfield$.
We then perform the following steps, that add new statements in place of the
procedure call.

\paragraph{Phase 1 – Argument Assignment:}
Assign the values of the arguments of the called procedure to the
corresponding fields of the node referenced by $\mathit{arg}_1$.
Specifically, let $\mathit{param}_1,\ldots,\mathit{param}_\ell$ be the list of parameters of the callee, for each $i \in [\ell]$, add an assignment of the form:
$$\mathit{arg}_1 \rightarrow \mathit{param}_i \coloneq \mathit{arg}_i.$$  
\paragraph{Phase 2 – Set Return Label Assignment:} Assign the label of the statement immediately following the current function call to $\mathit{arg}_1\rightarrow\mathit{return\_lab}$. This sets up the point to which the program should jump after the simulated procedure execution corresponding to the current call:
$$\mathit{arg}_1\rightarrow\mathit{return\_lab} \coloneq \langle \text{label of the next statement in the orginal code} \rangle$$

\paragraph{Phase 3 – Jump to the First Statement of the Called Procedure:}
Transfer control in the monolithic program to the first statement of the called procedure via a $\GOTO$ statement:

$$\GOTO~\langle\text{label of the first statement of the called function}\rangle.$$

\paragraph{Phase 4 – Jump Back to the Caller Procedure:}
Replace $\mathbf{return}$ statements in the original code with a $\GOTO$ statement that uses the label stored in the auxiliary field of the current node set at the time of procedure invocation, namely:
$$\GOTO~ \mathit{first}\rightarrow\mathit{return\_lab}.$$

This systematic transformation effectively simulates recursive procedure calls without using actual function calls or a system stack. By extending the heap with additional fields to store local variables, return addresses, and parent pointers, the program maintains the necessary state to manage control flow and variable references during execution.

\section{Additional Proofs}

\subsection{Proofs for~\Cref{sec:lace}} \label{app:proof4}

\thmcomputable*
\begin{proof}
The first part of the statement (i.e., the computability of $\KT$) can be shown by inspecting the encoding described in Section~\ref{sec:lace}, with further details in Appendix~\ref{app:encoding}. One can clearly see that the definition of \lacetree is entirely constructive.
In fact, the content of each new frame in the lace can be computed from the program and the log of the current node and
of the neighboring nodes.

The second part of the statement requires a function $\mathit{exec}$ that maps \lacetree{s} back to executions. To this aim, we first prove two supporting claims.

Claim 1: We can retrieve the state of the heap from any {\lacetree} prefix.
We show that we can recover the state of the heap at each point of the lace of a \lacetree.
We select the nodes where $\activefield$ is $\true$ in their top frame. We then examine the logs of the selected nodes to determine the values of their fields. 
For any node $u$, the current value of its data field $\datafield$ is the value of the
$\datafield$ field in the topmost frame of the log.
The current value of a pointer field $\pfield$ (assumed to be the $j$\textsuperscript{th} field in $\pointers$) can be recovered as follows: 
\begin{itemize}
\item If $\pfield$ has never been updated, $u.\pfield$ points to the $j$\textsuperscript{th} child of $u$, if $u.j$ exists and is active in the \lacetree; otherwise, $u.\pfield$ points to $\NIL$.
\item If the last update in $u$'s log is $\symb{\pfield \coloneqq \NIL}$, then that pointer field is also $\NIL$. 
\item Otherwise, the last update is $\symb{\pfield \coloneqq p}$ for some pointer variable $p$. 
We navigate the lace backward using $\FindPtr$ until finding a frame that reports $\symb{p \coloneqq \here}$, say in the log of node $v$, and set $u.\pfield$ to $v$.
\end{itemize}
We denote by $\out(\ktree)$ the heap corresponding
to the \lacetree (prefix) $\ktree$, as described above.

Claim 2: We can retrieve the value of all pointer variables from any \lacetree prefix.
The $\isnil$ fields in the last frame of the \lacetree prefix indicate which pointer variables are $\NIL$.
We can find the current value of the other pointer variables $p \in \PV$, by retracing the lace backward until the most recent assignment to it (i.e., $\symb{p \coloneq \here}$).
We denote by $\mathit{pv}(\ktree)$ the map that assigns each $p \in \PV$ to the node of $\out(\ktree)$
pointed to by $p$ at the end of the lace of $\ktree$. 

We can now prove the second part of the statement.
    Let $\ktree = (K, \mu)$ be a data tree that is the \lacetree of some execution.
    To define $\mathit{exec}(\ktree)$, we start by recovering the initial program configuration
    $c_0 = (\mathcal{H}_0,\nu_p,\nu_d,\pc_0)$, where the initial heap $\mathcal{H}_0$ 
    can be obtained from Claim~1 above, and the initial value of the pointer variables $\nu_p$
    from Claim~2. The initial value of the program counter is simply the label of the first
    line in the program,
    and the initial value for the data variables $\nu_d$ can be read from the
    $d$ fields in the second frame of the root of $\ktree$.
    To obtain each subsequent program configuration, it is sufficient to scan the lace chronologically, iteratively identify the frame $f$ that marks the end of each statement\footnote{Recall that the execution of a statement may
    require multiple frames in the lace, especially if a rewind is involved.},
    and at that point build a new program configuration using Claims~1 and~2 above,
    and reading the current value of the data variables and program counter from $f$.
    %
\qed
\end{proof}

\thmktexits*
\begin{proof}
{\bf Claim~\ref{prop:tree-error}.}
To prove the first statement, assume that $\ktree$ has exit status $\bf{E}$,
occurring when the last frame $f$ in $\ktree$'s lace has the event $\Error$.
By Thm~\ref{thm:computable}, the last configuration of $\pi$
can be recovered from $\ktree$ via the $\mathit{exec}$ function. 
By definition of \lacetree{s}, 
the event $\Error$ arises only from a null pointer dereference, so
 $c$ is an error configuration.

{\bf Claim~\ref{prop:execution-error}.}
For the second statement, assume that $\pi$ ends in an error configuration
and the exit status of $\ktree$ is  neither $\bf O$ nor $\bf M$. We show that the exit status of $\ktree$ is $\bf E$.
By Thm~\ref{thm:computable}, the last configuration of $\pi$
can be recovered from $\ktree$ via the $\mathit{exec}$ function. 
Since $\pi$ ends in an error configuration, the program counter
in the last frame of the lace of $\ktree$ points to an instruction causing a null pointer dereference, leading to the event $\Error$
and hence to the exit status $\bf E$.

{\bf Claim~\ref{prop:infinite}.}
Each step of an execution adds at least one frame to the lace
of its \lacetree. Since the total number of nodes and frames in a {\lacetree}
is finite, the lace will  eventually occupy a frame with index $n+1$
in the log hosting that frame, 
resulting in the exit status $\bf O$. \qed
\end{proof}

\subsection{Proofs for Section~\ref{sec:compositionality}} \label{app:proof5}

\lemconsistentchild*
\begin{proof}
\sloppy
    We follow the definition of $\ConsistentChild$ presented in Appendix~\ref{app:detailed-encoding-consistency-child}.
    The $\allowbreak\ConsistentFirstFrame(\tau,j,\sigma)$ predicate holds true because it follows the definition
    of the $\activefield$ and $\activechild$ fields in the first frame of every node, as detailed in Sec.~\ref{sec:baselabel}.
    The next two blocks in the definition of $\ConsistentChild$
    are responsible for checking that every pair of frames connected
    by the $\prevfr$ and $\nextfr$ fields and belonging to different labels
    encodes a step in the current instruction, according to the rules detailed
    in \S\ref{sec: labelling of lacetrees}.
    The two connected frames may represent a step downward (first block)
    or upward (second block) in the \lacetree.
    
    Consider the downward direction, as the other case is symmetrical.
    A downward connection between the frames $\sigma^a$ and $\tau^b$ 
    is detected in two cases:
    \begin{enumerate}
        \item First, if $\sigma^a$ is occupied (i.e., not available),
    linked to $\tau^b$ through the $\nextfr$ field, and $\tau^b$
    is also occupied. Then, we apply the predicate $\Psi_\FunDown$ to
    check that the step is valid.
    Note that it is possible that in a valid
    \lacetree $\tau^b$ is available, even though $\sigma^a.\nextfr$ points to it, but only when $\sigma^a$ is the last frame of the lace.
    \item Second, when $\tau^b$ is occupied and linked to $\sigma^a$
    via its $\prevfr$ field. In that case, $\sigma^a$ \emph{must} be
    occupied in any valid \lacetree, and we again check the correctness of the step using $\Psi_\FunDown$.\qed
    \end{enumerate}    
    \fussy
\end{proof}

\subsection{Proofs for Section~\ref{sec:reasoning}} \label{app:proof6}

\thmlabels*
\begin{proof}
We present only the ``if'' direction of the theorem; the ``only if'' direction is provided in Section~\ref{sec:reasoning}. 

Let $t$ be a node in $\ktree$ with log $\sigma$. We proceed by induction on the length $\ell$ of the lace of $\ktree$. 

\emph{Base case} ($\ell=1$): The root of $\ktree$ has one frame in the lace. {\CHC}~I or {\CHC}~II inserts $\sigma$ into  $\invCex$, depending on whether 1 or 2 frames are unavailable in $\sigma$.

\emph{Inductive step} ($\ell>1$): Let $\ktree_{\ell-1}$ be $\ktree$ with its the last frame of the lace removed. 
By inductive hypothesis, its logs are in $\invCex$. If $\sigma$ is a log of $\ktree_{\ell-1}$, the claim holds. Otherwise, $\sigma$ includes the last frame in the lace of $\ktree$, meaning that the logs of $t$'s neighboring nodes (its parent and children) remain unchanged from $\ktree_{\ell-1}$. By inductive hypothesis those logs belong to $\invCex$. Depending on the direction taken in the last step, $\sigma$ is added to $\invCex$ via {\CHC}~III or {\CHC}~V, as detailed in the above description of the individual {\CHC}s.\qed 
\end{proof}

\thmExitStatusProblem*
\begin{proof}
Let $\prob=(P, k, m, n, \ExitStatus)$. We prove the theorem by establishing the equivalent statement: \emph{ $\prob$ admits a negative answer if and only if $\Cex$ is satisfiable.}

For the ``only if'' direction, if $\prob$ has a negative answer, no $(m,n)$-\lacetree of $P$ has an exit status in $\ExitStatus$. 
Consider the interpretation of $\invCex$ containing exactly the node labels from all $(m,n)$-\lacetree prefixes of $P$.
By Theorem~\ref{thm:lab}, this is the minimal model of $\Ckt$  and  satisfies {\CHC}s~I-V.
Since none of these labels have an exit status in $\ExitStatus$,
the interpretation also satisfies the query VI, proving that $\Cex$ is satisfiable. 

For the ``if'' direction, if $\Cex$ is satisfiable, the interpretation of $\invCex$ in any model of $\Cex$ has the property that no label in $\invCex$ has an exit status in $\ExitStatus$ ({\CHC}~VI).
Since $\Cex$ implies $\Ckt$, by Theorem~\ref{thm:lab}
such interpretation of $\invCex$
contains at least all the logs of the $(m,n)$-{\lacetree} (prefixes) of $P$. We conclude that the answer to the exit status problem is negative.\qed
\end{proof}

\subsection{Proofs for Section~\ref{sec:memory safety}}\label{app:proof7}

Theorem~\ref{thm:exitstatus-positive--memorysafety-negative} is an immediate consequence of Claim~\ref{prop:tree-error} of Theorem~\ref{thm:kt-exits}.

\thmexitstatustomemorysafety*
\begin{proof}
Since no \lacetree has exit status $\mathbf{O}$, Claim~\ref{prop:infinite} of Theorem~\ref{thm:kt-exits} ensures that no execution is infinite.
Similarly, since no \lacetree has an exit status 
in $\{ \mathbf{O}, \mathbf{M}, \mathbf{E} \}$,
Claim~\ref{prop:execution-error} guarantees no execution ends in an error.
Thus, all executions reach a final configuration.\qed
\end{proof}

\section{Detailed Encoding}\label{app:encoding}

In this section, we describe in detail the predicates and functions appearing as constraints in the {\CHC}s
that recognize knitted-trees.

Consider again the system of {\CHC}s described in Figure \ref{fig:CHC-rules}.
Those {\CHC}s are based on the predicates
$\FirstFrame$, $\startframe$, $\ConsistentChild$, $\Psi_\FunInt$,
$\Psi_\FunDown$, $\Psi_\FunUp$, and $\LabelExit$,
each described in one of the following sections.

\subsection{The First Frame of Non-Root Nodes} \label{app:first-frame}

The predicate in this section is meant to constrain the first frame
of every non-root node of a \lacetree.
The only property to be enforced is a relationship
between the fields $\activefield$ and $\activechild$:
an auxiliary node ($\activefield = \false$) that is not the root
must be a leaf of the \lacetree.
Hence, no allocation of new nodes can be performed in this node.
This is enforced by setting the $\activechild_j$ flags to $\true$
for all indices $j \in [k+1,k+m]$.
Vice versa, all active nodes start with $m$ auxiliary children
with indices $k+1,\ldots,k+m$, available for allocation.
We obtain the following predicate:
\begin{align*}
    \underline{\FirstFrame}(f) \defeq\: 
    &\Big( f.\activefield \wedge {\textstyle\bigwedge_{j \in [k+1,k+m]}} \neg f.\activechild_j \Big) \vee \\
    &\Big( \neg f.\activefield \wedge {\textstyle\bigwedge_{j \in [k+1,k+m]}} f.\activechild_j \Big)
    \,.
\end{align*}

\subsection{The Start of the Lace} \label{app:start}

The predicate $\startframe$ sets constraints on the first two frames at the root of any {\lacetree}. These frames must agree on the fields $\activefield$ and $\datafield$.
The $\prevfr$ field of the second frame points to itself through the value $(\fwd,2)$.
For non-empty input trees, the first frame must be active and the second frame makes $\rootptr$ point to the root node, and $\activechild$ is non-deterministic on the first $k$ children, constraining only the children from index $k+1$ to $k+m$ as inactive.
For empty input trees, the first frame is inactive and the second frame sets $\rootptr$ to $\NIL$, and $\activechild$ requires all children to be inactive. 
In both cases, $\pc^2$ is set to the first program statement label, with all the other pointer variables set to $\NIL$.

\begin{align*}
\underline{\startframe}(\sigma) \defeq\, & \mathit{initial}(\sigma) \,\wedge \\
&\sigma.\activefield^2 = \sigma.\activefield^1 \,\wedge\,
\sigma.\datafield^2 = \sigma.\datafield^1 \wedge \\
&\sigma.\pc^2 = 0 \,\wedge\,
{\textstyle\bigwedge_{p \in \PV \setminus \{\rootptr\}}}\, \sigma.\isnil_{p}^2 \,\wedge \\
&\Big[ \, \Big( \sigma.\activefield^1 \,\wedge\, \sigma.\instr^2 = \symb{\rootptr \coloneq \here} \,\wedge\, \neg \sigma.\isnil_{\rootptr}^2 \,\wedge &&\text{\small\textcolor{gray}{Non-empty input tree}} \\
&\:\:\:\:\:{\textstyle\bigwedge_{j \in [k+1,k+m]}} \neg\sigma.\activechild^2_j \Big) \:\:\vee  \\
&\:\:\: \Big( \neg \sigma.\activefield^1 \wedge \sigma.\instr^2 = \NOP \,\wedge\, \sigma.\isnil_{\rootptr}^2 \,\wedge 
&&\text{\small\textcolor{gray}{Empty input tree}} \\
&\:\:\:\:\:{\textstyle\bigwedge_{j \in [k+m]}} \neg\sigma.\activechild^2_j \Big) \, \Big] \,,
\end{align*}
where the $\mathit{initial}$ predicate describes the distinguishing feature of the
initial frame of a lace -- namely, having itself as predecessor:
$$
\underline{\mathit{initial}}(\sigma) \defeq \neg \sigma.\avail^2 \wedge (\sigma.\prevfr^2 = (\fwd,2)) \,.
$$

\subsection{Lace Termination} \label{app:lace-end}

The function $\LabelExit(\sigma)$ returns the exit status
of any frame in $\sigma$, assuming that at most one frame in a label
is terminal. 
In the following, we denote by $P(i)$ the instruction located
at program counter $i$ in the program $P$.

\begin{align*}
\underline{\mathit{continues}}(f) &\defeq\: \neg f.\avail \wedge \FrameExit(f) = N \\
\underline{\FrameExit}(f) &\defeq 
   \begin{cases}
      C &\text{if } P(f.\pc) = \exit \\
      E &\text{if } f.\instr = \Error \\
      M &\text{if } f.\instr = \OOM \\      
      N &\text{otherwise.}
   \end{cases} \\
\underline{\LabelExit}(\sigma) &\defeq 
   \begin{cases}
      C &\text{if } \bigvee_{i \in [2,n]} \big( \neg\sigma.\avail^i \wedge \FrameExit(\sigma^i) = C \big) \\
      E &\text{if } \bigvee_{i \in [2,n]} \big( \neg\sigma.\avail^i \wedge \FrameExit(\sigma^i) = E \big) \\
      M &\text{if } \bigvee_{i \in [2,n]} \big( \neg\sigma.\avail^i \wedge \FrameExit(\sigma^i) = M \big) \\      
      O &\text{if } \neg \sigma.\avail^{n+1} \\
      N &\text{otherwise.}
   \end{cases} \\
\underline{\LabelExit}(\sigma, \ExitStatus) &\defeq 
\big( \LabelExit(\sigma) \in \ExitStatus \big)
\end{align*}

\subsection{Checking Parent-Child Consistency}\label{app:detailed-encoding-consistency-child}

The $\ConsistentChild(\tau,j,\sigma)$ 
predicate checks that the label $\tau$ of the $j$-th child
is consistent with the label $\sigma$ of the parent.
It uses the predicates $\ConsistentFirstFrame$, $\Psi_\FunDown$, and $\Psi_\FunUp$,
defined later.

{\small
\begin{align*}
\underline{\ConsistentChild}&(\tau, j, \sigma) \defeq\, 
\ConsistentFirstFrame(\tau,j,\sigma) \,\wedge \\
%
%
&\text{\small\textcolor{gray}{Consistency of steps, when going down:}} \\
&\bigwedge_{a \in [2,n] \atop b \in [2,n+1]} \Big[
\Big( \, \big( \neg\sigma.\avail^a \wedge \sigma.\nextfr^a = (j, b) \wedge \neg\tau.\avail^b \big) 
\,\vee \\[-2ex] 
&\hspace{1.5cm} \big( \neg\tau.\avail^b \wedge \tau.\prevfr^b = (\uparrow, a) \big) \, \Big)
\rightarrow \Psi_\FunDown(\sigma^{\leq a}, j, \tau^{<b}, \tau^b) 
\Big] \,\wedge \\[1ex]
&\text{\small\textcolor{gray}{Consistency of steps, when going up:}} \\
&\bigwedge_{a \in [2,n] \atop b \in [2,n+1]} \Big[
\Big( \, \big( \neg\tau.\avail^a \wedge \tau.\nextfr^a = (\uparrow, b) \wedge
\neg\sigma.\avail^b \big) \,\vee \\[-2ex] 
&\hspace{1.5cm} \big( \neg\sigma.\avail^b \wedge \sigma.\prevfr^b = (j, a) \big) \, \Big)
\rightarrow \Psi_\FunUp(\tau^{\leq a}, j, \sigma^{<b}, \sigma^b) \Big].
\end{align*}
}

It is worth pointing out that different child labels
may be consistent with the same parent label,
just like two different parent labels may be consistent with
the same child label.
First, $\ConsistentChild$ only constrains frames that
are directly involved in an interaction between parent and child,
i.e., pairs of frames belonging to these different nodes 
and adjacent in the lace.
Now, consider a pair of frames that are adjacent in the lace and belong
to two different nodes (parent and child). 
Here is a description of how $\ConsistentChild$ constrains each field in those frames:
\begin{description}
    \item{$\avail$:} this field must be false, because frames involved
    in an interaction belong to the lace;
    \item{$\activefield$:} this field is directly constrained only in the first frame of each label, where its value must follow the backbone rules described in Section~\ref{sec:backbone};
    this is ensured by the auxiliary predicate $\ConsistentFirstFrame$;
    the value in the subsequent frames of each label is only constrained indirectly,
    by the semantics of the $\New$ and $\Free$ instructions;
    \item{$\datafield$:} the default constraint for this field
    is to have the same value as the preceding frame in the
    same log; the only exception is the last step of the
    instruction $\field{p}{\data} \coloneq \mathit{exp}$, when the current
    value of $p$ has been ascertained via rewinding:
    in such a step the value of the $\datafield$ field of the new
    frame takes the value of the expression $\mathit{exp}$;
    \item{$\pc$:} the program counter must be consistent with the corresponding $\pc$ on the other side of the interaction;
    i.e., the two pc's must be equal if a rewind operation is 
    ongoing, and otherwise the instruction being executed is complete and the pc advances according to the control-flow graph of $P$;
    \item{$d$:} the value of this field in the two frames of an interaction is the
    same, except in the last step of an instruction of the type 
    $d \coloneq \field{p}{\data}$, when the current value of $p$ 
    has been ascertained via rewinding;
    \item{$\chg_p$:} the auxiliary predicates $\Psi_*$ constrain
    these fields to follow the rules described on page~\pageref{page:upd};
    \item{$\isnil_p$:} this field generally keeps its value in an interaction, except in the last step of an assignment to $p$;
    specifically, it can switch from $\true$ to $\false$
    in the last step of $p \coloneq q$ or $p \coloneq \field{q}{\pfield}$, when
    $q$ is not $\NIL$ and rewinding was needed to find its current
    value;
    moreover, it can switch from $\false$ to $\true$ 
    in the last step of $p \coloneq \field{q}{\pfield}$, when
    $q$ is not $\NIL$, rewinding was needed,
    and then $\field{q}{\pfield}$ is found to be $\NIL$;    
    \item{$\instr$:} the event in the second frame
    of an interaction is dictated by the program instruction being
    executed and by the meaning of each event; 
    for example, during a rewind process, the events
    $\RWD_i$ and $\RWD_{i,p}$ are used; the end of a rewinding
    may be marked by a event of the types 
    $\symb{\pfield \coloneq p}$, $\symb{\pfield \coloneq \NIL}$, 
    or $\symb{p \coloneq \here}$; in some cases, the final
    event is simply $\NOP$;
    \item{$\nextfr$ and $\prevfr$:} in a pair of interacting
    frames these fields must connect the two frames with each other;
    i.e., the $\nextfr$ field of the first frame points to the second
    frame, and the $\prevfr$ field of the second frame points to the first
    one.
\end{description}

The following auxiliary predicates are used within $\ConsistentChild$.
The $\ConsistentFirstFrame$ predicate describes the initial value of the 
$\activefield$ field, that starts as true on all nodes of the \lacetree that
correspond to nodes of the input tree, and false elsewhere.

\begin{align*}
    \underline{\ConsistentFirstFrame}(\tau, j, \sigma) \defeq\,
&\big( \mathit{initial}(\sigma) \vee \sigma.\activefield^1 \big) \,\wedge\,
\big( \tau.\activefield^1 \rightarrow \sigma.\activefield^1  \big) \,\wedge \\
&\big( j>k \rightarrow \neg\tau.\activefield^1 \big) \,\wedge\, 
\sigma.\activechild^1_j = \tau.\activefield^1 \,.
\end{align*}

We now present the three predicates that describe the next frame in the lace,
depending on the direction of the next step.
The first predicate, $\Psi_\FunInt$, describes internal steps.
It invokes the general predicate $\StepFrame$
that checks a single step in a lace, passing two adjacent frames from
the same log.
In addition, it makes sure that the $\chg_p$ flags
are all false after an internal step, according to the meaning of those
flags, as described on page~\pageref{page:upd}. \\
In the following formulas, for a label $\sigma$ and a frame $f$ we write
$\sigma \cdot f$ to denote the label obtained by the replacing the first
available frame of $\sigma$ with $f$.

\begin{align*}
\underline{\Psi_\FunInt}(\sigma, f) \defeq\; 
&\len(\sigma,a) \wedge \Continues(\sigma^{a}) \wedge \sigma.\nextfr^{a} = (\fwd, a+1) \,\wedge \\
&f.\prevfr = (\fwd, a) \wedge \StepFrame(\sigma, a; \sigma \cdot f, a+1)
\wedge {\textstyle\bigwedge_{p \in \PV}}\, \neg f.\chg_p \,.
\end{align*}

The following predicates $\Psi_\FunDown$ and $\Psi_\FunUp$ describe the next frame
in the lace, when the current step involves a change of node, either from a node
down to its $j$-th child or from a node up to its parent.
In particular, these predicates constrain the $\chg$ flags in the new frame
to respect their meaning, as described on page~\pageref{page:upd}.

Consider how the last two lines in $\Psi_\FunUp(\tau,j,\sigma,f)$
relate to the intended behavior of the $\chg$ flags.
In that context, $\tau^a$ is a frame in the $j$-th child of a parent node
with log $\sigma$, and the next frame in the lace is $f$, to be appended
on top of $\sigma$.
Note that in this situation the previous frame in the log of the parent    (i.e., $\sigma^{b-1}$) is followed in the lace by a frame in the $j$-th child.
Now, for every pointer variable $p \in \PV$, the $\chg_p$ flag when going
back to the parent is true iff $p$ is not
$\NIL$ at that time and we observe a marker for an update to $p$ within the frames of the child located between the step down and the step back up.
If the update to $p$ was performed in one of the frames of the child, the marker is a direct event of the type $\symb{p \coloneq \here}$.
If instead the update occurs further below in the subtree rooted in the child,
the marker observed in the child is a $\chg_p$ field set to true,
as encoded in the last lines of $\Psi_\FunUp$.

Finally, note that $\Psi_\FunUp$ is also responsible for synchronizing the $\activechild$
field in the parent with the $\activefield$ field in the child; these fields may become
out-of-sync after a $\Free$ operation performed on the child.

%

{\small
\begin{align*}
\underline{\Psi_\FunDown}(\sigma, j, \tau, f) \defeq\;
&\len(\sigma,a) \wedge \Continues(\sigma^{a}) \wedge \sigma.\nextfr^{a} = (j, b) \,\wedge 
\len(\tau,b-1) \,\wedge \\
&f.\prevfr = (\uparrow, a) \wedge \StepFrame(\sigma, a; \tau \cdot f, b) \,\wedge \\
&\bigwedge_{p \in \PV} \Big[\,  f\!.\chg_p \leftrightarrow \Big( 
\neg  f\!.\isnil_p \,\wedge b > 2 \wedge \tau\!.\nextfr^{b-1} = (\uparrow,a') \,\wedge \\[-1ex]
&\hspace{1cm} \Big( 
{\textstyle\bigvee_{c \in [a', a]}}   \big( \sigma.\instr^c = \symb{p \coloneq \here} \big) \,\vee\,
{\textstyle\bigvee_{c \in [a'+1, a]}} \sigma.\chg_p^c \Big) \, \Big)\, \Big] \\[1ex]
\underline{\Psi_\FunUp}(\tau, j, \sigma, f) \defeq\; 
&\len(\tau,a) \wedge \Continues(\tau^{a}) \wedge \tau.\nextfr^{a} = (\uparrow, b) \,\wedge 
\len(\sigma,b-1) \,\wedge \\
&f.\prevfr = (j, a) \wedge \StepFrame(\tau, a; \sigma \cdot f, b) \,\wedge 
f.\activechild = \tau.\activefield^{a} \wedge \\
&\bigwedge_{p \in \PV} \Big[\,  f\!.\chg_p \leftrightarrow \Big( 
\neg  f\!.\isnil_p \,\wedge b > 2 \wedge \sigma\!.\nextfr^{b-1} = (j,a') \,\wedge \\[-1ex]
&\hspace{1cm} \Big( 
{\textstyle\bigvee_{c \in [a', a]}}   \big(\tau.\instr^c = \symb{p \coloneq \here}\big) \,\vee\,
{\textstyle\bigvee_{c \in [a'+1, a]}} \tau.\chg_p^c \Big) \, \Big)\, \Big] .
\end{align*}
}

\setlength{\marginparsep}{-1cm}
\setlength{\marginparwidth}{2cm}

\subsection{Individual Statements}

The predicate $\StepFrame(\sigma,a;\tau,b)$ 
holds if $\tau^b$ is the logically correct frame to follow
$\sigma^a$ in a lace.
The structure of this predicate comprises an implication for each type of instruction
in our programming language, except for $\exit$.
The $\exit$ instruction does not produce a step,
because just having the program counter point to it
identifies a terminating execution.
{\small
\begin{align*} 
\underline{\StepFrame}&(\sigma, a; \tau, b) \defeq\: \\
   &P(\sigma.\pc^a) = \mathbf{skip} &&\rightarrow 
    \mathit{step\_skip(\sigma^a, \tau^b)} \;\wedge \\
   &P(\sigma.\pc^a) = (p \coloneqq \NIL) &&\rightarrow 
    \mathit{step\_assgn\_nil(\sigma^a, \tau^b, p)} \;\wedge \\
   &P(\sigma.\pc^a) = (d \coloneqq \dataexp) &&\rightarrow 
    \mathit{step\_assgn\_exp(\sigma^a, \tau^b, d, \dataexp)} \;\wedge \\
   &P(\sigma.\pc^a) = (d_\mathrm{bool} \coloneqq (p=q)) &&\rightarrow 
    \mathit{step\_assgn\_cond(\sigma, a, \tau, b, d_\mathrm{bool}, p, q)} \;\wedge \\
   &P(\sigma.\pc^a) = (p \coloneqq q) &&\rightarrow 
    \mathit{step\_assgn\_ptr(\sigma, a, \tau, b, p, q)} \;\wedge \\
   &P(\sigma.\pc^a) = (p \coloneqq \field{q}{\pfield}) &&\rightarrow 
    \mathit{step\_assgn\_from\_field(\sigma, a, \tau, b, p, q, \pfield)} \;\wedge \\    
   &P(\sigma.\pc^a) = (\field{p}{\pfield} \coloneq q) &&\rightarrow 
    \mathit{step\_assgn\_to\_field(\sigma, a, \tau, b, p, q, \pfield)} \;\wedge \\    
&P(\sigma.\pc^a) = (\field{p}{\datafield} \coloneq \dataexp) &&\rightarrow 
    \mathit{step\_assgn\_to\_data(\sigma, a, \tau, b, p, \dataexp)} \;\wedge \\
&P(\sigma.\pc^a) = (d \coloneq \field{p}{\datafield}) &&\rightarrow 
    \mathit{step\_assgn\_to\_var}(\sigma, a, \tau, b, d, p) \;\wedge \\
&P(\sigma.\pc^a) = (\New\ p) &&\rightarrow 
    \mathit{step\_new(\sigma^a, \tau^{b-1}, \tau^b, p)} \;\wedge \\    
&P(\sigma.\pc^a) = (\Free\ p) &&\rightarrow 
    \mathit{step\_free(\sigma, a, \tau, b, p)} \;\wedge \\    
&\big( P(\sigma.\pc^a) = (\mathbf{while}\ p=q ) \,\vee \\
&\:\: P(\sigma.\pc^a) = (\mathbf{if}\ p=q) \big)
&&\rightarrow 
    \mathit{step\_cmp\_ptr(\sigma, a, \tau, b, p, q, \false)} \;\wedge \\    
&\big( P(\sigma.\pc^a) = (\mathbf{while}\ \neg (p=q)) \,\vee \\
&\:\:      P(\sigma.\pc^a) = (\mathbf{if}\ \neg (p=q)) \big)
&&\rightarrow 
    \mathit{step\_cmp\_ptr(\sigma, a, \tau, b, p, q, \true)} \;\wedge \\    
&\big( P(\sigma.\pc^a) = (\mathbf{while}\ r(\dataexp_1,\ldots,\dataexp_l) ) \,\vee \\
&\:\:  P(\sigma.\pc^a) = (\mathbf{if}\ r(\dataexp_1,\ldots,\dataexp_l) ) \big)
&&\rightarrow 
    \mathit{step\_local\_branch}(\sigma, a, \tau, b, r, \dataexp_1,\ldots,\dataexp_l) \,.
\end{align*}
}

\subsubsection{Default values.} \label{app:defaults}

Regardless of any specific instruction, most fields in any new frame of the lace
are set to their default values, as described in Section~\ref{sec:defaults}.
The following $\mathit{default}$ predicate represents the full case, 
where all fields that have default values are constrained.
In the subsequent predicates, we add subscripts to $\mathit{default}$ to specify
which fields follow the default rules.
\begin{align*}
\underline{\mathit{default}}(f^\mathit{prev}, f^\mathit{below}, f) \defeq \:
&\neg f.\avail \,\wedge &&\textcolor{gray}{\avail} \\
&f.\activefield = f^\mathit{below}.\activefield \,\wedge &&\textcolor{gray}{\activefield} \\
&f.\datafield = f^\mathit{below}.\datafield \,\wedge &&\textcolor{gray}{\datafield} \\
&{\textstyle\bigwedge_{d \in \DV}}\, f.d = f^\mathit{prev}.d \,\wedge &&\textcolor{gray}{d} \\
&{\textstyle\bigwedge_{p \in \PV}}\, f.\isnil_p = f^\mathit{prev}.\isnil_p \,\wedge &&\textcolor{gray}{\isnil} \\
&f.\instr = \NOP \,\wedge &&\textcolor{gray}{\instr} \\
&f.\pc = f^\mathit{prev}.\pc \,\wedge &&\textcolor{gray}{\pc} \\
&\mathit{default}_\activechild(f^\mathit{prev}, f^\mathit{below}, f).  &&\textcolor{gray}{\activechild}
\end{align*}
The default for the $\activechild$ flags requires a little more care:
\begin{align*}
\underline{\mathit{default}_\activechild}&(f^\mathit{prev}, f^\mathit{below}, f) \defeq \: \\
&f.\prevfr = (j, *) \rightarrow \Big( f.\activechild_j = f^\mathit{prev}.\activefield \,\wedge \\
&\hspace{2.7cm} {\textstyle\bigwedge_{j \neq i}} \, \big( f.\activechild_i = f^\mathit{below}.\activechild_i \big) \, \Big) \:\wedge
\\
&f.\prevfr \neq (j, *) \rightarrow 
{\textstyle\bigwedge_{i \in [k+m]}}\, \big( f.\activechild_i = f^\mathit{below}.\activechild_i \big) \,. 
\end{align*}

\subsubsection{Internal instructions.} \label{app:local-steps}

Internal instructions push a single new frame on the current node of the \lacetree.
As our first type of internal instruction, we describe the $\mathit{step\_assgn\_nil}$ predicate.
As explained earlier, its encoding simply pushes a new frame $f_2$ 
on the current node, with the $\isnil_p$ flag set to \true.
All the other fields of the new frame take their default value, 
except the program counter, which advances to the next instruction:
\begin{align*}
\sidecomment{$p \coloneq \NIL$}
\underline{\mathit{step\_assgn\_nil}}&(f_1, f_2, p) \defeq\: 
(f_1.\nextfr = \fwd) \;\wedge \\
&f_2.\isnil_p \,\wedge\, {\textstyle\bigwedge_{q \in \PV \setminus \{ p \}}} (f_2.\isnil_q = f_1.\isnil_q) \,\wedge \\
&\mathit{advance\_pc}(f_1, f_2) \wedge \\
&\mathit{default}_{\avail,\activefield,\datafield,d,\instr,\activechild}(f_1, f_1, f_2).
\end{align*}
\noindent
The following auxiliary predicate encodes the advancement of the program counter:
$$
\underline{\mathit{advance\_pc}}(f_1, f_2) \defeq\, \big(f_2.\pc = \nextpc(f_1.\pc) \big).
$$
\noindent
Next, we present the predicate that encodes the $\New$ instruction.
Note that a failed $\New$ is only justified if all children of the current node
with position in $[k+1,k+m]$ are active (i.e., currently allocated).
\begin{align*}
\sidecomment{$\New\ p$}
\underline{\mathit{step\_new}}&(f, f', f'', p) \defeq\: \\
&\text{\small\textcolor{gray}{Case 1: Out-of-memory error}} \\
&\Big( {\textstyle\bigwedge_{j \in [k+1,k+m]}}\, f.\activechild_j \,\wedge\, f.\nextfr = (\fwd,*) \,\wedge\, f''.instr = \OOM \,\wedge \\ &\quad\mathit{default}_{\avail,\activefield,\datafield,d,\pc,\isnil,\activechild}(f, f', f'') \Big) \quad\vee \\
&\text{\small\textcolor{gray}{Case 2: Normal case}} \\
&\Big( \bigvee_{j \in [k+1,k+m]} \big( \neg f.\activechild_j \wedge {\bigwedge_{\mathclap{i \in [k+1,j-1]}}} \, f.\activechild_i \wedge f.\nextfr = (j,*) \big) \,\wedge \\
&\:\:\neg f''.\isnil_p \,\wedge\, \bigwedge_{q \in \PV \setminus \{ p \}} \big( f''.\isnil_q = f.\isnil_q \big) \;\wedge \\
&\:\:f''.\instr = \symb{ p \coloneq \here } \,\wedge\, f''.\activefield \,\wedge\,
\mathit{advance\_pc}(f, f'') \,\wedge \\
&\:\:\mathit{default}_{\avail,\datafield,d,\activechild}(f, f', f'') \,\Big) \,.
\end{align*}
\noindent
The following predicate handles assignments of arbitrary data expressions to data variables.
\begin{align*}
\sidecomment{$d \coloneq \dataexp$}
\underline{\mathit{step\_assgn\_exp}}&(f_1, f_2, d, \dataexp) \defeq\: \\
&\big( f_1.\nextfr = (\fwd,*) \big) \,\wedge\, f_2.d = exp[d' \mapsto f_1.d']_{d' \in \DV}  \,\wedge \\ 
&{\textstyle\bigwedge_{d' \in \DV \setminus \{d\} }}\, (f_2.d' = f_1.d') \,\wedge 
\mathit{advance\_pc}(f_1, f_2) \wedge \\
&\mathit{default}_{\avail,\activefield,\datafield,\isnil,\instr,\activechild}(f_1, f_1, f_2) \,.
\end{align*}
\noindent
Finally, the following is the straightforward encoding
of the $\Skip$ statement.
\begin{align*}
\sidecomment{$\Skip$}
\underline{\mathit{step\_skip}}(f_1, f_2, p) \defeq\: 
&\big( f_1.\nextfr = (\fwd,*) \big) \,\wedge\, \mathit{advance\_pc}(f_1, f_2) \,\wedge \\ &\mathit{default}_{\avail,\activefield,\datafield,d,\isnil,\instr,\activechild}(f_1, f_1, f_2).
\end{align*}

\subsubsection{Walking instructions.}

The following \emph{walking instructions} may add multiple frames to different nodes
of the \lacetree.
The first such instruction is the assignment of the form $p \coloneq q$,
which may require rewinding the lace to find the node currently
pointed by $q$. In fact, the encoding distinguishes three cases:
(1) when $q$ is $\NIL$, 
(3) when $q$ points elsewhere, and we need to start or keep rewinding the lace,
(2) when $q$ points to the current node (i.e.,
the node with label $\sigma$).
The predicates governing the rewinding operation are presented later in Section~\ref{app:rewind}.
\begin{align*}
\sidecomment{$p\coloneq q$}
 &\underline{\mathit{step\_assgn\_ptr}}(\sigma, a, \tau, b, p, q) \defeq\: \\
&\hspace{1cm}\text{\small\textcolor{gray}{Case 1: $q$ is $\NIL$; use the encoding for $p \coloneq \NIL$}} \\
 &\hspace{1cm}\Big( \sigma^a.\isnil_q \,\wedge\, 
  \mathit{step\_assgn\_nil}(\sigma^a, \tau^b, p) \Big) \quad\vee \\
&\hspace{1cm}\text{\small\textcolor{gray}{Case 2: $q$ points elsewhere; start or keep rewinding}} \\
&\hspace{1cm} \mathit{rewind(\sigma, a, \tau, b, q)} \quad\vee \\
&\hspace{1.0cm}\text{\small\textcolor{gray}{Case 3: $q$ points to the current node}} \\
 &\hspace{1cm}\Big( \mathit{stop\_rewind(\sigma, a, q)} \,\wedge\, 
\mathit{set\_ptr\_here}(\sigma^a, \tau^b, p) \Big) \,.
\end{align*}
The auxiliary predicate $\mathit{set\_ptr\_here}$ pushes a frame on the current
node with $\symb{ p \coloneqq \here }$ and updates $\isnil$ accordingly.
\begin{align*}
\underline{\mathit{set\_ptr\_here}}(f_1, f_2, p) \defeq\: 
&\big( f_1.\nextfr = (\fwd,*) \big) \,\wedge\, \mathit{advance\_pc}(f_1, f_2) \, \wedge \\
&f_2.\instr = \symb{ p \coloneqq \here } \,\wedge\, \\
&\neg f_2.\isnil_p \,\wedge\, {\textstyle\bigwedge_{q \in \PV \setminus \{ p \}}} \big( f_2.\isnil_q = f_1.\isnil_q \big) \;\wedge \\
&\mathit{default}_{\avail,\activefield,\datafield,d,\activechild}(f_1, f_1, f_2) \,.
\end{align*}
The following predicate encodes the statements of the form $\field{p}{\pfield} \coloneq q$.
\begin{align*}
\sidecomment{$\field{p}{\pfield} \coloneq q$}
&\underline{\mathit{step\_assgn\_to\_field}}(\sigma, a, \tau, b, p, q, \pfield) \defeq\: 
\mathit{find\_or\_fail}(\sigma, a, \tau, b, p) \:\vee \\
&\hspace{1cm}\text{\small\textcolor{gray}{$p$ points to the current node:}} \\ &\hspace{1cm} \Bigg( \mathit{stop\_rewind(\sigma, a, p)} \,\wedge\,  \sigma^a.\nextfr = (\fwd,a+1) \,\wedge 
 \mathit{advance\_pc}(\sigma^a, \tau^b) \, \wedge \\
 &\hspace{1.3cm} \Big( \big( \neg\sigma^a.\isnil_q \wedge \tau^b.\instr = \symb{ \pfield \coloneqq q } \big) \vee \\
 &\hspace{1.5cm} \big( \sigma^a.\isnil_q \wedge \tau^b.\instr = \symb{ \pfield \coloneqq \NIL } \big) \Big) \,\wedge\, \\
&\hspace{1.3cm}\mathit{default}_{\avail,\activefield,\datafield,d,\isnil,\activechild}(\sigma^a, \tau^{b-1}, \tau^b) 
    \Bigg).
\end{align*}

The predicate $\mathit{step\_assgn\_nil\_to\_field}$, corresponding to
the statements of the form $\field{p}{\pfield} \coloneq \NIL$,
is a simple variant of the above, that only sets
the event $\symb{ \pfield \coloneqq \NIL }$.

Next, we deal with the two instructions that write or read the data field
of a node.

\begin{align*}
\sidecomment{$\field{p}{\datafield} \coloneq \dataexp$}
&\underline{\mathit{step\_assgn\_to\_data}}(\sigma, a, \tau, b, p, \dataexp) \defeq\:
\mathit{find\_or\_fail}(\sigma, a, \tau, b, p) \:\vee \\
&\hspace{1cm}\text{\small\textcolor{gray}{$p$ points to the current node:}} \\ &\hspace{1cm}\Big( \mathit{stop\_rewind(\sigma, a, p)} \,\wedge\, \big( \sigma^a.\nextfr = (\fwd,a+1) \big) \,\wedge 
 \mathit{advance\_pc}(\sigma^a, \tau^b) \, \wedge \\
&\hspace{1.2cm}\tau^b.\datafield = \dataexp\big[d \mapsto \sigma^a.d\big]_{d \in \DV}  \,\wedge \\
&\hspace{1.2cm}\mathit{default}_{\avail,\activefield,d,\isnil,\instr,\activechild}(\sigma^a, \tau^{b-1}, \tau^b) 
    \Big) \,.
\end{align*}
\begin{align*} 
\sidecomment{$d \coloneq \field{p}{\datafield}$}
&\underline{\mathit{step\_assgn\_to\_var}}(\sigma, a, \tau, b, d, p) \defeq\:
\mathit{find\_or\_fail}(\sigma, a, \tau, b, p) \:\vee \\
&\hspace{1cm}\text{\small\textcolor{gray}{$p$ points to the current node:}} \\ &\hspace{1cm}\Big( \mathit{stop\_rewind(\sigma, a, p)} \,\wedge\, \big( \sigma^a.\nextfr = (\fwd,a+1) \big) \,\wedge 
 \mathit{advance\_pc}(\sigma^a, \tau^b) \, \wedge \\
&\hspace{1.2cm}\tau^b.d = \sigma^a.\datafield \,\wedge\,
{\textstyle\bigwedge_{d' \in \DV \setminus \{d\} }} (\tau^b.d' = \sigma^a.d') \,\wedge \\
&\hspace{1.2cm}\mathit{default}_{\avail,\activefield,\datafield,\isnil,\instr,\activechild}(\sigma^a, \tau^{b-1}, \tau^b) 
    \Big) \,.
\end{align*}

The following predicate handles the deallocation of the node pointed by a given pointer.

\begin{align*} 
\sidecomment{$\Free\ p$}
&\underline{\mathit{step\_free}}(\sigma, a, \tau, b, p) \defeq\: 
\mathit{find\_or\_fail}(\sigma, a, \tau, b, p) \:\vee \\
&\hspace{1cm}\text{\small\textcolor{gray}{$p$ points to the current node:}} \\ &\hspace{1cm}\Big( \mathit{stop\_rewind(\sigma, a, p)} \,\wedge\, \big( \sigma^a.\nextfr = (\fwd,*) \big) \,\wedge
 \mathit{advance\_pc}(\sigma^a, \tau^b) \, \wedge \\
&\hspace{1.2cm}\neg \tau^b.\activefield \,\wedge\, \\
&\hspace{1.2cm} \bigwedge_{q \in \PV} \Big( \, \big( \PointsHere(\sigma, a, q) \rightarrow \tau^b.\isnil_q \big) \,\wedge \\[-1ex]
&\hspace{2.3cm}\big( \neg\PointsHere(\sigma, a, q) \rightarrow \tau^b.\isnil_q = \sigma^a.\isnil_q \big) \,\Big) \,\wedge \\
&\hspace{1.2cm}\mathit{default}_{\avail,\datafield,d,\instr,\activechild}(\sigma^a, \tau^{b-1}, \tau^b) 
    \Big) \,.
\end{align*}

The following auxiliary predicates handle the search for the current value
of a pointer variable $p$, including issuing an error if $p$ is $\NIL$.
\begin{align*}
\underline{\mathit{find\_or\_fail}}(\sigma, a, \tau, b, p) \defeq\:
&\Big( \sigma^a.\isnil_p \,\wedge\, 
  \mathit{error}(\sigma^a, \tau^b) \Big) \vee \mathit{rewind(\sigma, a, \tau, b, p)} \\[1ex]
\underline{\mathit{error}}(f_1, f_2) \defeq\: 
&\big( f_1.\nextfr = (\fwd,*) \big) \,\wedge\,
(f_2.\instr = \Error) \,\wedge \\ &\mathit{default}_{\avail,\activefield,\datafield,d,\isnil,\pc,\activechild}(f_1, f_1, f_2) \,.
\end{align*}

Next, we move to the instruction that assigns to a Boolean variable
the result of the comparison between two pointers.
If at least one of the two pointers is $\NIL$, the comparison can be resolved locally.
Otherwise, a rewind operation may be necessary.
The auxiliary predicates $\mathit{rewind2}$, $\mathit{stop\_rewind2}$,
and $\mathit{are\_equal\_after\_rewind}$ are described in Sec.~\ref{app:rewind}.
{\allowdisplaybreaks
\begin{align*}
\sidecomment{$d_\mathrm{bool} \coloneq (p = q)$}
&\underline{\mathit{step\_assgn\_cond}}(\sigma, a, \tau, b, d_\mathrm{bool}, p, q) \defeq\: \\
&\hspace{1cm}\text{\small\textcolor{gray}{Case 1: at least one pointer is $\NIL$}}   \\
&\hspace{1cm}\Big( \, \big( \sigma^a.\isnil_p \vee \sigma^a.\isnil_q \big) \,\wedge \\
&\hspace{1cm}\quad \big( \sigma^a.\nextfr = (\fwd,*) \big) \,\wedge\, \mathit{advance\_pc}(\sigma^a, \tau^b) \,\wedge \\ 
&\hspace{1cm}\quad\tau^b.d_\mathrm{bool} = \big( \sigma^a.\isnil_p \leftrightarrow \sigma^a.\isnil_q) \big) \, \wedge \\
&\hspace{1cm}\quad{\textstyle\bigwedge_{d \in \DV \setminus \{d_\mathrm{bool}\} }} (\tau^b.d = \sigma^a.d) \,\wedge \\
&\hspace{1cm}\quad \mathit{default}_{\avail,\activefield,\datafield,\isnil,\instr,\activechild}\big( \sigma^a, \tau^{b-1}, \tau^b \big) \,\Big) \quad\vee\\[1ex]
&\hspace{1cm}\text{\small\textcolor{gray}{Case 2: start or keep rewinding}}  \\
&\hspace{1cm}\mathit{rewind2}(\sigma, a, \tau, b, p, q) \quad\vee \\[1ex] 
&\hspace{1cm}\text{\small\textcolor{gray}{Case 3: stop rewinding}}  \\
 &\hspace{1cm}\Big(\, \mathit{stop\_rewind2}(\sigma, a, p, q)
 \,\wedge\, \big( \sigma^a.\nextfr = (\fwd,*) \big) \,\wedge\, \mathit{advance\_pc}(\sigma^a, \tau^b) \,\wedge \\ 
&\hspace{1cm}\quad\tau^b.d_\mathrm{bool} = \mathit{are\_equal\_after\_rewind}(\sigma, a, p, q) \big) \, \wedge \\
&\hspace{1cm}\quad{\textstyle\bigwedge_{d \in \DV \setminus \{d_\mathrm{bool}\} }} (\tau^b.d = \sigma^a.d) \,\wedge \\
&\hspace{1cm}\quad \mathit{default}_{\avail,\activefield,\datafield,\isnil,\instr,\activechild}(\sigma^a, \tau^{b-1}, \tau^b) 
    \Big) \,.
\end{align*}
}

The most complex walking instruction is the assignment of the form
$p \coloneq \field{q}{\pfield}$, because it may involve \emph{two}
consecutive rewinding phases.
{\allowdisplaybreaks
\begin{align*} 
\sidecomment{$p \coloneq \field{q}{\pfield}$}
 &\underline{\mathit{step\_assgn\_from\_field}}(\sigma, a, \tau, b, p, q, \pfield) \defeq\: \\
&\hspace{1cm}\text{\small\textcolor{gray}{Case 1: $q$ is $\NIL$; null pointer dereference}} \\
&\hspace{1cm}\big( \sigma^a.\isnil_q \,\wedge\, 
  \mathit{error}(\sigma^a, \tau^b) \big) \quad\vee \\[1ex]
&\hspace{1cm}\text{\small\textcolor{gray}{Case 2a: phase I; $q$ points elsewhere; start or keep rewinding}} \\
&\hspace{1cm} \big( \sigma^a.\instr \neq \RWD_{*,*} \,\wedge\, \mathit{rewind(\sigma, a, \tau, b, q)} \big) \quad\vee \\
&\hspace{1cm}\text{\small\textcolor{gray}{Case 2b: phase II; $r$ points elsewhere; start or keep rewinding}} \\
&\hspace{1cm} \big( \sigma^a.\instr = \RWD_{i,r} \,\wedge\, \mathit{rewind\_special}(\sigma, a, \tau, b, r, i) \big) \quad\vee \\[1ex]
&\hspace{1cm}\text{\small\textcolor{gray}{Case 3a: end of phase I; $q$ points to the current node}} \\ 
&\hspace{1cm}\Bigg( \sigma^a.\instr \neq \RWD_{*,*} \wedge \mathit{stop\_rewind(\sigma, a, q)} \,\wedge\, \\[-1ex]
&\hspace{1.2cm}
\Big( \big( \mathit{is\_pfield\_nil}(\sigma, a, \pfield)
\vee \\ 
&\hspace{1.4cm} (\mathit{is\_pfield\_implicit}(\sigma, a, \pfield) \wedge 
 \neg \sigma^a.\activechild_{\mathit{index(\pfield)}} ) \big)
\rightarrow \\
&\hspace{1.7cm}\mathit{step\_assign\_nil}(\sigma^a, \tau^b, p) \Big)
 \,\wedge \\
&\hspace{1.2cm}\Big( \big( \mathit{is\_pfield\_implicit}(\sigma, a, \pfield) \,\wedge\,
\sigma^a.\activechild_{\mathit{index(\pfield)}} \big)
\rightarrow \\
&\hspace{1.7cm}\big( \sigma^a.\nextfr = (\mathit{index}(\pfield), b) \,\wedge\, 
\mathit{advance\_pc}(\sigma^a, \tau^b) \, \wedge \\
&\hspace{1.9cm} \tau^b.\instr = \symb{ p \coloneqq \here } \,\wedge\, \\
&\hspace{1.9cm}\mathit{default}_{\avail,\activefield,\datafield,d,\isnil,\activechild}(\sigma^a, \tau^{b-1}, \tau^b) \big) \,\Big) \,\wedge \\
&\hspace{1cm}\bigwedge_{r \in \PV, i \in [2,n]} \Big(\,
           \big( \mathit{is\_pfield\_ptr}(\sigma, a, \pfield, r, i) \wedge 
            \PointsHere(\sigma, i, r) \big) \rightarrow \\[-2ex]
&\hspace{3cm}\mathit{set\_ptr\_here}(\sigma^a, \tau^b, p) \,\Big) \:\wedge \\
 &\hspace{1cm}\bigwedge_{r \in \PV, i \in [2,n]} 
           \Big(\, \big( \mathit{is\_pfield\_ptr}(\sigma, a, \pfield, r, i) \wedge 
            \neg\PointsHere(\sigma, i, r) \big) \rightarrow \\[-2ex]
&\hspace{3cm}\mathit{rewind\_special}(\sigma, a, \tau, b, r, i) \Big) \Bigg) \quad\vee \\
&\hspace{1cm}\text{\small\textcolor{gray}{Case 3b: end of phase II; $r$ points to the current node}} \\
&\hspace{1cm}\Big( \sigma^a.\instr = \RWD_{i,r} \wedge \PointsHere(\sigma, i, r) \,\wedge\, \mathit{set\_ptr\_here}(\sigma^a, \tau^b, p) \Big) \,.
\end{align*}}

\noindent
The following auxiliary predicates check the value of the node field $\pfield$
at $(\sigma,a)$.
{\allowdisplaybreaks
\begin{align*}
&\text{\small\textcolor{gray}{$\pfield$ is $\NIL$ at the frame $(\sigma,a)$ of the lace:}}\\
&\underline{\mathit{is\_pfield\_nil}}(\sigma, a, \pfield) \defeq \\
&\hspace{1cm}\bigvee_{i \in [2,a]} \Big( \sigma^i.\instr = \symb{ \pfield \coloneq \NIL } \wedge
   \bigwedge_{j \in [i+1,a]} \sigma^j.\instr \neq \symb{ \pfield \coloneq * } \Big) \\[1ex]
&\text{\small\textcolor{gray}{$\pfield$ has the same value as $r$ at the frame $(\sigma,a)$ of the lace:}}\\
&\underline{\mathit{is\_pfield\_ptr}}(\sigma, a, \pfield, r, i) \defeq \\
&\hspace{1cm}\sigma^i.\instr = \symb{ \pfield \coloneq r } \wedge
   \bigwedge_{j \in [i+1,a]} \sigma^j.\instr \neq \symb{ \pfield \coloneq * } \\[1ex]
&\text{\small\textcolor{gray}{$\pfield$ has never been assigned up to the frame 
$(\sigma,a)$ of the lace:}}\\
&\underline{\mathit{is\_pfield\_implicit}}(\sigma, a, \pfield) \defeq
   \bigwedge_{j \in [2,a]} \sigma^j.\instr \neq \symb{ \pfield \coloneq * } \,.
\end{align*}}

\subsection{Boolean Conditions and Control-Flow Instructions}

Control-flow statements $\mathbf{if}$ and $\mathbf{while}$ have two successors, 
depending on the value of their Boolean condition. To support those statements, we introduce
a 3-argument version of the predicate that advances the program counter:
\begin{align*}
\underline{\mathit{advance\_pc}}(f_1, f_2, \mathit{cond}) \defeq\,
\big(f_2.\pc = \nextpc(f_1.\pc, \mathit{cond}) \big).
\end{align*}
The next predicate handles the Boolean conditions of the form $p = q$:
\begin{align*}
\sidecomment{if $p = q$}
\underline{\mathit{step\_cmp\_ptr}}&(\sigma, a, \tau, b, p, q, \mathit{neg}) \defeq\: \\
&\text{\small\textcolor{gray}{Case 1: at least one pointer is $\NIL$}}   \\
&\Big( \, \big( \sigma^a.\isnil_p \vee \sigma^a.\isnil_q \big) \,\wedge\,
\big( \sigma^a.\nextfr = (\fwd, a+1) \big) \,\wedge\, \\
&\quad\mathit{advance\_pc}(\sigma^a, \tau^b, \mathrm{xor}(\mathit{neg}, \sigma^a.\isnil_p \leftrightarrow \sigma^a.\isnil_q) ) \, \wedge \\
&\quad \mathit{default}_{\avail,\activefield,\datafield,d,\isnil,\instr,\activechild}\big( \sigma^a, \tau^{b-1}, \tau^b \big) \,\Big) \quad\vee\\[1ex]
&\text{\small\textcolor{gray}{Case 2: start or keep rewinding}}  \\
&\mathit{rewind2}(\sigma, a, \tau, b, p, q) \quad\vee \\[1ex] 
&\text{\small\textcolor{gray}{Case 3: stop rewinding}}  \\
 &\Big(\, \mathit{stop\_rewind2}(\sigma, a, p, q)
 \,\wedge\, \big( \sigma^a.\nextfr = (\fwd, a+1) \big) \,\wedge \\
&\;\; \mathit{advance\_pc}\big( \sigma^a, \tau^b, \mathrm{xor}(\mathit{neg}, \mathit{are\_equal\_after\_rewind}(\sigma, a, p, q) ) \big) \, \wedge \\
&\;\; \mathit{default}_{\avail,\activefield,\datafield,d,\isnil,\instr,\activechild}(\sigma^a, \tau^{b-1}, \tau^b) 
    \Big) \,.
\end{align*}
The next predicate handles the Boolean conditions based on the content of the data variables.
Those conditions can always be resolved locally.
\begin{align*}
\sidecomment{if $r(\dataexp_1,\ldots,\dataexp_l)$}
&\underline{\mathit{step\_local\_branch}}(f_1, f_2, r, \dataexp_1, \ldots, \dataexp_l) \defeq\: \\
&\hspace{1cm}\big( f_1.\nextfr = (\fwd,*) \big) \,\wedge\, \\
&\hspace{1cm}\mathit{advance\_pc}\big(\, f_1, f_2, 
r(\dataexp_1[d \mapsto f_1.d]_{d \in \DV}, \ldots, 
  \dataexp_l[d \mapsto f_1.d]_{d \in \DV}) \,\big) \, \wedge \\
&\hspace{1cm}\mathit{default}_{\avail,\activefield,\datafield,d,\isnil,\instr,\activechild}(f_1, f_1, f_2) \,.
\end{align*}

\subsection{Rewinding the Lace} \label{app:rewind}

The following predicate is true when the pointer $q$ is not $\NIL$ at $(\sigma,a)$
and its value cannot be ascertained by analyzing the label $\sigma$ alone.
In that case, it is necessary to perform at least one rewinding step,
represented by the frame $(\tau,b)$.
\begin{align*}
\underline{\mathit{rewind}}&(\sigma, a, \tau, b, q) \defeq \\
&\sigma^a.\nextfr = (\dir, b) \,\wedge\, \neg\sigma^a.\isnil_q \,\wedge\, a' = \mathit{cur\_rewind\_pos}(\sigma, a) \,\wedge \\
&\neg\PointsHere(\sigma, a', q) \,\wedge\, 
a'' = \mathit{last\_upd}(\sigma,a',q) \,\wedge \\
&\sigma^{a''}.\prevfr = (\dir, b') \,\wedge\, 
\tau^{b'}.\nextfr = (\dir', a'') \,\wedge \\
&b = \mathit{last\_visit}(\tau,\dir',a') +1 \,\wedge\,
 \tau^b.\instr = \RWD_{b'}  \,\wedge\, \\
&\mathit{default}_{\avail,\activefield,\pc,\datafield,d, \isnil,\activechild}(\sigma^a, \tau^{b-1}, \tau^b) 
 \,.
\end{align*}
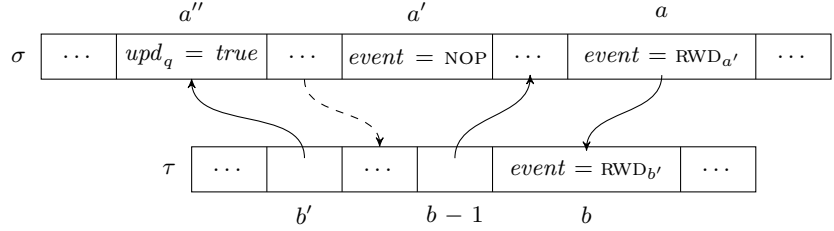
\begin{figure}
\begin{tikzpicture}[font=\footnotesize]
\begin{scope}[align=center]
    \draw (0,0.2) rectangle (10.5,0.8);
    \foreach \x in {1,3,4,6,7,9.5}
        \draw[-] (\x,0.2) -- (\x,0.8);
        
    \node[text width=2cm,font=\small] at (-0.3,0.5) {$\sigma$};
    \node[text width=1.5cm] at (0.5,0.5) {\dots};
    \node[text width=2.5cm] (a2) at (2,0.5) {$\chg_q = \true$};
    \node[text width=2cm] at (2,1.1) {$a''$};
    \node[text width=1.5cm] at (3.5,0.5) {\dots};
    \node[text width=2.5cm] at (5,0.5) {$\instr = \NOP$};
    \node[text width=2cm] at (5,1.1) {$a'$};
    \node[text width=1.5cm] (space) at (6.5,0.5) {\dots};
    \node[text width=2.3cm] (a0) at (8.25,0.5) {$\instr = \RWD_{a'}$};
    \node[text width=2cm] at (8.25,1.1) {$a$};
    \node[text width=1.5cm] at (10,0.5) {\dots};
    \coordinate[below=0.1cm of space] (below_space);
    \coordinate (sigma_space) at (3.5,0.2);
\end{scope}
\begin{scope}[align=center,
              xshift=2cm,yshift=-1.5cm]
    \draw (0,0.2) rectangle (7.5,0.8);
    \foreach \x in {1,2,3,4,6.5}
        \draw[-] (\x,0.2) -- (\x,0.8);
        
    \node[text width=2cm,font=\small] at (-0.3,0.5) {$\tau$};
    \node[text width=1.5cm] at (0.5,0.5) {\dots};
    \node[text width=2cm] (b1) at (1.5,-0.1) {$b'$};
    \node[text width=1.5cm] at (2.5,0.5) {\dots};
    \node[text width=2cm] (bminus) at (3.5,-0.1) {$b-1$};
    \node[text width=2.2cm] (b0) at (5.25,0.5) {$\instr = \RWD_{b'}$};
    \node[text width=2cm] at (5.25,-0.1) {$b$};
    \node[text width=1.5cm] at (7,0.5) {\dots};
    \coordinate[above=0.5cm of b1] (above_b1);
    \coordinate[above=0.5cm of bminus] (above_bminus);
    \coordinate (tau_space) at (2.5,0.8);
\end{scope}
\draw[->]  (above_b1) to[out=90, in=-90] (a2);
\draw[->]  (a0) to[out=-90, in=90] (b0);
\draw[->]  (above_bminus) to[out=90, in=-90] (below_space);
\draw[dashed,->] (sigma_space) to[out=-90, in=90] (tau_space);
\end{tikzpicture}
\caption{An example of the $\mathit{rewind}(\sigma,a,\tau,b,q)$ predicate.
Arrows connect a frame to its successor in the lace.}
\end{figure}

The above predicate uses several auxiliary functions and predicates.
The function $\mathit{cur\_rewind\_pos}(\sigma,a)$ 
returns the position within $\sigma$ of the current rewinding.
If the rewinding is just starting, this position will simply be $a$.
Otherwise, the current instruction will be of the type $\RWD_{b}$,
and the current rewinding position will be $b$.
$$
 \underline{\mathit{cur\_rewind\_pos}}(\sigma,a) \defeq 
 \begin{cases}
    b &\text{if } \sigma^a.\instr =\RWD_{b} \text{ for some } b \in [n], \\
    a &\text{otherwise.}
 \end{cases}
$$
The predicate $\PointsHere(\sigma, a, q)$ holds when $q$ points
to the node labeled with $\sigma$ when the lace is at $(\sigma,a)$.
This occurs when an instruction of the type $\symb{ q \coloneq \here }$
is found in $\sigma$ at position at most $a$, and all $\chg_q$ flags
from that position to position $a$ are false.
$$
\underline{\PointsHere}(\sigma, a, q) \defeq
  \bigvee_{i \in [2,a]} \Big( 
    \sigma^i.\instr = \symb{ q \coloneq \here } \,\wedge\,
    \bigwedge_{\mathclap{j \in [i+1,a]}} \neg\sigma^j.\chg_q \Big) \,.
$$
When the value of pointer $q$ cannot be ascertained by the current label $\sigma$,
we search for its most recent assignment by identifying the latest
position in $\sigma$ that comes before position $a$ and contains the flag $\chg_q = \true$.
This is the job of the function $\mathit{last\_upd}(\sigma, a, q)$.
If no position in $\sigma$ before $a$ contains $\chg_q = \true$,
this function returns 2, because in that case the rewinding operation must go back
to the frame that led to the first visit to the current node. 
{\small
$$
 \underline{\mathit{last\_upd}}(\sigma, a, q) \defeq\,
\begin{cases}
    \max \{ i \in [2,a] \mid \sigma^i.\chg_q = \true \} &\text{if } \{ i \in [2,a] \mid \sigma^i.\chg_q = \true \} \neq \emptyset \\
    2 &\text{otherwise.}
\end{cases} 
$$
}
The last auxiliary function, $\mathit{last\_visit}(\tau, \dir, j)$,
returns the position of the last frame from $\tau$ whose $\nextfr$
field is of the form $(\dir, c)$ for some $c \leq j$.
 \begin{align*}
  \underline{\mathit{last\_visit}}(\tau, \dir, j) &\defeq\,
 \max \{ i \in [n] \mid \tau^i.\nextfr = (\dir,c), c \leq j \}  \,.
\end{align*}

\subsubsection{Searching for two pointers.}

We describe a variant of the predicate 
$\mathit{rewind}$,
for the instructions that search for \emph{two}
different pointers at the same time.
In fact, this only happens in the encoding of the Boolean condition $q_1 = q_2$.
\begin{align*}
   \underline{\mathit{rewind2}}&(\sigma, a, \tau, b, q_1, q_2) \defeq \\
&\sigma^a.\nextfr = (\dir, b) \,\wedge\, \neg\sigma^a.\isnil_{q_1} \,\wedge\, \neg\sigma^a.\isnil_{q_2} \,\wedge\, \\
&a' = \mathit{cur\_rewind\_pos}(\sigma, a) \,\wedge \\
&\neg\PointsHere(\sigma, a', q_1) \,\wedge\, 
                \neg\PointsHere(\sigma, a', q_2) \, \wedge \\
&a'' = \max \big\{ \mathit{last\_upd}(\sigma,a',q_1), \mathit{last\_upd}(\sigma,a',q_2) \big\} \,\wedge \\
&\sigma^{a''}.\prevfr = (\dir, b') \,\wedge\, \tau^{b'}.\nextfr=(\dir', a'') \,\wedge \\
&b = \mathit{last\_visit}(\tau,\dir',a') +1 \,\wedge\,
 \tau^b.\instr = \RWD_{b'} \,\wedge\, \\
&\mathit{default}_{\avail,\activefield,\pc,\datafield,d, \isnil,\activechild}(\sigma^a, \tau^{b-1}, \tau^b) 
 \,. 
\end{align*}
Next, the version of rewinding used by the second phase of the 
statement $p \coloneq \field{q}{\pfield}$.
\begin{align*}
\underline{\mathit{rewind\_special}}&(\sigma, a, \tau, b, r, a') \defeq \\
&\sigma^a.\nextfr = (\dir, b) \,\wedge\, \neg\sigma^a.\isnil_{r} \,\wedge\,  \\
&\neg\PointsHere(\sigma, a', r) \,\wedge\, 
 a'' = \mathit{last\_upd}(\sigma, a', r) \,\wedge \\
&\sigma^{a''}.\prevfr=(\dir, b') \,\wedge\, \tau^{b'}.\nextfr=(\dir', a'') \,\wedge \\
&b = \mathit{last\_visit}(\tau,\dir', a') +1 \,\wedge\,
 \tau^b.\instr = \RWD_{b', r}  \,\wedge\, \\
&\mathit{default}_{\avail,\activefield,\pc,\datafield,d, \isnil,\activechild}(\sigma^a, \tau^{b-1}, \tau^b) 
 \,.
\end{align*}

The following predicates check whether the search for one or two pointers
is finished because those variables point to the current node (i.e., the node
with label $\sigma$):
\begin{align*}
\underline{\mathit{stop\_rewind}} (\sigma, a, q) \defeq\: 
& \neg \sigma^a.\isnil_q \,\wedge\, a' = \mathit{cur\_rewind\_pos}(\sigma, a) \,\wedge \\
&\PointsHere(\sigma, a', q) \\
\underline{\mathit{stop\_rewind2}}(\sigma, a, q_1, q_2) \defeq\:
&\mathit{stop\_rewind}(\sigma, a, q_1) \vee \mathit{stop\_rewind}(\sigma, a, q_2) .
 \end{align*}

The following predicate is used by the instructions that have already searched for two pointers
and want to check whether they are equal.
\begin{multline*}    
\underline{\mathit{are\_equal\_after\_rewind}}(\sigma, a, p, q) \defeq    \\
a' = \mathit{cur\_rewind\_pos}(\sigma, a) \,\wedge \,
\PointsHere(\sigma, a', p) \wedge \PointsHere(\sigma, a', q) \,.
\end{multline*}

\end{document}